\documentclass[english]{smfart}
\usepackage{amssymb}
\usepackage{amsmath}
\usepackage{smfthm}
\usepackage{graphicx}
\usepackage{multicol}
\usepackage[english]{babel}
\usepackage[utf8x]{inputenc}
\def\localversion{1}

\usepackage{enumerate}
\usepackage[pdftex,all]{xy}

\usepackage{amscd}
\usepackage{algorithm}
\usepackage{algpseudocode}

\usepackage{color}
\usepackage[breaklinks]{hyperref}

\theoremstyle{plain}
\newtheorem{theorem}{Theorem}
\newtheorem{proposition}[theorem]{Proposition}
\newtheorem{lemma}[theorem]{Lemma}
\newtheorem{corollary}[theorem]{Corollary}
\theoremstyle{definition}
\newtheorem{example}[theorem]{Example}
\newtheorem{remark}[theorem]{Remark}
\newtheorem{definition}[theorem]{Definition}

\newcommand{\vladut}{Vl\u{a}du\c{t}}

\newcommand{\Z}{\mathbb{Z}}
\newcommand{\N}{\mathbb{N}}

\newcommand{\eqdef}{\stackrel{\text{def}}{=}}

\newcommand{\linspan}{\mathbf{\operatorname{Span}}}

\renewcommand{\leq}{\leqslant} 
 
\renewcommand{\geq}{\geqslant}

\newcommand{\Fq}{\ensuremath{\mathbb{F}_q}}

\newcommand{\Fqm}{\ensuremath{\mathbb{F}_{q^m}}}
\newcommand{\F}{\ensuremath{\mathbb{F}}}

\newcommand{\Fqbar}{\overline{\F}_{q}}

\renewcommand{\P}{\mathbb{P}}

\newcommand{\code}[1][C]{\ensuremath{\mathcal{#1}}}

\newcommand{\AC}{\code[A]}
\newcommand{\BC}{\code[B]}
\newcommand{\CC}{{\code}}

\newcommand{\rest}[2]{{#1}_{|#2}}

\newcommand{\dmin}{d_{\textrm{min}}}

\newcommand{\deuxbracket}[1][2]{^{\langle#1\rangle}}
\newcommand{\deuxstar}[1][2]{^{\star#1}}
\let\deux\deuxstar
\newcommand{\stab}{\mathbf{Stab}}

\newcommand{\word}[1]{\ensuremath{\boldsymbol{#1}}}
\newcommand{\av}{\word{a}}
\newcommand{\bv}{\word{b}}
\newcommand{\cv}{\word{c}}
\newcommand{\dv}{\word{d}}
\newcommand{\ev}{\word{e}}

\newcommand{\mv}{\word{m}}

\newcommand{\uv}{\word{u}}
\newcommand{\vv}{\word{v}}

\newcommand{\xv}{{\word{x}}}
\newcommand{\yv}{{\word{y}}}
\newcommand{\zv}{{\word{z}}}

\newcommand{\wt}[1]{w_H(#1)}

\newcommand{\eucl}[2]{{\left\langle #1, #2 \right\rangle}_{\textrm{Eucl}}}

\newcommand{\card}[1]{\ensuremath{|#1|}}

\newcommand{\mat}[1]{\ensuremath{\boldsymbol{#1}}}

\newcommand{\Gm}{\mat{G}}
\newcommand{\Hm}{\mat{H}}

\newcommand{\GRS}[3]{\text{\bf GRS}_{#1}\left(#2,#3\right)}
\newcommand{\RS}[2]{\text{\bf RS}_{#1}\left(#2\right)}

\newcommand{\CL}[3]{{\code}_L \left(#1, #2, #3\right)}
\newcommand{\COm}[3]{{\code}_{\Omega} \left(#1, #2, #3\right)}

\newcommand{\dgop}{d_{\textrm{Gop}}}
\newcommand{\dorder}{d_{\textrm{Ord}}}

\DeclareMathOperator{\supp}{Supp}
\renewcommand{\div}{\operatorname{div}}
\DeclareMathOperator{\res}{res}
\DeclareMathOperator{\Div}{Div}

\DeclareMathOperator{\Cl}{Cl}
\DeclareMathOperator{\Pic}{Pic}

\newcommand{\cP}{\mathcal{P}}

\newcommand{\map}[5][]{
\begin{array}{cccc}
#1 & #2 & \longrightarrow & #3 \\
& #4 & \longmapsto     & #5
        \end{array}
}

\DeclareMathOperator{\Aut}{Aut}

\title{Algebraic geometry codes and some applications}
\author{Alain Couvreur}
\address{Inria \& LIX, \'Ecole polytechnique, 1 rue Honor\'e d'Estienne d'Orves,
91120 Palaiseau Cedex}
\email{alain.couvreur@inria.fr}

\author{Hugues Randriambololona}
\address{ANNSI, Laboratoire de cryptographie \& T\'el\'ecom Paris,
Tour Mercure, 31 quai de Grenelle, 75015 Paris}
\email{hugues.randriam@ssi.gouv.fr}

\begin{document}
\maketitle
\begin{abstract}
  This article surveys the development of the theory of algebraic
  geometry codes since their discovery in the late 70's. We summarize
  the major results on various problems such as: asymptotic parameters,
  improved estimates on the minimum distance, and decoding
  algorithms.  In addition, we present various modern applications of
  these codes such as public-key cryptography, algebraic complexity
  theory, multiparty computation or distributed storage.
\end{abstract}
 
\setcounter{tocdepth}{1}
{\small \tableofcontents}

\section*{Introduction}\label{intro}
Algebraic geometry codes is a fascinating topic at the confluence of
number theory and algebraic geometry on one side and computer science
involving coding theory, combinatorics, information theory and
algorithms, on the other side.

\subsection*{History} The beginning of the story dates back
to the early 70's where the Russian mathematician V.~D.~Goppa proposed
code constructions from algebraic curves, using rational functions or
differential forms on algebraic curves over finite fields
\cite{goppa1977ppi,goppa1981dansssr,goppa1982iansssr}. In particular,
this geometric point of view permitted to regard Goppa's original
construction of codes from rational functions
\cite{goppa1970ppi,goppa1971ppi} (see \cite{berlekamp1973it} for a
description of these codes in English) as codes from
differential forms on the projective line.

Shortly after, appeared one of the most striking result in history of
coding theory, which is probably at the origin of the remarkable
success of algebraic geometry codes. In 1982, Tsfasman, \vladut{}
and Zink \cite{tsfasman1982mn} related the existence of a sequence of
curves whose numbers of rational points go to infinity and grow linearly with
respect to the curves' genera to the existence of sequences of
asymptotically good codes. Next, using sequences of modular curves and
Shimura curves they proved the existence of sequences of codes over a
field $\F_{q}$ where $q = p^2$ or ${p^4}$, for $p$ a prime number,
whose asymptotic rate $R$ and asymptotic relative distance $\delta$
satisfy
\begin{equation}\label{eq:tvzbound}
  R \geq 1 - \delta - \frac{1}{\sqrt q -1}\cdot
\end{equation}
In an independent work and using comparable arguments, Ihara
\cite{ihara1981jfsuTokyo} proves a similar result over any field $\Fq$
where $q$ is a square. For this reason, the ratio
\[
  \limsup_{g \rightarrow +\infty} \frac{\max |X(\Fq)|}{g},
\]
where the $\max$ is taken over the set of curves $X$ of genus $g$ over
$\Fq$ is usually referred to as the {\em Ihara constant} and denoted
by $A(q)$.  Further, \vladut{} and Drinfeld \cite{vladut1983faa}
proved an upper bound for the Ihara constant showing that the families
of curves exhibited in \cite{tsfasman1982mn,ihara1981jfsuTokyo} are
optimal.

The groundbreaking aspect of such a result appears when
$q$ is a square and $q \geq 49$,
since, for such a parameter, Tsfasman--\vladut{}--Zink bound is
better than the asymptotic {\em Gilbert--Varshamov bound}.  Roughly
speaking, this result asserts that some algebraic geometry codes are
better than random codes, while the opposite statement was commonly
believed in the community. 

For this breakthrough, Tsfasman, \vladut{} and Zink received the
prestigious {\em Information Theory Society Paper Award} and their
result motivated an intense development of the theory of algebraic
geometry codes. The community explored various sides of this theory in
the following decades. First, the question of producing sequences of
curves with maximal Ihara constant or the estimate of the Ihara
constant $A(q)$ when $q$ is not a square became a challenging new
problem in number theory. In particular, some constructions based on
class field theory gave lower bounds for $A(q)$ when $q$ is no longer
a square. In addition, in 1995, Garcia and Stichtenoth \cite{GS1995}
obtained new optimal sequences of curves (i.e. reaching Drinfeld
\vladut{} bound) using a much more elementary construction called {\em
  recursive towers}. Beside the asymptotic questions, many works
consisted in improving in some specific cases Goppa's estimate for the
minimum distance of algebraic geometry codes. Such results permitted
to construct new codes of given length whose parameters beat the
tables of best known codes (see for instance
\cite{grassl:codetables}). Third, another fruitful direction is on the
algorithmic side with the development of polynomial time decoding
algorithms correcting up to half the designed distance and even further
using list decoding.  This \ifx\localversion\undefined {chapter}
\else {article}
\fi
presents known results on
improved bounds on the minimum distance and discusses unique and list
decoding. The asymptotic performances of algebraic geometry codes are
also quickly surveyed without providing an in--depth study of the
Ihara constant and the construction of optimal towers. The latter
topic being much too rich would require a separate treatment which we
decided not to develop it in the present survey.

It should also be noted that codes may be constructed from higher
dimensional varieties. This subject, also of deep interest will not be
discussed in the present \ifx\localversion\undefined {chapter}
\else {article}
\fi
. We refer the interested reader to \cite{little2008chapter} for a
survey on this question, to \cite{sullivan1998decoding} for a decoding
algorithm and to \cite{couvreur2020cm} for a first attempt toward good
asymptotic constructions of codes from surfaces.

\subsection*{Applications of algebraic geometry codes}
Algebraic geometry codes admit many interesting properties that make
them suitable for a very wide range of applications.  Most of these
properties are inherited from Reed-Solomon (RS) codes and their
variants\footnote{Beware that the terminology on Reed--Solomon codes
  varies in the literature with several names for variants: {\em
    generalized Reed--Solomon codes}, {\em extended Reed--Solomon
    codes}, {\em doubly extended Reed--Solomon codes}, etc. In this
  \ifx\localversion\undefined {chapter}
  \else {article}
  \fi
  , we refer to {\em Reed--Solomon} or {\em generalized Reed--Solomon
    codes} as the algebraic geometry codes from a curve of genus
  $0$. See \S~\ref{subsec:GRS_codes} for further details.}, of which
AG codes are a natural extension:
\begin{itemize}
\item AG codes can be explicitly constructed.
\item AG codes can be efficiently decoded.
\item AG codes admit good bounds on their parameters: although they
  might not be MDS, they remain close to the Singleton bound
  (\cite[Th.~1.11]{sloane1977book}).
\item AG codes behave well under duality: the dual of an AG code is an AG code.
\item AG codes behave well under multiplication: the $\star$-product
  of two AG codes is included in, and in many situations is equal to,
  an AG code.
\item AG codes may have automorphisms, reflecting the geometry of the
  underlying objects.
\end{itemize}
However, as already discussed above, AG codes enjoy an additional
property over their Reed-Solomon genus $0$ counterparts, which was perhaps the main
motivation for their introduction:
\begin{itemize}
\item For a given $q$, the length of an (extended) RS code over $\F_q$ cannot
  exceed $q+1$ while one can construct arbitrarily long AG codes
  over a given fixed field $\Fq$.
\end{itemize}
It would be an endless task to list all applications of AG codes.
Below we focus on a selection of those we find most meaningful.
Basically, in every situation where Reed--Solomon codes
are used, replacing them by algebraic geometry codes is natural and
frequently leads to improvements. These application topics may be
symmetric cryptography, public-key cryptography, algebraic complexity
theory, multiparty computation and secret sharing, distributed
storage and so on.
The present survey aims at presenting various aspects of the theory
of algebraic geometry codes together with several applications.

\subsection*{Organization of the
  \ifx\localversion\undefined
  {chapter}
  \else
  {article}
  \fi}
We start by fixing some general notation in Section~\ref{sec:notation}.
The background on algebraic geometry and number theory is recalled in
Section~\ref{sec:curves} and the construction and first properties of
algebraic geometry codes are recalled in Section~\ref{sec:basics}.  In
particular, Goppa bound for the minimum distance is recalled.  In
Section~\ref{sec:asymptotics}, we discuss asymptotic performances of
algebraic geometry codes and relate this question to the Ihara
constant. As said earlier, the construction and study of good families
of curves is too rich to be developed in the present survey and would
require a separate treatment. Next, Section~\ref{sec:dmin} is devoted
to various improvements of Goppa designed distance.
Section~\ref{sec:decoding} surveys the various decoding algorithms.
Starting with algorithms correcting up to half
the designed distance and then moving to the more recent developments
of list decoding permitting to exceed this threshold.  Finally, the
three last sections present various applications of algebraic geometry
codes. Namely, we study their possible use for post--quantum public key
cryptography in Section~\ref{sec:McEliece}. Thanks to their nice
behavior with respect to the so--called {\em $\star$--product},
algebraic geometry codes have applications to algebraic complexity
theory, secret sharing and multiparty computation, which are presented in
Section~\ref{sec:MPC}. Finally, applications to distributed storage
with the algebraic geometric constructions of {\em locally recoverable
  codes} are presented in Section~\ref{sec:LRC}.
 \section{Notation}\label{sec:notation}
For any prime power $q$, the finite field with $q$ elements is denoted
by $\Fq$ and its algebraic closure by $\Fqbar$. Given any ring
$R$, the group of invertible elements of $R$ is denoted by $R^\times$.
In particular, for a field $\F$, the group $\F^\times$
is nothing but $\F \setminus \{0\}$.

Unless otherwise specified, any code in this
\ifx\localversion\undefined chapter \else article \fi is linear.  The
vector space $\Fq^n$ is equipped with the {\em Hamming weight} denoted
by $w_H(\cdot)$ and the Hamming distance between two vectors
$\xv, \yv$ is denoted by $d_H(\xv, \yv)$.  Given a linear code
$\CC \subseteq \Fq^n$, as usually in the literature, the fundamental
parameters of $\CC$ are listed as a triple of the form $[n, k, d]$,
where $n$ denotes its block length, $k$ its dimension as an
$\Fq$--space and $d$ its minimum distance, which is sometimes also
referred to as $d(\CC)$. When the minimum distance is unknown, we
sometimes denote by $[n,k]$ the known parameters.

Another important notion in the sequel is that of {\em
  $\star$-product}, which is nothing but the component wise
multiplication in $\Fq^n$: for $\word{x}=(x_1,\dots,x_n)$ and
$\word{y}=(y_1,\dots,y_n)$ we
have \[\word{x}\star\word{y}\eqdef(x_1y_1,\dots,x_ny_n).\] This notion
extends to codes: given two linear codes $\code,\code'\subseteq\F_q^n$
we let
\[\code\star\code'
\eqdef\linspan_{\F_q}\{\word{c}\star\word{c'} ~|~ \word{c}\in\code,\,\word{c'}\in\code'\}\]
be the \emph{linear span} of the pairwise products of codewords from
$\code$ and $\code'$.  Observe that $\code\star\code'\subseteq\F_q^n$
is again a linear code, since we take the linear span.  Then the
\emph{square} of $\code$ is defined as
\[\code\deuxbracket\eqdef \code\deuxstar=\code\star\code.\]
The two notations $\code\deuxbracket$ and $\code\deuxstar$ are
equivalent and depend only on the authors. In this text we use the
notation $\code\deux$.

Recall that $\Fq^n$ is equipped with a canonical Euclidean
bilinear form defined as
\[ 
  \map[\eucl{\cdot}{\cdot}~:]{\Fq^n \times \Fq^n}{\Fq}{(\xv,
    \yv)}{\sum_{i=1}^n x_i y_i.}
\]
The $\star$--product and the Euclidean product are related by the following adjunction
property

\begin{lemma}\label{lem:adjunction}
  Let $\xv, \yv, \zv \in \Fq^n$, then
  $
    \eucl{\xv \star \yv}{\zv} = \eucl{\xv}{\yv \star \zv}\ \!.
  $
\end{lemma}
 \section{Curves and function fields}\label{sec:curves}
Algebraic geometry, the study of geometric objects defined by
polynomial equations, has a long history.  In the last century it
received rigorous foundations, in several waves, each bringing its own
language.  For \emph{most} applications to coding theory, we will only
need to work with some of the simplest geometric objects, namely
curves.  These can be described equally well in the following
languages:
\begin{itemize}
\item the language of \emph{algebraic function fields}, for which a
  recommended reference is \cite{stichtenoth2009book};
\item the language of \emph{varieties over an algebraically closed
    field}, as in \cite{fulton1989book} or \cite{shafarevich1994book};
\item the language of \emph{schemes}, for which we refer to
  \cite{hartshorne1977book}.
\end{itemize}

As long as one works with one fixed curve, these languages have
the same power of expression.

We briefly recall some of the basic notions and results that we
need, and explain how they correspond in these different languages.
For more details the reader should look in the references given above.

\subsection{Curves, points, function fields and places}

\begin{definition}
\label{def_function_field}
An algebraic function field $F$ with constant field $\F_q$ is a finite
extension of a purely transcendental extension of $\F_q$ of
transcendence degree~$1$, in which $\F_q$ is {\em algebraically closed},
i.e. any element $\alpha \in F$ which is algebraic over $\Fq$ is actually
in $\Fq$.
\end{definition}

Any such $F$ is of the form
$F=\operatorname{Frac}(\F_q[x,y]/(P(x,y)))$ where $P \in \Fq[x,y]$ is
{\em absolutely irreducible}, i.e. irreducible even when regarded as
an element of $\Fqbar[x,y]$.

\begin{definition}
\label{def_curve}
A {\em curve} over $\F_q$ is a geometrically irreducible smooth projective
variety of dimension~$1$ defined over $\F_q$.
\end{definition}

Any such curve can be obtained as the projective completion and
desingularization of an affine plane curve of the form $\{P(x,y)=0\}$
where $P$ is an absolutely irreducible polynomial in the two
indeterminates $x$ and $y$ over $\F_q$.

From this observation we see that Definitions~\ref{def_function_field}
and~\ref{def_curve} are essentially equivalent.  This can be made more
precise:
\begin{theorem}
There is an equivalence of categories between:
\begin{itemize}
\item algebraic function fields, with field morphisms, over $\F_q$;
\item curves, with dominant (surjective) morphisms, over $\F_q$.
\end{itemize}
\end{theorem}
The proof can be found e.g. in \cite[\S~I.6]{hartshorne1977book}.  In
one direction, to each curve $X$ one associates its field of rational
functions $F=\F_q(X)$, which is an algebraic function field.  We then
have a natural correspondence between:
\begin{itemize}
\item places, or discrete valuations of $F$;
\item Galois orbits in the set $X(\overline{\F}_q)$ of points of $X$
  with coordinates in the algebraic closure of $\F_q$;
\item closed points of the topological space of $X$ seen as a scheme.
\end{itemize}

If $P$ is a closed point of $X$, or a place of $F$, we denote by
$v_P:F\to\Z\cup\{\infty\}$ the corresponding discrete valuation, by
$\mathcal{O}_P=\{f\in F ~|~ v_P(f)\geq0\}$ its valuation ring (the local
ring of $X$ at $P$), and by $\mathfrak{m}_P=\{f\in F ~|~ v_P(f)>0\}$ its
maximal ideal.  An element $t_P\in\mathfrak{m}_P$ is called a
\emph{local parameter}, or a \emph{uniformizer} at $P$, if it
satisfies $v_P(t_P)=1$.
The residue field $k_P=\mathcal{O}_P/\mathfrak{m}_P$ is a finite
extension of $\F_q$.  We define the {\em degree} of $P$ as
the degree of the field extension:
\[
  \deg(P)=[k_P:\F_q].
\]
This degree is equal to the cardinality of the Galois orbit
corresponding to $P$ in $X(\overline{\F}_q)$.  Conversely, for any
finite extension $\F_{q^d}$ of $\F_q$, the set $X(\F_{q^d})$ of
$\F_{q^d}$-rational points of $X$ (i.e. points of $X$ with coordinates
in $\F_{q^d}$) identifies with the set of degree $1$ places in the
base field extension $\F_{q^d}F$ seen as a function field over
$\F_{q^d}$.

\subsection{Divisors}
The \emph{divisor group} $\Div(X)$ is the free abelian group generated
by the set of closed points of $X$ or equivalently by the set of
places of its function field $\Fq(X)$.  Thus, a divisor is a
formal sum
\[D=\sum_P n_P P\]
where $P$ ranges over closed points and $v_P(D) \eqdef n_P \in\Z$ are almost
all zero.  The {\em support} of $D$ is the finite set $\supp(D)$ of
such $P$ with $n_P\neq0$.  The degree of $D$ is
\[\deg(D)\eqdef\sum_P n_P\deg(P).\]
We say $D$ is {\em effective} if $n_P\geq0$ for all $P$. We write
$D_1\geq D_2$ if $D_1-D_2$ is effective.

A divisor is \emph{principal} if it is of the form
\[
  \div(f)\eqdef\sum_Pv_P(f)P
\]
for $f\in F^\times = F\setminus \{0\}$.
So we can write $\div(f)=(f)_0-(f)_\infty$
where $(f)_0=\sum_{v_P(f)>0}v_P(f)P$ is the divisor of zeros of $f$,
and $(f)_\infty=\sum_{v_P(f)<0}-v_P(f)P$ is its divisor of poles.
Principal divisors have degree zero (``a
rational function on a curve has as many poles as zeros''), and they
form a subgroup of $\Div(X)$.  We say that two divisors $D_1$ and
$D_2$ are {\em linearly equivalent}, and we write
\begin{equation}\label{eq:equivalence_lineaire}
  D_1\sim D_2,
\end{equation}
when $D_1-D_2$ is principal.  Passing to the quotient we get the
\emph{divisor class group} $\Cl(X)=\Div(X)/\sim$ of $X$, together with
a degree map $\Cl(X)\to\Z$ that is a surjective group
morphism~\cite[Th.~3.2(i)]{moreno1990book}.

The \emph{Riemann-Roch space} of a divisor $D$ is the vector space
\[L(D)=\{f\in F^\times ~|~ \div(f)\geq -D\}\cup\{0\}.\]
It has finite dimension $\ell(D)=\dim_{\F_q}L(D)$.  Actually $\ell(D)$
only depends on the linear equivalence class of $D$ in $\Cl(X)$.

To a divisor $D=\sum_P n_P P$, one can associate an invertible sheaf
(or line bundle) $\mathcal{O}(D)$ on $X$, generated locally at each
$P$ by $t_P^{-n_P}$, for $t_P$ a uniformizer at $P$.  There is then a
natural identification
\[L(D)=\Gamma(X,\mathcal{O}(D))\]
between the Riemann-Roch space of $D$ and the space of global sections
of $\mathcal{O}(D)$.  Conversely, given an invertible sheaf
$\mathcal{L}$ on $X$, any choice of a nonzero rational section
$s$ of $\mathcal{L}$ defines a divisor $D=\div(s)$ with
$\mathcal{L}\simeq\mathcal{O}(D)$. Another choice of $s$ gives a
linearly equivalent $D$.  From this we get an isomorphism
\[
  \Cl(X)\simeq\Pic(X)
\]
where $\Pic(X)$, the \emph{Picard group} of $X$, is the group of
isomorphism classes of invertible sheaves on $X$ equipped with the
tensor product.

\subsection{Morphisms of curves and
  pullbacks}\label{subsec:morphisms}
  A {\em morphism of curves} is a map
$\phi: X \rightarrow Y$ that is component wise described by polynomials
or rational functions. To such a map is associated a function field extension
$\phi^* : \Fq(Y) \hookrightarrow \Fq(X)$: given a rational function $f$ on $Y$,
one defines the {\em pullback} of $f$ by $\phi$ denoted $\phi^* f$ to be
the function $f \circ \phi$ on $X$.
The {\em degree} of $\phi$ is the extension degree $[\Fq(X) : \Fq(Y)]$
induced by the $\phi^*$ field extension.

\begin{definition}[Pullback of a divisor]\label{def:pullback}
Given a divisor $D = \sum_{i=1}^r n_i P_i$ on $Y$, one defines the {\em
  pullback of $D$ by $\phi$} and denotes it by $\phi^* D$:
\[
  \phi^* D \eqdef \sum_{i=1}^r \sum_{Q \stackrel{\phi}{\rightarrow} P_i} n_i \cdot e_{Q|P_i} \cdot Q,
\]
where $e_{Q|P_i}$ denotes the {\em ramification index at $Q$} (see
\cite[Def.~3.1.5]{stichtenoth2009book}).
\end{definition}

The divisor $\phi^* D$ is sometimes also called the {\em conorm} of
$D$ (\cite[Def.~3.1.8]{stichtenoth2009book}). In addition, it is
well--known that
\begin{equation}\label{eq:pullback_div}
  \deg \phi^* D = \deg \phi \cdot \deg D.
\end{equation}

\subsection{Differential forms}
The space of \emph{rational differential forms} on $X$ is the
one-dimensional $F$-vector space $\Omega_F$ whose elements are of the
form \[\omega=u\mathrm{d}v\] for $u,v\in F$, subject to the usual
Leibniz rule $\mathrm{d}(u_1u_2)=u_1\mathrm{d}u_2+u_2\mathrm{d}u_1$.
Given $\omega\in\Omega_F$ and $t_P$ a uniformizer at $P$, we can write
$\omega=f_P\mathrm{d}t_P$ for some $f_P\in F$ and we define the
valuation of $\omega$ at $P$ as
\[
  v_P(\omega)=v_P(f_P).
\]
One can prove that the definition does not depend on the choice of $t_P$.
In the same spirit as functions, to any nonzero rational differential form
$\omega$, one associates its divisor
\[
  \div (\omega) \eqdef \sum_{P} v_P(\omega) P.
\]
Equivalently, $\Omega_F$ is the space of rational sections of the
invertible sheaf $\Omega^1_X$, called the {\em canonical sheaf}, or
the sheaf of differentials of $X$, generated locally at each $P$ by
the differential $\mathrm{d}t_P$, for $t_P$ a uniformizer at $P$.

For any divisor $D$ we set
\[
  \Omega(D) \eqdef \Gamma(X, \Omega^1_X \otimes \mathcal O(-D)) =
  \{\omega \in \Omega_F \setminus \{0\} ~|~ \div(\omega) \geq D\} \cup \{0\}.
\]

\begin{remark}
  Beware of the sign change compared to the definition of the $L(D)$
  space.  This choice of notation could seem unnatural, but it is
  somehow standard in the literature on algebraic geometry codes (see
  e.g. \cite{stichtenoth2009book}), and is related to Serre duality.
\end{remark}

\subsubsection{Canonical divisors}
A \emph{canonical divisor} is a divisor $K_X$ on $X$ such that
$\Omega^1_X\simeq\mathcal{O}(K_X)$.  Thus a canonical divisor is of
the form
\[
  K_X=\div(\omega)
\]
for any choice of
$\omega=u\mathrm{d}v\in \Omega_F \setminus \{0\}$.  More explicitly, we have
$K_X=\sum_Pv_P(f_P)P$ where locally at each $P$ we write
$\omega=f_P\mathrm{d}t_P$. Any two canonical divisors are linearly
equivalent (see for instance \cite[Prop.~1.5.13(b)]{stichtenoth2009book}).

\subsubsection{Residues}
Given a rational differential $\omega_P$ at $P$ and a uniformizer $t_P$
we have a local Laurent series expansion
\[\omega_P=a_{-N}t_P^{-N}\mathrm{d}t_P+\cdots+a_{-1}t_P^{-1}\mathrm{d}t_P+\eta_P\]
where $N$ is the order of the pole of $\omega_P$ at $P$ and $\eta_P$ is
regular at $P$, i.e. $v_P(\eta_P)\geq 0$.
Then \[\res_P(\omega_P)=a_{-1}\] is independent of the choice of $t_P$
and is called the \emph{residue} of $\omega_P$ at $P$.  In particular,
if $\omega_P$ is regular at $P$, then we have $\res_P(\omega_P)=0$.
We refer to \cite[Chap.~IV]{stichtenoth2009book} for further details.

\subsection{Genus and Riemann--Roch theorem}
An important numerical invariant of a curve $X$ is its \emph{genus}
\[g=\ell(K_X).\]
We then also have $\deg(K_X)=2g-2$ (\cite[Cor.~1.5.16]{stichtenoth2009book}).

The following result, which is for instance proved in
\cite[Th.~1.5.15]{stichtenoth2009book}, is a central result in the theory of
curves.
\begin{theorem}[Riemann-Roch]
For a divisor $D$ on $X$, we have
\[\ell(D)-\ell(K_X-D)=\deg(D)+1-g.\]
\end{theorem}
In particular, we always have
\[
  \ell (D) \geq \deg D + 1 - g.
\]
In addition, Riemann-Roch spaces satisfy the following properties
(see e.g. Cor.~1.4.12(b) and Th.~1.5.17 of \cite{stichtenoth2009book}).

\begin{proposition}\label{prop:RR}
  Let $D$ be a divisor on a curve $X$ such that $\deg (D) < 0$. Then,
  $L(D)=\{0\}$.
\end{proposition}

\begin{corollary}\label{cor:RR}
When $\deg(D)>2g-2$ we have $\ell(K_X-D)=0$, and
then \[\ell(D)=\deg(D)+1-g.\] 
\end{corollary}

 \section{Basics on algebraic geometry codes}\label{sec:basics}
\subsection{Algebraic geometry codes, definitions and elementary
  results}
Let $X$ be a curve over $\Fq$ and fix a divisor $G$ and an ordered
sequence of $n$ distinct rational points $\cP=(P_1,\dots,P_n)$
disjoint from $\supp(G)$. The latter sequence $\cP$ will be referred
to as the {\em evaluation points sequence} and a divisor
is associated to it, namely:
\begin{equation}\label{eq:DP}
  D_{\cP} \eqdef P_1 +\cdots + P_n.
\end{equation}
With this data, we
can define two codes:
\begin{definition}
\label{defCL}
The {\em evaluation} code, or {\em function code}, $\CL{X}{\cP}{G}$ is the
image of the map
\begin{equation*}
  \map{L(G)}{\F_q^n}{f}{(f(P_1),\dots,f(P_n)).}
\end{equation*}
\end{definition}
\begin{definition}
\label{defCOm}
The {\em residue} code, or {\em differential code}, $\COm{X}{\cP}{G}$ is the
image of the map
\begin{equation*}
\map{\Omega(G-D_{\cP})}{\F_q^n}{\omega}{    (\res_{P_1}(\omega),\dots,\res_{P_n}(\omega))}
\end{equation*}
where $D_{\cP}$ is defined in \eqref{eq:DP}.
\end{definition}

The two constructions are dual to each other (see
for instance \cite[Th.~2.2.8]{stichtenoth2009book}):
\begin{theorem}\label{thm:duality_CL_COm}
The two codes defined just above are dual of each other:
\[\COm{X}{\cP}{G}=\CL{X}{\cP}{G}^\perp.\]
\end{theorem}

Recall that $\star$ denotes the component wise multiplication in
$\F_q^n$, so if $\mathbf{x}=(x_1,\dots,x_n)$ and
$\mathbf{y}=(y_1,\dots,y_n)$, then
$\mathbf{x}\star\mathbf{y}=(x_1y_1,\dots,x_ny_n)$.
\begin{definition}
Two codes $\code_1,\code_2\subseteq\F_q^n$ are diagonally equivalent
(under $\mathbf{a}\in(\F_q^\times)^n$) if
\[\code_2=\code_1\star\mathbf{a}\]
or equivalently if
\[\mathbf{G_2}=\mathbf{G_1}\mathbf{D_a}\]
where  $\mathbf{G_1},\mathbf{G_2}$ are generator matrices
of $\code_1,\code_2$ respectively, and $\mathbf{D_a}$
is the $n\times n$ diagonal matrix whose diagonal entries are
the entries of $\mathbf{a}$.
\end{definition}

\begin{remark}
  Diagonally equivalent codes are isometric with respect to the Hamming
  distance.
\end{remark}

\begin{lemma}
\label{lineq=diageq}
Let $G_1\sim G_2$ be two linearly equivalent divisors on $X$,
both with support disjoint from $\cP=\{P_1,\dots,P_n\}$.
Then $\CL{X}{\cP}{G_1}$ and $\CL{X}{\cP}{G_2}$
are diagonally equivalent under $\mathbf{a}=(h(P_1),\dots,h(P_n))$
where $h$ is any choice of function with $\div(h)=G_1-G_2$.

Likewise $\COm{X}{\cP}{G_1}$ and $\COm{X}{\cP}{G_2}$
are diagonally equivalent under $\mathbf{a}^{-1}$.
\end{lemma}

See \cite[Prop.~2.2.14]{stichtenoth2009book} for a proof
of the latter statement.

\begin{remark}\label{rem:homotecy}
  In the previous statement, the choice of $h$ is not unique. More
  precisely, $h$ may be replaced by any nonzero scalar multiple
  $\lambda h$ of $h$. This would replace $\av$ by $\lambda \av$, which
  has no consequence, since the codes are linear and hence globally
  invariant by a scalar multiplication.
\end{remark}

\begin{remark}
\label{UseWeakApprox}
If we accept codes defined only up to diagonal equivalence, then we
can relax the condition that $\supp(G)$ is disjoint from $\cP$ in
Definitions~\ref{defCL} and~\ref{defCOm}. Indeed, if $\supp(G)$ is not
disjoint from $\cP$, then by the weak approximation theorem
\cite[Th.~1.3.1]{stichtenoth2009book} we can find $G'\sim G$ with
support disjoint from $\cP$, and then we can use $\CL{X}{\cP}{G'}$ in
place of $\CL{X}{\cP}{G}$, and $\COm{X}{\cP}{G'}$ in place of
$\COm{X}{\cP}{G}$.  Lemma~\ref{lineq=diageq} then shows that up to
diagonal equivalence, these codes do not depend on the choice of $G'$.

In summary, the usual restriction ``the support of $G$ should avoid
the $P_i$'s'' in the definition of $\CL{X}{\cP}{G}$ can always be ruled out
at the cost of some technical clarifications.
\end{remark}

A slightly more general construction is the following:
\begin{definition}
  \label{generalAGcode}
  Given a curve $X$ over $\Fq$, an invertible sheaf $\mathcal{L}$ on
  $X$, and an ordered sequence of $n$ distinct rational points
  $\cP=(P_1,\dots,P_n)$.  After some choice of a trivialisation
  $\mathcal{L}|_{P_i}\simeq\F_q$ for the fibres of $\mathcal{L}$ at
  the $P_i\in\cP$, the code $\code(X,\cP,\mathcal{L})$ is the image of
  the map
\begin{equation*}
  \map{\Gamma(X,\mathcal{L})}{
    \bigoplus_{1\leq i\leq n}\mathcal{L}|_{P_i}\simeq\F_q^n
  }{s}{(s|_{P_1},\dots,s|_{P_n}).}
\end{equation*}
\end{definition}
Choosing another trivialisation of the fibres, and also replacing the
invertible sheaf $\mathcal{L}$ with an isomorphic one, leaves
$\code(X,\cP,\mathcal{L})$ unchanged up to diagonal equivalence.

Definition~\ref{defCL} is a special case of this construction with
$\mathcal{L}=\mathcal{O}(G)$ together with the natural trivialisation
$\mathcal{O}(G)|_{P_i}=\F_q$ when $P_i\not\in\supp(G)$.  Relaxing this
last condition, we can use a trivialisation
$\mathcal{O}(G)|_{P_i}=\F_q\cdot h_i|_{P_i}$ depending on the choice
of a local function $h_i$ at $P_i$ with minimal valuation
$v_{P_i}(h_i)=-v_{P_i}(G)$. This defines $\CL{X}{\cP}{G}$ as the image
of the map
\begin{equation*}
\map{L(G)}{\F_q^n}{f}{((f/h_1)(P_1),\dots,(f/h_n)(P_n)).}
\end{equation*}
A possible choice for $h_i$ is $h_i=t_{P_i}^{-v_{P_i}(G)}$ where
$t_{P_i}$ is a uniformizer.  Alternatively, given $G'\sim G$ with
$\supp(G')\cap\cP=\emptyset$, one can find $h$ with $\div(h)=G'-G$ and
set $h_i=h$ for all $i$. Doing so, we obtain
Remark~\ref{UseWeakApprox}.

Likewise Definition~\ref{defCOm} is a special case of
Definition~\ref{generalAGcode} with
$\mathcal{L}=\Omega^1_X\otimes \mathcal O(D_{\cP}-G)$ and
trivialisation given by the residue map when
$\supp(G)\cap\cP=\emptyset$, and can be relaxed in a similar way when
this condition is relaxed. This also gives:
\begin{lemma}
\label{dualityCLCOm}
For any canonical divisor $K_X$ on $X$, the codes $\COm{X}{\cP}{G}$
and $\CL{X}{\cP}{K_X+D_{\cP}-G}$ are diagonally equivalent.

Actually, if $\supp(G)$ is disjoint from $\cP$, then there is a choice
of a canonical divisor $K_X$, of support disjoint from $\supp(G)$ and
$\cP$, that turns this diagonal equivalence into an equality:
$\COm{X}{\cP}{G}=\CL{X}{\cP}{K_X-D_{\cP}+G}$.
\end{lemma}

\begin{remark}\label{rem:good_diff}
  The canonical divisors providing the equality between the
  $\mathcal C_\Omega$ and the $\mathcal C_L$ is the divisor of a
  differential form having simple poles with residue equal to $1$ at
  all the $P_i's$. The existence of such a differential is a
  consequence of the weak approximation theorem (see \cite[Lem.~2.2.9
  \& Prop.~2.2.10]{stichtenoth2009book}).
\end{remark}

The parameters of AG codes satisfy the following basic estimates
(\cite[Th.~2.2.2~\&~2.2.7]{stichtenoth2009book}):
\begin{theorem}\label{thm:basics}
The evaluation code $\CL{X}{\cP}{G}$ is a linear code of 
length $n=\card{\cP}=\deg(D_{\cP})$ and dimension
\[k=\ell(G)-\ell(G-D_{\cP}).\]
In particular, if $\deg(G)<n$, then
\[k=\ell(G)\geq\deg(G)+1-g,\]
and if moreover $2g-2<\deg(G)<n$, then
\[k=\ell(G)=\deg(G)+1-g,\]
where $g$ is the genus of $X$.

Its minimum distance $d=d(\CL{X}{\cP}{G})$ satisfies
\[d\geq \dgop^* \eqdef n-\deg(G)\]
where $\dgop^*$ is the so-called \emph{Goppa designed distance} of $\CL{X}{\cP}{G}$.
\end{theorem}

Joint with Lemma~\ref{dualityCLCOm} it gives likewise:
\begin{corollary}\label{cor:Goppa_bound_COmega}
If $2g-2<\deg(G)<n$, the residue code $\COm{X}{\cP}{G}$ has dimension
\[
  k=n+g-1-\deg(G)
\]
and minimum distance
\[
  d\geq\dgop\eqdef\deg(G)+2-2g.
\]
\end{corollary}

Another important consequence of these bounds is the following
statement, providing a comparison of these bounds with the well--known
{\em Singleton bound} \cite[Th.~1.11]{sloane1977book}.

\begin{corollary}
\label{Singleton_defect}
Let $\code=\CL{X}{\cP}{G}$ with $\deg(G)<n$, or $\code=\COm{X}{\cP}{G}$ with $\deg(G)>2g-2$. Then
\[k+d\geq n+1-g\]
i.e. the \emph{Singleton defect} of $\code$ is at most $g$.
\end{corollary}

\begin{remark}
  In the sequel, both quantities $\dgop^*$ and $\dgop$ are referred
  to as the {\em Goppa bound} or the {\em Goppa
    designed distance}. They do not provide {\em a priori} the actual
  minimum distance but yield a lower bound. In addition, as we will
  see in Section~\ref{sec:decoding}, correcting errors up to half
  these lower bounds will be considered as good ``targets'' for
  decoding.
\end{remark}

Finally, automorphisms of $X$ give rise to automorphisms of evaluation
codes on it:
\begin{proposition}
\label{auto_courbe_code}
Assume $\supp(G)$ is disjoint from $\cP$, and let $\sigma$ be an
automorphism of $X$ such that $\sigma(\cP)=\cP$ and $\sigma^*G\sim G$
(see Definition~\ref{def:pullback}).
Let $\mathbf{P}_\sigma$ be the permutation matrix given by
$(\mathbf{P}_\sigma)_{i,j}=1$ if $P_i=\sigma(P_j)$ and
$(\mathbf{P}_\sigma)_{i,j}=0$ otherwise.  Also set
$\word{v}=(h(P_1),\dots,h(P_n))$, where $\div(h)=\sigma^*G-G$.  Then
the map
\[\word{c}\mapsto\word{c}\mathbf{P}_\sigma\star\word{v}\]
defines a linear automorphism of $\CL{X}{\cP}{G}$.
\end{proposition}

The proof of Proposition~\ref{auto_courbe_code} uses the following
lemma.

\begin{lemma}\label{lem:iso_RR_from_auto}
  In the context of Proposition~\ref{auto_courbe_code},
  the map 
  \[
    \map[\varphi_{\sigma}:]{\Fq(X)}{\Fq(X)}{f}{f\circ \sigma}
  \]
  induces an isomorphism
  $L(G) \stackrel{\scriptstyle{\sim}}{\longrightarrow} L(\sigma^* G)$.
\end{lemma}

\begin{proof}
  The map $\varphi_{\sigma}$ is clearly an isomorphism with inverse
  $h \mapsto h \circ \sigma^{-1}$.  Hence, we only need to prove that
  $\varphi_{\sigma}(L(G)) \subseteq L(\sigma^* G)$. From
  Definition~\ref{def:pullback},
  we have
  \begin{equation}\label{eq:sigma_pullback}
    \sigma^* G = \sum_{P} v_P(G) \sigma^{-1}(P).
   \end{equation}
   Next,
   for any place $P$ of $\Fq(X)$ and any $f \in \Fq(X)$,
   \begin{equation}\label{eq:sigma_pullback_f}
     v_P(f) = v_{\sigma^{-1}(P)}(f\circ \sigma).
   \end{equation}
   Combining (\ref{eq:sigma_pullback}) and (\ref{eq:sigma_pullback_f}),
   for any place $P$ of $\Fq(X)$, we have
   \[
     v_P(f) \geq -v_P(G) \quad \Longrightarrow
     \quad v_{\sigma^{-1}(P)}(f \circ \sigma) \geq
     v_{\sigma^{-1}(P)}(\sigma^*G).
   \]
   This yields the result.
\end{proof}

\begin{proof}[Proof of Proposition~\ref{auto_courbe_code}]
  The map $\varphi_{\sigma}$ of Lemma~\ref{lem:iso_RR_from_auto}
  induces an isomorphism
  \[
    \map[\phi_{\sigma} :]{\CL{X}{\cP}{G}}{\CL{X}{\cP}{\sigma^* G}}
    {(f(P_1),\dots , f(P_n))}{(f(\sigma(P_1)), \dots, f(\sigma(P_n))),}
  \]
  which, by definition of $\mathbf{P}_{\sigma}$, is nothing but the map
  $\cv \mapsto \cv \mathbf{P}_{\sigma}$.
  Next, Lemma~\ref{lineq=diageq} yields
  an isomorphism
  \[
    \map[\psi :]{\CL{X}{\cP}{\sigma^*G}}{\CL{X}{\cP}{G}}{\cv}{\cv \star \vv.}
  \]
  The composition map $\psi \circ \phi_{\sigma}$ provides
  an automorphism of $\CL{X}{\cP}{G}$ which is 
  the map $\cv \mapsto \cv \mathbf{P}_{\sigma} \star \vv$.
\end{proof}

\subsection{Genus $0$, generalized Reed--Solomon and classical Goppa
  codes}\label{subsec:GRS_codes}

In this section, we focus on algebraic geometry codes from the
projective line $\P^1$. We equip this line with homogeneous
coordinates $(X:Y)$. We denote by $x$ the rational function
$x \eqdef \frac X Y$ and for any $x_i \in \Fq$ we associate the point
$P_i = (x_i:1)$. Finally, we denote by $P_\infty \eqdef (1:0)$.

\begin{remark}
  In this \ifx\localversion\undefined {chapter}
  \else {article}
  \fi
  , accordingly to the usual notation in algebraic geometry,
  the projective space of dimension $m$ is denoted as $\P^m$.
  In particular, its set of rational points $\P^m(\Fq)$
  is the finite set sometimes denoted as $PG(m, q)$
  in the literature of finite geometries and combinatorics.
\end{remark}

\subsubsection{The $\mathcal C_L$ description}
One of the most famous families of codes is probably that of
Reed--Solomon codes.

\begin{definition}\label{def:GRS}
  Let $\xv = (x_1, \dots, x_n)$ be an $n$--tuple of distinct elements
  of $\Fq$ and $\yv = (y_1, \ldots, y_n)$ be an $n$--tuple of nonzero
  elements of $\Fq$. Let $k < n$, the {\em generalized Reed--Solomon}
  (GRS) code of dimension $k$ associated to the pair $(\xv,\yv)$ is
  defined as
  \[
    \GRS{k}{\xv}{\yv} \eqdef \{(y_1f(x_1), \dots, y_nf(x_n)) ~|~ f
    \in \Fq[X],\ \deg f < k\},
  \]
  where, by convention, the zero polynomial has degree $- \infty$.
  A {\em Reed--Solomon} code is a GRS one with $\yv = (1, \dots, 1)$
  and is denoted as $\RS{k}{\xv}$.
\end{definition}

\begin{remark}\label{rem:diagonal_eq_RS}
  In terms of diagonal equivalence, any generalized Reed--Solomon code
  is diagonally equivalent to a Reed--Solomon one thanks to the
  obvious relation
  \[
    \GRS{k}{\xv}{\yv} = \RS{k}{\xv} \star \yv.
  \]
\end{remark}

\begin{remark}
  Beware that our definition of {\em Reed--Solomon codes} slightly
  differs from the most usual one in the literature
  \ifx\localversion\undefined { and in particular may differ from that
    of other chapters of this encyclopedia}
  \else
  \fi
  . Indeed, most of the references define Reed--Solomon codes as
  cyclic codes of length $q-1$, i.e. as a particular case of BCH
  codes. For instance, see \cite[\S~10.2]{sloane1977book},
  \cite[\S~5.2]{huffman03book} \cite[\S~5.2]{roth06book} or
  \cite[Def.~6.8.1]{vanLint99book}. Note that this commonly
  accepted definition is not exactly
  the historical one made by Reed and Solomon themselves in
  \cite{reed60siam}, who introduced a code of length $n = 2^m$ over
  $\F_{2^m}$ which is not cyclic.

  Further, Reed--Solomon codes ``of length $q$'' are sometimes referred to
  as {\em extended Reed--Solomon codes}. Next, using
  Remark~\ref{UseWeakApprox},
  one can actually define generalized
  Reed--Solomon codes of length $q+1$, corresponding to the codes
  called {\em (generalized) doubly extended Reed--Solomon codes} in
  the literature.
\end{remark}

Generalized Reed--Solomon codes are known to have length $n$,
dimension $k$ and minimum distance $d = n-k+1$. That is to say, such
codes are {\em Maximum Distance Separable} (MDS), i.e. they reach
Singleton bound \cite[Th.~1.11]{sloane1977book} asserting that for any
code of length $n$ and dimension $k$ and minimum distance $d$, we
always have $k+d \leq n+1$.  In addition, many algebraic constructions
of codes such as BCH codes, Goppa codes or Srivastava codes derive
from some particular GRS codes by applying the {\em subfield subcode
  operation}.
\begin{definition}\label{def:subfield_subcode}
  Consider a finite field $\Fq$ and its degree $m$ extension $\Fqm$
  for some positive integer $m$. Let $\CC \subseteq \Fqm^n$
  be a linear code, the {\em subfield subcode} of $\CC$ is the code:
  \[
  \CC \cap \Fq^n.
  \]
\end{definition}
The above--defined operation is
of particular interest for public-key cryptography applications. All
these codes fit in a broader class called {\em alternant codes}. See
for instance \cite[Chap.~12, Fig. 12.1]{sloane1977book}.

In some sense, algebraic geometry codes are natural generalisations of
generalized Reed--Solomon codes, the latter being algebraic geometry
codes from the projective line $\P^1$. Let us start with the case of
Reed--Solomon codes.

\begin{proposition}\label{prop:RS_from_P1}
  Let $\xv = (x_1, \ldots, x_n) \in \Fq^n$ be an $n$--tuple of
  distinct elements. Set $\cP= ((x_1:1), \dots,
  (x_n:1)) \in \P^1$ and 
  $P_\infty = (1:0)$. Then, the code
  $\CL{\P^1}{\cP}{(k-1)P_\infty}$ is nothing but the Reed--Solomon
  code $\RS{k}{\xv}$.
\end{proposition}

\begin{proof}
By definition, the Riemann--Roch
space $L((k-1)P_\infty)$ is the space of rational functions with
a pole of order less than $k$ at infinity, which is nothing but
the space of polynomials of degree less than $k$.  
\end{proof}

More generally, the equivalence between GRS codes and algebraic geometry
codes from $\P^1$ is summarized here.

\begin{theorem}\label{thm:codes_on_P1}
  Any generalized Reed--Solomon code is an AG code from $\P^1$
  whose evaluation points avoid $P_\infty$. Conversely,
  any such AG code is a GRS one.
  More precisely:
  \begin{enumerate}[(i)]    
  \item\label{item:GRS_to_P1} for any pair $(\xv, \yv)$ as in
    Definition~\ref{def:GRS}, we have
    $\GRS{k}{\xv}{\yv} = \CL{\P^1}{\cP}{G}$,
    with
    \[
      \cP = ((x_1:1),\dots, (x_n:1)) \quad {\rm and} \quad
      G = (k-1)P_\infty - \div (h),
    \]
    where $h$ is the Lagrange interpolation polynomial satisfying
    $\deg h < n$  and $h(x_i) = y_i$ for any $1 \leq i \leq n$.
  \item\label{item:P1_to_GRS} Conversely, for any ordered $n$--tuple
    $\cP$ of distinct rational points of $\P^1\setminus \{P_\infty\}$
    with coordinates $(x_1:1),\dots, (x_n:1)$ and any divisor $G$ of
    $\P^1$ of degree $k-1$ whose support avoids $\cP$,
    we have $\CL{\P^1}{\cP}{G} = \GRS{k}{\xv}{\yv}$ where
    \[
      \xv = (x_1, \ldots, x_n) \quad {\rm and} \quad
      \yv = (f(x_1), \dots, f(x_n))
    \]
    for some function
    $f \in L(G - (k-1)P_\infty) \setminus \{0\}$.
  \end{enumerate}
\end{theorem}

\begin{proof}
  Using Remark~\ref{rem:diagonal_eq_RS}, it suffices to prove that
  $\CL{\P^1}{\cP}{G} = \RS{k}{\xv} \star \yv$. This last equality is
  deduced from Proposition~\ref{prop:RS_from_P1} together with
  Lemma~\ref{lineq=diageq} since $(k-1)P_\infty - G = \div (h)$ ,
  proving (\ref{item:GRS_to_P1}).

  Conversely, since $\deg G = k-1$ and the genus of $\P^1$ is $0$, then, from
  Corollary~\ref{cor:RR}, we deduce that the space
  $L(G-(k-1)P_\infty)$ has dimension $1$. Let us take any nonzero
  function $f$ in this space. Then, by definition,
  \[
    \div (f) \geq - G + (k-1)P_\infty
  \]
  and, for degree reasons, the above inequality is an equality.
  Set $\yv = (f(x_1), \dots, f(x_n))$.
  Again from Proposition~\ref{prop:RS_from_P1} and
  Lemma~\ref{lineq=diageq}, we deduce that
  \[
    \CL{\P^1}{\cP}{G} = \CL{\P^1}{\cP}{(k-1)P_\infty} \star \yv =
    \RS{k}{\xv} \star \yv = \GRS{k}{\xv}{\yv}.
  \]
  This proves (\ref{item:P1_to_GRS}).
\end{proof}

\begin{remark}
  Theorem~\ref{thm:codes_on_P1} asserts that any AG code from the
  projective line is diagonally equivalent to a {\em one--point code},
  i.e. an AG code whose divisor $G$ is supported by one point.
  This fact can be directly observed without introducing the notation
  of Reed--Solomon codes, by using a classical result in algebraic geometry
  asserting that on $\P^1$, two divisors are linearly equivalent if and
  only if they have the same degree.
\end{remark}

\subsubsection{The $\mathcal C_{\Omega}$ description}\label{subsec:class_Goppa}
Another way to describe codes from the projective line is to use the
differential description. Note that, the equivalence between the
$\mathcal C_L$ and the $\mathcal C_\Omega$ description can be made
easily explicit in the $\P^1$ case as follows.

\begin{proposition}
  Let $\cP = (P_1, \ldots, P_n)$ be an ordered $n$--tuple of distinct
  points of $\P^1 \setminus \{P_\infty\}$ with respective homogeneous
  coordinates $(x_1:1), \dots, (x_n:1)$ and $G$ be a divisor on $\P^1$
  whose support avoids the $P_i$'s. Then
  \[ \COm{\P^1}{\cP}{G} = \CL{\P^1}{\cP}{\div (\omega) - D_\cP + G}, \]
  with
  $\omega \eqdef \frac{dh}{h}$,
  where $h(x) \eqdef \prod_{i=1}^n (x-x_i)$.
\end{proposition}

\begin{proof}
  A classical result on logarithmic differential forms asserts that
  $\frac{dh}{h} = \sum_{i} \frac{dx}{x-x_i}$ has simple poles at the
  $P_i$'s with residue $1$ at them. Then, we conclude using
  Lemma~\ref{dualityCLCOm} and Remark~\ref{rem:good_diff}.
\end{proof}

Despite $\mathcal C_L$ and $\mathcal C_\Omega$ definitions are
equivalent thanks to Lemma~\ref{dualityCLCOm}, the differential
description is of interest since it permits to redefine the so--called
{\em classical Goppa codes} \cite{goppa1970ppi,goppa1971ppi}.  Given an
$n$--tuple $\xv = (x_1, \ldots, x_n) \in \Fq$ with distinct entries, a
polynomial $f \in \Fq[X]$ which does not vanish at any of the entries
of $\xv$ and a subfield $\mathbb K$ of $\Fq$, then the {\em classical
  Goppa code} associated to $(\xv, f, \mathbb K)$ is defined as
\begin{equation}\label{eq:classical_Goppa_code}
  \Gamma (\xv, f, \mathbb K) \eqdef \left\{ \cv = (c_1, \dots, c_n)
    \in \mathbb K^n ~\bigg|~ \sum_{i=1}^n \frac{c_i}{X - x_i} \equiv 0
    \mod (f) \right\}.
\end{equation}

The major interest of this definition lies in the case of a proper
subfield $\mathbb K \varsubsetneq \Fq$, for which the corresponding
code benefits from a better estimate of its parameters compared to
general alternant codes \cite{sloane1977book, sugiyama1976it}.  This
estimate comes with an efficient decoding algorithm called {\em
  Patterson algorithm} \cite{patterson1975it} correcting up to half
the designed minimum distance.  However, in what follows, we are
mainly interested in the relation between this construction and that
of $\mathcal C_\Omega$ codes and, for this reason, we will focus on
the case $\Fq = \mathbb K$.

\begin{remark}
  Note that the terminology might be misleading here. Algebraic
  geometry codes are sometimes referred to as {\em Goppa codes}, while
  {\em classical Goppa codes} are {\bf not} algebraic geometry codes from
  $\P^1$ in general since $\mathbb K$ may be different from $\Fq$.
  These codes are subfield subcodes (see
  Definition~\ref{def:subfield_subcode}) of some algebraic geometry
  codes from $\P^1$.

  For this reason and to avoid any confusion, in the present
  \ifx\localversion\undefined chapter, \else article, \fi we refer to
  {\em algebraic geometry codes} when speaking about $\mathcal C_L$
  and $\mathcal C_{\Omega}$ codes and to {\em Goppa codes} or {\em
    classical Goppa codes} when dealing with codes as defined in
  \eqref{eq:classical_Goppa_code} with $\mathbb K \varsubsetneq \Fq$.
\end{remark}

\begin{theorem}
  Denote by $\cP = (P_1, \dots, P_n) = ((x_1:1),\dots, (x_n:1))$. The
  code $\Gamma (\xv, f, \Fq)$ equals $\COm{\P^1}{\cP}{(f)_0-P_\infty}$
  where $(f)_0$ is the effective divisor given by the zeroes of the
  polynomial $f$ counted with multiplicity.
\end{theorem}

\begin{proof}
  Let $\cv \in \Gamma(\xv, f, \Fq)$ and set
  \[
    \omega_{\cv} \eqdef \sum_{i=1}^n \frac{c_i dx}{x-x_i} \cdot
  \]
  By definition of $\Gamma (\xv, f, \Fq)$, the form $\omega_{\cv}$
  vanishes on $(f)_0$. In addition, it has simple poles at all the
  $P_i$'s and is regular at any other point of $\P^1 \setminus \{P_\infty\}$.
  There remains to check its valuation at infinity. This can be done
  by replacing $x$ by $1/u$ and get:
  \[
    \omega_{\cv} = \sum_{i=1}^n \frac{-c_i \frac{du}{u^2}}{\frac{1}{u} - x_i}
    = -\sum_{i=1}^n \frac{c_i du}{u(1-x_iu)}\cdot
  \]
  We deduce that $\omega_{\cv}$ has valuation $\geq -1$ at $P_\infty$
  and hence $\omega_{\cv} \in \Omega((f)_0 - P_\infty - D_{\cP})$,
  which yields
  $\Gamma(\xv, f, \Fq) \subseteq \COm{\P^1}{\cP}{(f)_0 - P_{\infty}}$

  Conversely, given
  $\omega \in \Omega((f)_0 - P_\infty - D_{\cP})$, consider the
  differential form
  \[
    \eta \eqdef \sum_{i=1}^n \frac{\res_{P_i}(\omega) dx}{x-x_i}\cdot
  \]
  The previous discussion shows that the poles of $\eta$ are simple
  and contained in $\{P_1, \dots, P_n, P_\infty\}$. In addition, the
  two forms have simple poles and the same residue at any of the points
  $P_1, \ldots, P_n$ and, by the residue formula
  \cite[Cor.~4.3.3]{stichtenoth2009book} they also should have the
  same residue at $P_\infty$. Therefore, the differential form
  $\eta - \omega$ has no pole on $\P^1$. Moreover, since the degree of
  a canonical divisor is $2g-2 = -2$, a nonzero rational differential
  form on $\P^1$ should have poles. Therefore, $\eta = \omega$ and we
  deduce that the rational function
  $\sum_i \frac{\res_{P_i}(\omega)}{(x-x_i)}$ vanishes on $(f)_0$ or,
  equivalently that
  \[
    \sum_{i=1}^n \frac{\res_{P_i}(\omega)}{x-x_i} \equiv 0 \mod (f).
  \]
  This yields the converse inclusion
  $\COm{\P^1}{\cP}{(f)_0 - P} \subseteq \Gamma(\xv, f, \Fq)$
  and concludes the proof.
\end{proof}

 \section{Asymptotic parameters of algebraic geometry
  codes}\label{sec:asymptotics}
\subsection{Preamble}

A nonconstructive argument shows that asymptotically good codes,
i.e. $[n,k,d]$-codes with
\begin{itemize}
\item $n\to\infty$
\item $\liminf\frac{k}{n}\geq R>0\quad$ (positive asymptotic rate)
\item $\liminf\frac{d}{n}\geq\delta>0\quad$ (positive asymptotic
  relative minimum distance)
\end{itemize}
exist over any given finite field $\F_q$.  More precisely, the
asymptotic version of the Gilbert-Varshamov bound shows that this is
possible for $0<\delta<1-1/q$ and $R=1-H_q(\delta)$,
where \[H_q(x)\eqdef x\log_q(q-1)-x\log_q(x)-(1-x)\log_q(1-x)\] is the
$q$-ary entropy function.

For a long time it remained an open question to give an explicit
description, or even better an effectively computable construction, of
asymptotically good codes.  The first such construction was proposed
by Justesen in \cite{Justesen1972}.  Let us briefly recall how it
works (in a slightly modified version).

Let $m$ be an integer, and assume that we have an explicit
$\F_q$-linear identification of $\F_{q^m}$ with $\F_q^m$.  Set
$n=q^m-1$ and let $\word{x}=(x_1,\dots,x_n)$ be the collection of all
nonzero elements in $\F_{q^m}$.  For an integer $k<n$, consider the
evaluation map
\begin{equation*}
\map{\F_{q^m}[X]_{<k}}{(\F_{q^m})^{2n}}{f(X)}{(f(X),Xf(X))(\word{x})},
\end{equation*}
where $\F_{q^m}[X]_{<k}$ denotes the subspace of polynomials of degree less than $k$, and
\[(f(X),Xf(X))(\word{x})\eqdef (f(x_1),x_1f(x_1),f(x_2),x_2f(x_2),
\dots,f(x_n),x_nf(x_n)).\]
The image of this evaluation map is a $[2n,k]$-code over $\F_{q^m}$,
and identifying $\F_{q^m}$ with $\F_q^m$ this becomes a
$[2nm,km]$-code $\code$ over $\F_q$.
\begin{theorem}
  Let $0<R<\frac{1}{2}$.  Then as $m\to\infty$ and
  $\frac{k}{n}\to 2R$, the codes $\code$ are asymptotically good, with
  asymptotic rate $R$ and asymptotic relative minimum distance at
  least $(1-2R)H_q^{-1}(\frac{1}{2})$.
\end{theorem}
The idea of the proof is that, for $\epsilon>0$ and $m\to\infty$, the
number of words in $\F_q^{2m}$ of relative weight less than
$H_q^{-1}(\frac{1}{2}-\epsilon)$ is roughly
$q^{2m(\frac{1}{2}-\epsilon)}=q^{m(1-2\epsilon)}$, which is
exponentially negligible compared to $n=q^m-1$.  Moreover, if such
a vector, seen in $\F_{q^m}\times\F_{q^m}$, is of the form
$(\alpha,x_i\alpha)$, then it uniquely determines $x_i$.  Now if $f$
has degree $k\leq 2Rn$, there are at least $(1-2R)n$ values of $x_i$
such that $f(x_i)\neq0$, and then, except for a negligible fraction of
them, $(f(x_i),x_if(x_i))$ has weight at least
$2mH_q^{-1}(\frac{1}{2}-\epsilon)$ in $\F_q^{2m}$.  This concludes.

\subsection{The Tsfasman-\vladut-Zink bound}
Algebraic geometry also provides a\-symptotically good codes.  Dividing
by $n$ in the basic estimate of Corollary~\ref{Singleton_defect} and
setting $R=\frac{k}{n}$ and $\delta=\frac{d}{n}$ gives codes whose
rate $R$ and relative minimum distance $\delta$ satisfy
$R+\delta\geq 1+\frac{1}{n}-\frac{g}{n}$.  Letting $n\to\infty$
motivates the following:
\begin{definition}\label{Ihara_constant}
  The Ihara constant of $\F_q$
  is \[A(q)=\limsup_{g(X) \rightarrow \infty}\frac{n(X)}{g(X)}\] where $X$ ranges over all
  curves over $\F_q$, and $n(X)=|X(\F_q)|$ is the number of rational
  points of $X$.
\end{definition}
We then readily get:
\begin{theorem}\label{generalTVZ}
  Assume $A(q)>1$. Then, for any $R,\delta>0$
  satisfying \[R+\delta=1-\frac{1}{A(q)},\] there exist asymptotically
  good codes with asymptotic rate at least $R$ and asymptotic relative
  minimum distance at least $\delta$.
\end{theorem}
For this result to be meaningful, we need estimates on $A(q)$. 

Let us start with upper bounds. First, the well-known Hasse-Weil bound
\cite[Th.~5.2.3]{stichtenoth2009book}
implies $A(q)\leq 2\sqrt{q}$.  Several improvements were proposed,
culminating with the following, known as the Drinfeld-\vladut{} bound.
\begin{theorem}[\cite{vladut1983faa}]\label{DVbound}
For any $q$ we have $A(q)\leq \sqrt{q}-1$.
\end{theorem}

On the other hand, lower bounds on $A(q)$ combine with
Theorem~\ref{generalTVZ} to give lower bounds on codes.
\begin{itemize}
\item When $q=p^{2m}$ ($m\geq1$) is a square, we have
  $A(q)\geq\sqrt{q}-1$. This was first proved by Ihara
  in~\cite{ihara1981jfsuTokyo} then independently in
  \cite{tsfasman1982mn}, using modular curves if $m=1$, and Shimura
  curves for general $m$.  Observe that Ihara's lower bound matches
  the Drinfeld-\vladut{} bound, so we actually get equality:
  $A(q)=\sqrt{q}-1$.  Other more effective constructions matching the
  Drinfeld-\vladut{} bound were later proposed, for instance in
  \cite{GS1995}.  These constructions use recursive towers of curves,
  although it was observed by Elkies \cite{elkies1997allerton,elkies2001ecm}
  that they in fact yield modular curves.

  Combined with Theorem~\ref{generalTVZ}, Ihara's lower bound gives
  asymptotically good codes with
  \[R+\delta\geq 1-\frac{1}{\sqrt{q}-1}\cdot\] This is known as the
  Tsfasman-\vladut-Zink bound~\cite{tsfasman1982mn}.  For 
  $q\geq 49$, it is shown that this TVZ bound beats the GV bound on a
  certain interval as illustrated in Figure~\ref{fig:TVZ}
  
\item When $q=p^{2m+1}$ ($m\geq1$) is a non-prime odd power of a
  prime, Bassa, Beelen, Garcia, and Stichtenoth~\cite{BBGS2015} show
  \[A(p^{2m+1})\geq \frac{2(p^{m+1}-1)}{p+1+\frac{p-1}{p^m-1}}=
  \left(\frac{1}{2}((p^m-1)^{-1}+(p^{m+1}-1)^{-1})\right)^{-1}.\]
The proof is constructive and uses a recursive tower of curves,
although these curves can also be interpreted in terms of Drinfeld
modular varieties.
\item For general $q$, Serre~\cite{Serre1983} shows
  $A(q)\geq c\log(q)>0$ for a certain constant $c$ which can be taken
  as $c = \frac{1}{96}$ (see \cite[Th.~5.2.9]{niederreiter2001book}).
  When $q=p$ is prime this is often the best one
  knows.  
\end{itemize}
\begin{figure}[h]
  \centering
  \includegraphics[scale=.7]{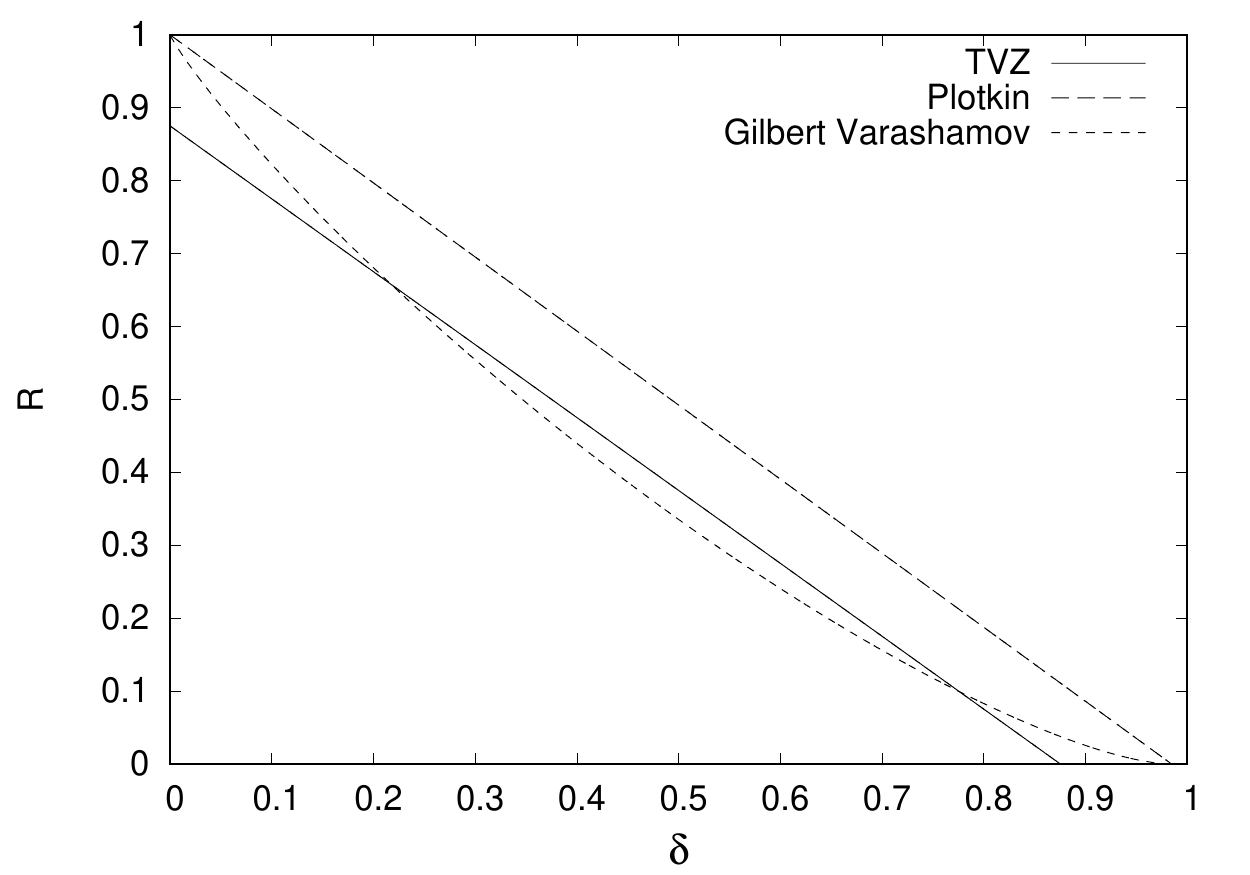}
  \caption{Tsfasman--\vladut--Zink bound for $q = 64$}
  \label{fig:TVZ}
\end{figure}
There are some explicit lower bounds for $A(p)$ where $p=2,3,5,\dots$
is a small prime.  However these bounds cannot be used in
Theorem~\ref{generalTVZ}, which requires $A(q)>1$.  Indeed, the
Drinfeld-\vladut{} bound actually shows $A(p)<1$ for $p=2$ or $3$.

\subsection{Subfield subcodes, Katsman-Tsfasman-Wirtz bound}
Tsfasman--\vladut{}--Zink bound provides remarkable results when
$q\geq 49$ but, on the other hand, it turns out to be inefficient for
smaller values of $q$ and is in particular irrelevant for $q \leq 4$ where
it does not even prove the existence of asymptotically good families of
algebraic geometry codes.

Very often, lower bounds on codes over a small base field can be
obtained by considering good codes over an extension field, and then
using either a concatenation argument or a subfield subcode argument
(see Definition~\ref{def:subfield_subcode}).
In \cite{wirtz1988it} and independently in \cite{katsman1989cm} the
parameters of subfield subcodes of algebraic geometric codes and
asymptotic parameters are studied. More recently, a slightly different
construction called {\em Cartier codes} \cite{couvreur2014ams} using
the Cartier operator has been proposed providing the same asymptotic
parameters.  Namely,

\begin{theorem}
  For any even positive integer $\ell$, there exists a sequence of
  codes over $\Fq$ defined as subfield subcodes of codes over
  $\F_{q^\ell}$ whose asymptotic parameters $(R, \delta)$ satisfy
  \begin{equation}\nonumber
    R \geq 1-\frac{2(q-1)\ell}{q(q^{\ell /2} -1)}-\frac{(q-1)\ell}{q}\delta \qquad {\rm for}\qquad \frac{q-2}{q^{\ell/2}-1} \leq \delta \leq \frac{q}{m(q-1)} - \frac{2}{q^{\ell/2}-1}\cdot
\end{equation}
\end{theorem}
In particular, when $\ell \rightarrow \infty$ the asymptotic parameters
of such codes reach Gilbert--Varshamov bound for $\delta \sim 0$.

\subsection{Nonlinear codes}

Works of Xing~\cite{Xing2003}, Elkies~\cite{Elkies2003}, and Niederreiter-\"Ozbudak~\cite{NO2004}
show that it is possible to construct codes with asymptotic parameters better
than those in Theorem~\ref{generalTVZ}, and in particular better than the TVZ bound,
if one turns to \emph{nonlinear} codes.
Observe then that, strictly speaking, these codes are not AG codes of the form $\CL{X}{\cP}{G}$.
However they are still obtained by evaluation of functions on an algebraic curve.
More precisely, the original construction of these codes uses derivative evaluation of functions on curves,
combined with certain intricate combinatorial arguments.
An alternative, arguably simpler construction is proposed in \cite{SX2005},
still based on curves.

\begin{theorem}
For any prime power $q$, there exist asymptotically good nonlinear codes over $\F_q$
with asymptotic relative minimum distance $\delta$
and asymptotic rate \[R\geq 1-\delta-\frac{1}{A(q)}+\log_q\left(1+\frac{1}{q^3}\right).\]
This holds for any $\delta>0$ such that this quantity is positive.
\end{theorem}

 \section{Improved lower bounds for the minimum
  distance}\label{sec:dmin}
To write this section, the authors followed some advices of
Iwan Duursma. They warmly thank him for this help.

Given a curve $X$ over $\Fq$ of genus $g$, a divisor $G$ and an
ordered set of $n$ rational points $\cP$ such that $n > \deg G$
(resp. $\deg G > 2g-2$), from Corollary~\ref{Singleton_defect}
\begin{equation}\label{eq:almost_MDS}
  d+k \geq n+1-g.
\end{equation}
Therefore, the parameters of an algebraic geometry code are ``at
distance at most $g$'' from the Singleton bound. However, Goppa bound
is not always reached and improvements may exist under some
hypotheses. The present section is devoted to such possible
improvements.  It should be emphasized that this concerns only
``finite length'' codes: the statements to follow do not provide any
improvement on the asymptotic performances of (linear) algebraic geometry
codes presented in Section~\ref{sec:asymptotics}.

\medskip

First, let us try to understand why Goppa bound may be not reached
by proving the following lemma.

\begin{lemma}\label{lem:Reach_Goppa}
  The code $\CL{X}{\cP}{G}$ has minimum distance
  $d = \dgop^* = n - \deg G$ if and only if there exists $f \in L(G)$
  such that the positive divisor $\div(f)+G$ is of the form
  $P_{i_1} + \cdots + P_{i_s}$ where the $P_{i_j}'s$ are distinct
  points among $P_1, \dots, P_n$.
\end{lemma}

\begin{proof}
  Suppose there exists such an $f \in L(G)$ satisfying
  $\div (f)+G = P_{i_1} + \cdots + P_{i_s}$. Since principal
  divisors have degree $0$, then $s = \deg G$. Consequently, the
  corresponding codeword vanishes at positions with index
  $i_1, \dots, i_s$ and hence has weight $n - \deg G$. Thus, such a
  code reaches Goppa bound.

  Conversely, if Goppa bound is reached, then there exists $f \in L(G)$
  vanishing at $s=\deg G$ distinct points $P_{i_1}, \dots, P_{i_s}$
  among the $P_i$'s. Hence
  \[
    \div (f) \geq - G + P_{i_1} + \cdots + P_{i_s}.
  \]
  For degree reasons, the above inequality is an equality.
\end{proof}

\begin{remark}
  In algebraic geometry, the {\em complete linear system} or {\em
    complete linear series} associated to a divisor $G$ denoted by
  $|G|$ is the set of positive divisor linearly equivalent to
  $G$. Such a set is parameterized by the projective space $\P
  (L(G))$. Using this language, one can claim that {\em the code
    $\CL{X}{\cP}{G}$ reaches Goppa bound if and only if $|G|$ contains a
    reduced divisor supported by the $P_i$'s}.
\end{remark}

For instance, one can provide examples of codes from elliptic curves
(i.e. curves with $g=1$) that are MDS, i.e. reach the Singleton bound.

\begin{example}
  Recall that given an elliptic curve $E$ over a finite field $\Fq$ with a
  fixed rational point $O_E$, the set of rational points has a natural
  structure of abelian group with zero element $O_E$ given by the
  chord--tangent group law (see \cite[\S~III.2]{silverman2009book}).
  The addition law will be denoted by $\oplus$.
  This group law can also be deduced from linear equivalence of
  divisors (see \eqref{eq:equivalence_lineaire})  as follows:
  \begin{equation}\label{eq:elliptic_group_law}
    P\oplus Q=R \quad \Longleftrightarrow \quad (P-O_E) + (Q-O_E) \sim (R-O_E)
  \end{equation}
  (see \cite[Prop.~3.4]{silverman2009book}).

  From \cite[Th.~3.3.15.5a]{tsfasman2007book}, there exists an
  elliptic curve $E$ over a field $\Fq$ whose group of rational points
  $E(\Fq)$ is cyclic of cardinality $q+1$. Let $P \in E(\Fq)$ be a
  generator of the group of points and set
  $\cP = (P_1, P_2, \dots, P_n)$ for some positive integer $n$ such
  that $\frac{n(n+1)}{2} < q$ and where for any
  $i \in \{1, \dots, n\}$, $P_i = iP = P \oplus \cdots \oplus P$.
  Now, let $q+1 > r > \frac{n(n+1)}{2}$ and set $Q = rP$, choose an
  integer $n > k > 0$ and consider $Q_1, \ldots, Q_k \in E(\Fq)$
  (possibly non distinct) such that
  $Q_1 \oplus \cdots \oplus Q_k = Q$. Let $G \sim Q_1+\cdots + Q_k$ be
  a divisor. From (\ref{eq:elliptic_group_law}), we have
  \begin{equation}
    \label{eq:G}
    (G - kO_E) \sim Q-O_E.
  \end{equation}
  Consider the code $\CC \eqdef \CL{E}{\cP}{G}$. Its length equals $n$.
  Moreover, since $\deg G > 0 = 2g-2$, from Theorem~\ref{thm:basics},
  we have $\dim \CC = k+1-g = k$. 

  We claim that $\CC$ is MDS. Indeed, from the Singleton bound and
  (\ref{eq:almost_MDS}), the minimum distance of the code is either
  $n-k$ or $n-k+1$. Suppose the minimum distance is $n-k$, then, from
  Lemma~\ref{lem:Reach_Goppa}, there exists $f \in L(G)$ and $k$
  distinct points $P_{i_1}, \dots, P_{i_k}$ among $P_1,\dots, P_n$
  such that
  \begin{equation}\label{eq:div_f_elliptic_curve}
    \div (f) = P_{i_1} + \cdots + P_{i_k} - G.
  \end{equation}
  Therefore
  \begin{align*}
    G \sim P_{i_1}+ \cdots + P_{i_k} \quad &\Longleftrightarrow \quad
                                             G-kO_E \sim (P_{i_1}-O_E) +
                                             \cdots + (P_{i_k} - O_E)\\
                                           &\Longleftrightarrow \quad
                                             Q =rP= P_{i_1} \oplus \cdots
                                             \oplus P_{i_k},
  \end{align*}
  where the last equivalence is a consequence of (\ref{eq:elliptic_group_law})
  and (\ref{eq:G}), but contradicts the assertion $r > \frac{n(n+1)}{2}$.
  Thus, $\CC$ is MDS.
\end{example}

\begin{remark}
  The existence of MDS codes from elliptic curves is
  further discussed in \cite[\S~4.4.2]{tsfasman2007book}.
\end{remark}

A significant part of the
literature on algebraic geometry codes proposes improvements of Goppa
bound under some assumptions on the pair $(\cP, G)$ and applies them on
examples in order to beat records.
These improved bounds can roughly be split into two main categories
\begin{enumerate}
\item {\em Floor bounds} rest on the use of {\em base points}
  of divisors related to $G$;
\item {\em Order bounds} rest on filtrations of the code
  \[
    \CL{X}{\cP}{G} \supset \CL{X}{\cP}{G_1} \supset \CL{X}{\cP}{G_2}
    \supset \cdots
    \]for some strictly decreasing sequence of divisors
  $(G_i)_i$ and the iterated estimates of the minimum distance of the
  sets $\CL{X}{\cP}{G_i} \setminus \CL{X}{\cP}{G_{i+1}}$.
\end{enumerate}
A nice overview of many improved bounds in the literature is given in
\cite{duursma2011jpaa}.

Thanks to Lemma~\ref{dualityCLCOm}, we know that any $\CC_L$ code is a 
$\CC_\Omega$ one. 
Therefore, always choosing the most convenient point of view, we alternate between 
improved lower bounds for the minimum distance of $\CC_L$ and $\CC_\Omega$
codes. Our point is to provide lower bounds for the minimum distance
that improve the Goppa designed distance
\[
  \dgop \eqdef \deg (G-K) = \deg G +2 - 2g,
\]
where $K$ denotes a canonical divisor on $X$.

\subsection{Floor bounds}
Given a divisor $A$ on a curve $X$, a point $P$ of $X$ is said to be a
{\em base point of $A$} if $L(A) = L(A-P)$.

\begin{remark}
  Actually, the notion of base point depends only on the divisor class.
  Thus, if $A' \sim A$, then $P$ is also a base point of $A'$. 
\end{remark}

\begin{remark}
  From Riemann--Roch Theorem, any divisor $A$ such that $\deg A > 2g-1$,
  has no base points.
\end{remark}

If a divisor $G$ has a base point $P$ outside the set $\cP$,
then $\CL{X}{\cP}{G} = \CL{X}{\cP}{G-P}$ which entails that
the minimum distance of this code satisfies $d \geq n - \deg G +1$
instead of $n - \deg G$.
Similarly, for $\CC_\Omega$ codes, if $G-K$ has a base point, we have:

\begin{lemma}[{\cite[Lem.~1.3]{duursma2011jpaa}}]
  Let $P$ be a base point of $G-K$ and $P$ is disjoint from the elements of
  $\cP$, then 
  \[
    d(\COm{X}{\cP}{G}) \geq \dgop + 1.
  \]
\end{lemma}

\begin{remark}
  According to Remark~\ref{UseWeakApprox}, one can get rid of the hypothesis
  that $P$ is disjoint from the elements of $\cP$. In such a situation,
  the codes $\CL{X}{\cP}{G}$ and $\CL{X}{\cP}{G-P}$ are only
  diagonally equivalent, but the result on the minimum distance
  still holds.
\end{remark}

More generally, the {\em floor} $\lfloor A \rfloor$ of a divisor $A$
is the divisor of smallest degree $A'$ such that $L(A') = L(A)$. Such
a divisor satisfies $\lfloor A \rfloor \leq A$
\cite[Prop.~2.1]{maharaj2005jpaa} and we have

\begin{theorem}[Maharaj, Matthews, Pirsic
  {\cite[Th.~2.9]{maharaj2005jpaa}}]
  The code $\CL{X}{\cP}{G}$ has dimension $k$ and minimum distance
  $d$ which satisfy
  \[
    k \geq \deg G +1 - g \qquad d \geq n - \deg \lfloor G \rfloor.
  \]
\end{theorem}

This principle has been used and improved in several references such
as
\cite{maharaj2005jpaa,lundell2006jpaa,guneri2009jpaa,duursma2010ffa}. The
so--called {\em ABZ bound} due to Duursma and Park
\cite[Th.~2.4]{duursma2010ffa}, inspired from the AB bound of van
Lint and Wilson \cite[Th.~5]{vanLint1986it} permits to deduce many
other ones.

\begin{theorem}[Duursma, Park
  {\cite[Th.~2.4]{duursma2010ffa}}]\label{thm:ABZ_bound}
  Let $G=A+B+Z$ for a divisor $Z\geq 0$ whose support is disjoint from
  $\cP$. Then,
  \[
    d(\COm{X}{\cP}{G}) \geq \ell (A) - \ell (A-G+K) + \ell (B) - \ell
    (B-G+K).
  \]
\end{theorem}

\begin{remark}
  Goppa designed distance can be deduced from ABZ bound by choosing
  $A = G$ and $B = Z = 0$. Indeed, we get
  \begin{align*}
    d (\COm{X}{\cP}{G}) &\geq \ell (G) - \ell (K) + \ell (0) - \ell (K-G).  
  \end{align*}
  From Riemann--Roch theorem, we have $\ell (G) - \ell (K-G)= \deg G +1 -g$
  and $\ell (0) - \ell(K) = 1 - g$. Therefore,
  \[
    d (\COm{X}{\cP}{G}) \geq \deg G + 2 - 2g = \dgop.  
  \]
\end{remark}

Several floor bounds can be deduced from Theorem~\ref{thm:ABZ_bound},
such as:

\begin{theorem}[Lundell, McCullough {\cite[Th.~3]{lundell2006jpaa}}]
  \label{thm:LM_bound}
  Let $G = A+B+Z$ where $Z \geq 0$, the support of $Z$ is disjoint
  from $\cP$, $L(A+Z) = L(A)$ and $L(B+Z) = L(B)$. Then,
  \[
    d(\COm{X}{\cP}{G}) \geq d_{{\rm LM}} \eqdef \deg G +2-2g + \deg Z
    = \dgop + \deg Z.
  \]
\end{theorem}

\begin{proof}
  It is a consequence of Theorem~\ref{thm:ABZ_bound}. See
  \cite[Cor.~2.5]{duursma2011jpaa}.
\end{proof}

\begin{remark}
  The case $A = B$ was previously proved by Maharaj, Matthews and Pirsic
  in \cite[Th.~2.10]{maharaj2005jpaa}.
\end{remark}

\begin{theorem}[G\"uneri, Stichtenoth, Ta\c{s}kin
  {\cite[Th.~2.4]{guneri2009jpaa}}]\label{thm:GST_bound}
  Let $A,B,C,Z$ be divisors on $X$ satisfying
  \begin{enumerate}[(i)]
  \item The support of $A+B+C+Z$ is disjoint from $\cP$;
  \item $L(A) = L(A-Z)$;
  \item $L(B) = L(B+Z)$;
  \item $L(B) = L(C)$.
  \end{enumerate}
  If $G = A+B$, then
  \[
    d(\COm{X}{\cP}{G}) \geq d_{{\rm GST}} \eqdef \dgop + \deg Z + \ell
    (K-A) - \ell (K-G+C).
  \]
\end{theorem}

\begin{proof}
  It is proved in \cite[Cor.~2.6]{duursma2011jpaa} as another consequence
  of Theorem~\ref{thm:ABZ_bound}.
\end{proof}

\begin{example}
  This example is borrowed from \cite{duursma2011jpaa}.
  Calculations have been verified using {\sc Magma} \cite{bosma1997jsc}.
  Consider the {\em Suzuki curve} $X$ over $\F_8$ defined by the
  affine equation
  \[
    y^8+y = x^2(x^8+x).
  \]
  This curve is known to have genus $14$ and 65 rational points. We
  set $P$ and $Q$ to be the places above the points of respective
  homogeneous coordinates $(0:1:0)$ and $(0:0:1)$.  Let $\cP$ contain
  all the rational points of $X$ but $P, Q$ and set $G = 22P+6Q$. The
  code $\COm{X}{\cP}{G}$ has length $63$ and dimension 48. According
  to Goppa bound, its minimum distance satisfies
  \[
    d \geq \dgop = 2.
  \]
  If we set $A = 16P$, $B = 5P+4Q$ and $Z = P+2Q$,
  a calculation gives $\ell (A) = \ell (A+Z) = 6$
  and $\ell (B) = \ell (B+Z) = 1$. Then,
  Lundell--McCullough bound may also be applied
  and yield
  \[
    d \geq d_{{\rm LM}} = 5.
  \]
  Next, taking $A = 14P+2Q$, $B = 8P+4Q$, $C = 8P$ and $Z = 2Q$,
  one can check that the conditions of Theorem~\ref{thm:GST_bound}
  are satisfied and that
  \[
    d \geq d_{{\rm GST}} = 6.
  \]
  Actually, Theorem~\ref{thm:ABZ_bound} permits also to give this
  lower bound $d \geq 6$ using $A = 14P$, $B = 8P$ and $Z = 6Q$.
\end{example}

\subsection{Order bounds}\label{subsec:order}
In \cite{feng1993it}, Feng and Rao propose a new decoding algorithm
for algebraic geometry codes associated to a divisor $G$ supported by
a single rational point. This algorithm, further discussed in
Section~\ref{sec:decoding}, permits to correct errors up to half the
designed distance. For this sake, they introduced a new manner to
bound from below the minimum distance. The obtained lower bound turned
out to be always at least as good as Goppa designed distance.  Their
approach is at the origin of the so--called {\em order bounds}.

Feng and Rao bound and the corresponding algorithm applied only on
codes of the form $\COm{X}{\cP}{G}$, where the divisor $G$ is
supported by one point. That is to say, $G = rP$ for some rational
point $P$ and some positive integer $r$. Next, some generalisation to
arbitrary divisors appeared. One of the most general order bound is
due to Beelen \cite{beelen2007ffa}.
Before stating it, we need to introduce
some notation and definitions.

\begin{definition}\label{def:non-gaps}
  Given a divisor $F$ on a curve $X$ and a point $P$, the {\em
    non--gaps semi--group} of $F$ at $P$ denoted by $\nu(F, P)$ is
  defined as
  \[
    \nu(F,P) \eqdef \{j \in \Z ~|~ L(F+(j-1)P) \neq L(F+jP)\}.
  \]
\end{definition}

\begin{remark}
When $F = 0$ we find the classical notion of {\em Weierstrass gaps}.
\end{remark}

\begin{lemma}
  The semi--group $\nu(F, P)$ satisfies
  the following conditions.
  \begin{enumerate}[(i)]
  \item\label{item:interval_F-nongaps}
    $\nu (F, P) \subseteq \{n\in \Z ~|~ n \geq \deg F\}$;
  \item\label{item:F-gaps} the set
    $\{n \in \Z ~|~ n \geq \deg F\} \setminus \nu(F, P)$ is finite of
    cardinality $g$ and contained in
    $\{ -\deg F,\dots, -\deg F + 2g - 1 \}$ and usually referred to
    as the set of {\em $F$--gaps} at $P$.
  \end{enumerate}
\end{lemma}

\begin{proof}
  See for instance \cite[Rem.~2]{beelen2007ffa} or
  \cite[Rem.~3.2]{kirfel1995it}.
\end{proof}

Now, let $F_1, F_2, G$ be divisors on a curve $X$ such that
$F_1+F_2 = G$.  We aim at giving a lower bound for the minimum
distance of $\COm{X}{\cP}{G}$. Fix a rational point $P$ outside $\cP$
and denote respectively by $(\mu_i)_{i \in \N}, (\nu_i)_{i\in \N}$ and
$(\rho_i)_{i \in \Z}$ the $F_1$-, $F_2$- and $G$-non-gap sequences at
$P$.  In addition, for any $r \geq 0$, set
\begin{equation}\label{eq:nr}
  n_r \eqdef \Big|\left\{(i, j) \in \N^2 ~|~ \mu_i + \nu_j = \rho_{r+1} \right\}\Big|.
\end{equation}

\begin{proposition}
  Let $r$ be a positive integer. The minimum weight of a codeword in
  $\COm{X}{\cP}{G + rP} \setminus \COm{X}{\cP}{G+(r+1)P}$
  is bounded from below by $n_r$.
\end{proposition}

\begin{proof}
  See {\cite[Prop.~4]{beelen2007ffa}}.
\end{proof}

A direct consequence of this result is:

\begin{theorem}\label{thm:order_bound}
  Let $F_1, F_2, G$ be divisors on $X$ with $G = F_1 + F_2$, $\cP$ be
  an ordered set of rational points of $X$ and $P$ be a fixed rational point
  outside $\cP$. Then the minimum distance $d$ of the code
  $\COm{X}{\cP}{G}$ satisfies
  \[
    d \geq \dorder \eqdef \min_{r\geq 0} \{n_r\},
  \]
  where $n_r$ is defined in \eqref{eq:nr}.
\end{theorem}

\begin{proposition}
  The lower bound given in Theorem~\ref{thm:order_bound} is
  at least as good as Goppa bound:
  \[
    \dorder \geq \dgop.
  \]
\end{proposition}

\begin{proof}
  See \cite[Prop.~10]{beelen2007ffa}.
\end{proof}

\begin{remark}
  Of course, the previous result may be refined in order to get
  the minimum distance of a code of the form $\COm{X}{\cP}{G + sP}$
  for some positive $s$. In particular, Feng and Rao's original
  approach \cite{feng1993it} can be interpreted by choosing
  $F_1 = F_2 = G = 0$ in view to estimate the minimum distance
  (and decode) a code of the form $\COm{X}{\cP}{mP}$
  for some positive integer $m$.
\end{remark}

\begin{remark}
  In \cite{beelen2007ffa}, Beelen proposes a slightly more general
  statement in which, instead of considering a filtration of the
  form
  \[
    \COm{X}{\cP}{G} \supseteq \cdots \supseteq \COm{X}{\cP}{G + iP} \supseteq
    \COm{X}{\cP}{G+(i+1)P} \supseteq \cdots
  \]
  associated to a single point $P$, he considers a filtration
  associated to various points $Q_1, \dots, Q_s$ where two consecutive
  terms of the filtration are of the form
  \[
    \COm{X}{\cP}{G_i} \supseteq \COm{X}{\cP}{G_{i+1}}
  \]
  where $G_{i+1} - G_i \in \{Q_1, \dots, Q_s\}$.
  The choice of the optimal path is obtained by a kind of tree exploration
  using some classical backtracking tricks.
\end{remark}

\begin{example}
  Here we refer to \cite[Ex.~8]{beelen2007ffa} for an application of
  this bound to the Klein quartic: the genus 3 curve of affine equation
  \[
    x^3y + y^3 + x = 0.
  \]
\end{example}

\subsection{Further bounds} The literature provides further more
technical and involved bounds that are not discussed in the present
\ifx\localversion\undefined {chapter}
\else {article}
\fi
. The best reference surveying them and explaining in a clear
manner how these bounds are related is due to Duursma, Kirov and Park
\cite{duursma2011jpaa}.

\subsection{Geometric bounds for codes from embedded curves.}
To conclude this section, let us notice that all the previous bounds
arise from the intrinsic geometry of the curve and independently from
any particular embedding. In another direction, bounds deduced from the
geometry of the ambient space for a given embedding have been studied.
In \cite{couvreur2012jalg}, the following statement is proved.

\begin{theorem}\label{thm:geometric_bound}
  Let $X \subseteq \P^N$ be a smooth curve which is a complete
  intersection in $\P^N$. Let $m \geq 2$ and $G_m$ be a divisor
  obtained from the intersection of $X$ and a hypersurface of degree
  $m$ (the points being counted with the intersection
  multiplicities). Finally, let $\cP = (P_1, \dots, P_n)$ be an
  ordered $n$--tuple of rational points of $X$. Then, the minimum
  distance $d$ of the code $\COm{X}{\cP}{G_m}$ satisfies:
  \begin{enumerate}[(i)]
    \item\label{item:colin} $d = m+2$ iff $m+2$ of the $P_i$'s are collinear;
    \item\label{item:conic} $d = 2m+2$ iff (\ref{item:colin}) is not
      satisfied and $2m+2$ of the $P_i$'s lie on a plane conic;
    \item $d = 3m$ iff both (\ref{item:colin}) and (\ref{item:conic}) are
      not satisfied and $3m$ of the $P_i$'s lie at the intersection of
      a plane cubic and a plane curve of degree $m$ with no common components;
    \item $d > 3m$ iff none of the previous configurations exist.
  \end{enumerate}
\end{theorem}

\begin{remark}
  Actually, Theorem~\ref{thm:geometric_bound}
  applies not only to codes from embedded curves
  but to duals of $\mathcal C_L$ codes from arbitrary dimensional
  complete intersections in a projective space. 
\end{remark}

In \cite[Th.~4.1]{couvreur2012jalg}, it is proved that, for plane
curves, Theorem~\ref{thm:geometric_bound} provides a nontrivial lower bound
even in cases where Goppa bound is negative and hence irrelevant.

\begin{example}
  This example is borrowed from \cite[Ex.~4.3]{couvreur2012jalg}.
  Consider the finite field $\F_{64}$ and the curve $X$ of homogeneous equation
\begin{align*}
w^{24} x^{11} &+ w^{44} x^6 y^2 z^3 + w^{24} x^5 y z^5 + w^{20} x^4 y^6 z + w^{33} x^2 z^9 +\\
     &w^{46} x y^5 z^5 + w^{46} x z^{10} + w^{39} y^{11} + w^{30} y^2 z^9 = 0,
\end{align*}
where $w$ is a primitive element of $\F_{64}$ over $\F_2$ with minimal
polynomial $x^6 + x^4 + x^3 + x + 1$.  This curve has genus $45$ and
$80$ rational points in the affine chart $\{z \neq 0\}$ and $1$
rational point \textit{at infinity}.
We set $\cP$ to be the whole set of affine points with some arbitrary ordering.
The divisor $G_m$ has degree
$11m$, hence Goppa designed distance applied to
$\COm{X}{\cP}{G_m}$ gives
\[
  \dgop = 11 m - 88.
\]
which is negative for $m\leq 8$.  Using
Theorem~\ref{thm:geometric_bound} together with a computer aided
calculation, we prove that the codes
$\COm{X}{D_{\cP}}{G_m}$ for $m\in \{1, \ldots, 8\}$ are respectively
of the form: $[80,77,\geq 3]$, $[80,74,\geq 4],$ $[80,70,\geq 5],$
$[80,65,\geq 6]$, $[80,59,\geq 7]$, $[80,52,\geq 8]$, $[80,46,\geq 9]$
and $[80,35,\geq 10]$.
\end{example}

 \section{Decoding algorithms}\label{sec:decoding}
If AG codes appeared at the very early 80's, the first decoding
algorithm has been proposed in 1989 by Justesen et al. in
\cite{justesen89it} for codes from smooth plane curves. Then,
Skorobogatov and \vladut{} gave a generalisation to arbitrary AG codes
in \cite{skorobogatov90it}. Further, Pellikaan and independently
K\"otter gave an abstract version of the algorithm expurgated from
algebraic geometry \cite{pellikaan88,pellikaan92dm,koetter92acct}.
All these algorithms permitted to correct errors up to half the
designed distance minus some defect proportional to the curve's
genus. In the 90's many works have consisted in trying to fill this
gap \cite{skorobogatov90it,
  duursma1993it,pellikaan1989it,vladut1990it,ehrhard1993it}.

In the late 90's, after Sudan's breakthrough \cite{sudan97jcomp}
showing that, at the cost of possibly returning a list of codewords
instead of a single word, it was possible to correct errors on
Reed--Solomon codes beyond half the designed distance, a generalisation
of Sudan's algorithm is proposed by Shokrollahi and Wasserman in
\cite{shokrollahi99it}. Further, Guruswami and Sudan gave an improved
algorithm correcting errors up to the so--called Johnson bound
\cite{guruswami99it}.

For a detailed survey on decoding algorithms up to half the designed distance
see \cite{hoeholdt95it}. For a more recent survey including list decoding see
\cite{beelen08}.

\begin{remark}
  In the sequel, we suppose that for any divisor $A$ on the studied
  curve $X$, bases of the spaces $L(A)$ and
  $\Omega(A)$ can be efficiently computed. It is worth noting that the
  effective computation of Riemann--Roch spaces, is a difficult algorithmic
  problem of deep interest but which requires an independent
  treatment. For references on this topic, we refer the reader for
  instance to \cite{lebrigand1988smf,hess2002jsc}.
\end{remark}

\subsection{Decoding below half the
  designed distance}\label{subsec:unique_dec}
\subsubsection{The basic algorithm}\label{subsec:basic}
We first present what is sometimes referred to as the {\em basic
  algorithm} in the literature. However, compared to the usual
approach we present this algorithm for a $\mathcal C_L$ code instead
of a $\mathcal C_{\Omega}$ one. For this reason and despite it does
not represent a high difficulty, we detail the proofs in the sequel
instead of referring to the literature. This algorithm permits to
correct any error pattern of weight $t$ satisfying
\[
  t \leq \frac{\dgop^* - 1}{2} - \frac g 2 \cdot
\]

Let $X$ be a curve of genus $g$, let $\cP = (P_1, \dots, P_n)$
an ordered $n$--tuple of distinct rational points of $X$ and $G$
a divisor whose support avoids the $P_i$'s.
Let $\CC$ be the algebraic geometry code $\CL{X}{\cP}{G}$, let
$\cv \in \CC$ be a codeword and $\ev \in \Fq^n$ of Hamming weight
$\wt{\ev} = w \leq t$ for some positive integer $t$ and consider the
{\em received word}
\[
  \yv \eqdef \cv + \ev.
\]
Recall that $\dgop^* = n - \deg G$ denotes the designed distance of
$\CC$ (see Theorem~\ref{thm:basics}).

By definition, there exists a function $f \in L(G)$ such that
$\cv = (f(P_1), \dots, f(P_n))$. In addition, one denotes by
$\{i_1, \dots, i_w\} \in \{1, \dots, n\}$ the {\em support} of $\ev$,
that is to say, the set of indexes corresponding to the nonzero
entries of $\ev$ i.e. the positions at which errors occurred.

A decoding algorithm correcting $t$ errors takes as inputs
$(\CC, \yv)$ and returns either $\cv$ (or equivalently $\ev$) if
$w_H(\ev) \leq t$ or ``?'' if $w_H(\ev)$ is too large.  The first
algorithms in the literature \cite{justesen89it,skorobogatov90it} rest
on the calculation of an {\em error locating function}. For this sake,
one introduces an extra divisor $F$ whose support avoids $\cP$ and
whose additional properties are to be decided further. The point is to
compute a nonzero function $\lambda \in L(F)$ ``locating the error
positions'' i.e. such that
$\lambda(P_{i_1}) = \cdots = \lambda (P_{i_w}) = 0$.  Once such a
function is computed, its zero locus provides a subset of indexes
$J \subseteq \{1, \dots, n\}$ such that
$\{i_1, \dots, i_w\} \subseteq J$.  If this set $J$ is small enough,
then its knowledge permits to decode by solving a linear system as
suggested by the following statement. This is a classical result of
coding theory: once the errors are located, decoding
reduces to correct erasures.

\begin{proposition}\label{prop:erasure_decoding}
  Let $\Hm$ be a parity--check matrix for $\CC$
  and $J \subseteq \{1, \dots, n\}$ such that $|J| < d(\CC)$ and which contains
  the support of $\ev$. Then $\ev$ is the unique solution of the system:
  \begin{equation}\label{eq:syst_erasures}
    \left\{
      \begin{array}{cccl}
        \Hm \cdot \ev^\top & = & \Hm \cdot \yv^\top & \\
        e_i& = & 0, & \forall i \in \{1, \dots, n\} \setminus J.
      \end{array}
    \right.
  \end{equation}
\end{proposition}

\begin{proof}
  Clearly $\ev$ is solution. If $\ev'$ is another solution, then
  $\ev - \ev' \in \ker \Hm = \CC$ and has support included in
  $J$. This word has weight less than the code's minimum distance and
  hence is $0$.
\end{proof}

For a function
$\lambda \in L(F)$ vanishing at the error positions, the fundamental
observation is that,
\begin{equation}\label{eq:WB_eq}
  \forall i \in \{1, \dots, n\}, \quad \lambda(P_i)y_i = \lambda(P_i)f(P_i).
\end{equation}
Indeed, either there is no error at position $i$,
i.e. $e_i = 0$ and hence $y_i = f(P_i)$, or there is an error but in
this situation $\lambda(P_i) = 0$, making the above equality obviously
true. Next, since $\lambda f \in L(G+F)$, we deduce that
$(\lambda(P_1)y_1, \dots, \lambda(P_n)y_n) \in \CL{X}{\cP}{G+F}$.
This motivates to introduce the following space.
\begin{equation}\label{eq:def_Ky}
  K_{\yv} \eqdef \{ \lambda \in L(F) ~|~ (\lambda(P_1)y_1, \dots,
  \lambda(P_n)y_n) \in \CL{X}{\cP}{G+F} \}.
\end{equation}

\begin{lemma}\label{lem:incl_Ky}
  Let $D_{\ev}\eqdef P_{i_1}+ \cdots + P_{i_w}$ be the sum of points
  at which an error occurs. Then,
  \[\CL{X}{\cP}{F-D_{\ev}} \subseteq K_{\yv}.\]
\end{lemma}

\begin{proof}
  If $\lambda \in L(F-D_{\ev})$, then $\lambda$ vanishes at the error points
  and the result is a consequence of~(\ref{eq:WB_eq}).
\end{proof}

\begin{proposition}\label{prop:key_of_basic}
  If
  $t \leq \dgop^* - \deg F -1= n - \deg(G-F)-1$, then
  \[
    K_{\yv} = L(F-D_{\ev}).
  \]
\end{proposition}

\begin{proof}
  Inclusion $\supseteq$ is given by
  Lemma~\ref{lem:incl_Ky}. Conversely, if
  $(\lambda(P_1)y_1, \dots, \lambda(P_n)y_n) \in \CL{X}{\cP}{G+F}$,
  then, since
  $(\lambda(P_1)f(P_1), \dots, \lambda(P_n)f(P_n)) \in
  \CL{X}{\cP}{G+F}$, this entails that
  $\uv \eqdef (\lambda(P_1)e_1, \dots, \lambda(P_n)e_n)$ also lies in
  this code. In addition, $w_H(\uv) \leq w_H(\ev) \leq t$. On the
  other hand, from Theorem~\ref{thm:basics},
  $d(\CL{X}{\cP}{G+F}) \geq n - \deg(G+F)$.  Therefore, $w_H(\uv)$ is
  less than the code's minimum distance, and hence
  $\uv = (\lambda(P_1)e_1, \dots, \lambda(P_n)e_n)= 0$. Thus, $\lambda$
  vanishes at any position where $\ev$ does not. Hence
  $\lambda \in L(F-D_{\ev})$.
\end{proof}

The previous statements provide the necessary material to describe a
decoding algorithm for the code $\CL{X}{\cP}{G}$.
Mainly, the algorithm consists in
\begin{enumerate}
\item Computing $K_{\yv}$;
\item taking a nonzero function $\lambda$ in it;
\item compute the zeroes of $\lambda$ among the $P_i$'s, which,
  hopefully, should provide a set localizing the errors;
\item Find $\ev$ by solving a linear system using
  Proposition~\ref{prop:erasure_decoding}.
\end{enumerate}
More precisely the pseudo-code of the complete procedure is given in
Algorithm~\ref{algo:basic}.

\begin{algorithm}
  \caption{Basic decoding algorithm}\label{algo:basic}
  \begin{description}
  \item[{\bf Inputs}] A code $\CC=\CL{X}{\cP}{G}$, a vector $\yv \in \Fq^n$
    and an integer $t > 0$;
  \item[{\bf Output}] A vector $\ev$ such that $w_H(\ev) \leq t$ and
    $\yv - \ev \in \CC$ if exists, ``?'' else.
  \end{description}
  \begin{algorithmic}[1]
    \State Compute the space $K_{\yv}$ defined in \eqref{eq:def_Ky}.
    \If{$K_{\yv} = \{0\}$}
    \State \Return ?
    \Else
    \State{Take $\lambda \in K_{\yv} \setminus \{0\}$;}
    \State{Compute $i_1, \dots, i_s \in \{1, \dots, n\}$ such that
      $\lambda (P_{i_j}) = 0$;}
    \State {Let $S$ be a affine space of solutions of the system
      \eqref{eq:syst_erasures}}
    \If{$S = \emptyset$ or $|S| > 2$}
    \State \Return ``?''
    \Else
    \State{Return the unique solution $\ev$ of \eqref{eq:syst_erasures}}
    \EndIf
    \EndIf
  \end{algorithmic}
\end{algorithm}

\begin{theorem}
  If $t \leq \frac{\dgop^* - 1}{2} - \frac g 2$ and $\deg F  = t + g$,
  then Algorithm~\ref{algo:basic} is correct and returns the good solution
  in $O(n^\omega)$ operations in $\Fq$, where $\omega$ is the complexity
  exponent of linear algebraic operations (in particular $\omega \leq 3$).
\end{theorem}

\begin{proof}
  If $\deg F \geq t+g$, then $\deg F-D_{\ev} \geq g$ and hence, by
  Riemann--Roch theorem, $\ell (F- D_{\ev}) > 0$. In particular, from
  Lemma~\ref{lem:incl_Ky}, $K_{\yv} \neq 0$.
  Next, by assumption on $t$ and $\deg F$, we get
  \begin{equation}\label{eq:bound_t}
    2t + g \leq \dgop^* - 1 \quad \Longrightarrow \quad
    t \leq \dgop^* - \deg F - 1
  \end{equation}
  which, from Proposition~\ref{prop:key_of_basic}, yields the equality
  $K_{\yv} = L(F-D_{\ev})$.  Therefore, one can compute
  $L(F - D_{\ev})$. Take any nonzero function in this space, there
  remains to prove that the conditions of
  Proposition~\ref{prop:erasure_decoding} are satisfied. That is to
  say, that the zero locus of $\lambda$ in $\{P_1, \dots, P_n\}$ is
  not too large. But, from~(\ref{eq:bound_t}), $\lambda \in L(F)$,
  then its zero divisor $(\lambda)_0$ has degree at most $\deg F$ and
  since $0 \leq t \leq \dgop^* - \deg F - 1$, yielding immediately
  $\deg (\lambda)_0 \leq \deg F < \dgop^*$ which proves that the
  resolution of System~\eqref{eq:syst_erasures} will provide $\ev$
  as the unique possible solution.

  About the complexity, the computation of $K_{\yv}$ such as
  the resolution of the system \eqref{eq:syst_erasures} are nothing
  but the resolution of linear systems with $O(n)$ equations and
  $O(n)$ unknowns. The other operations in the algorithm are negligible.
\end{proof}

\subsubsection{Getting rid of algebraic geometry, error correcting pairs}
As observed by Pellikaan \cite{pellikaan88,pellikaan92dm} and
independently by K\"otter \cite{koetter92acct}, the basic algorithm
can be defined on the level of codes without involving any algebraic
geometric notion. To do that, observe that
\[
  \av \eqdef (\lambda(P_1), \dots, \lambda(P_n)) \in \CL{X}{\cP}{F}.
\]
Next, on the level of codes, we searched a vector $\av$
such that $\av \star \yv \in \CL{X}{\cP}{G+F}$.
Thus, the space $K_{\yv}$ can be redefined on the level of codes
as the space:
\[
  \hat{K}_{\yv} \eqdef \left\{
  \av \in \CL{X}{\cP}{F} ~|~ \forall \bv \in
  \CL{X}{\cP}{G+F}^\perp,\ \eucl{\av \star \yv}{\bv} = 0
  \right\}.
\]
Next, one shows easily that the adjunction property noticed
in Lemma~\ref{lem:adjunction} yields an equivalent reformulation of the above
definition as:
\[
  \hat{K}_{\yv} = \left\{
  \av \in \CL{X}{\cP}{F} ~|~ \forall \bv \in
  \CL{X}{\cP}{G+F}^\perp,\ \eucl{\av \star \bv}{\yv} = 0
  \right\}.
\]
If we set
\[
  \AC = \CL{X}{\cP}{F} \quad {\rm and} \quad \BC \eqdef
  \CL{X}{\cP}{G+F}^\perp = \COm{X}{\cP}{G+F}
\]
then, one can prove in particular that
\[
  \AC \star \BC \subseteq \COm{X}{\cP}{G} = \CC^\perp
\]
and this
material turns out to be sufficient to provide a decoding algorithm.

\begin{definition}[Error correcting pair]
  \label{def:ECP}
  Let $\CC \subseteq \Fq^n$ be a code and $t$ be a positive integer. A
  pair of codes $(\AC, \BC)$ is said to be a {\em $t$--Error Correcting
    Pair} (ECP) for $\CC$ if it satisfies the following conditions:
  \begin{enumerate}[(ECP1)]
  \item $\AC \star \BC \subseteq \CC^\perp$;
  \item $\dim \AC > t$;
  \item $\dim \BC^\perp > t$;
  \item $d(\AC) + d(\CC) > n$.
  \end{enumerate}
\end{definition}

\begin{theorem}\label{thm:ECP_correct}
  Let $\CC \subseteq \Fq^n$ be a code equipped with a $t$--error
  correcting pair $(\AC, \BC)$ with $t \leq \frac{d(\CC) -
    1}{2}$. Then, there is a decoding algorithm for $\CC$ correcting
  any error pattern of weight less than or equal to $t$ in $O(n^{\omega})$
  operations in $\Fq$.
\end{theorem}

\begin{proof}
  See for instance \cite[Th.~2.14]{pellikaan92dm}.
\end{proof}

\begin{remark}
  Using Theorem~\ref{thm:basics} and
  Corollary~\ref{cor:Goppa_bound_COmega}, one can easily observe that
  the codes $\AC = \CL{X}{\cP}{F}$ and $\BC = \COm{X}{\cP}{G+F}$
  satisfy the conditions of Definition~\ref{def:ECP}.  Next, the
  algorithm mentioned in Theorem~\ref{thm:ECP_correct} is nothing but
  Algorithm~\ref{algo:basic} described at the level of codes: in
  particular, replacing elements $\lambda \in L(F)$ by the
  corresponding evaluation vectors
  $(\lambda(P_1), \dots, \lambda(P_n)) \in \CL{X}{\cP}{F}$.
\end{remark}

\begin{remark}
  As noticed earlier, the tradition in the literature is to present
  the basic algorithm such as error correcting pairs to decode
  $\COm{X}{\cP}{G}$. This code benefits of a very similar
  decoding algorithm correcting up to $\frac{\dgop - 1}{2} - \frac g 2$
  errors using the error correcting pair
  \[\AC = \CL{X}{\cP}{F} \quad {\rm and} \quad \BC =
    \CL{X}{\cP}{G-F}
  \]
  with an extra divisor $F$ with $\deg F \geq t+g$.
\end{remark}

\begin{remark}
  Note that the defect $\frac{g}{2}$ corresponds actually to a
  {\em worst case}; it is observed for instance in
  \cite[Rem.~4.6]{hoeholdt95it} that this algorithm can actually
  correct a uniformly random error pattern of weight
  $t = \lfloor \frac{\dgop^* - 1}{2} \rfloor$ with a failure
  probability of $O(1/q)$.
\end{remark}

\subsubsection{Reducing the gap to half the designed distance}
After the basic algorithm, several attempts appeared in the literature
to reduce the gap between the decoding radius of this algorithm and half
the designed distance:
\begin{itemize}
\item The {\em modified algorithm} \cite{skorobogatov90it} and the
  {\em extended modified algorithm} \cite{duursma1993it} combine the
  basic algorithm with an iterative search of a relevant choice of the
  extra divisor $F$. These algorithms
  permit to reduce the gap $\frac g 2$ to about $\frac g 4$
  (see \cite[Rem.~4.11 \& 4.15]{hoeholdt95it}).
\item In \cite{pellikaan1989it,vladut1990it} the question of the
  existence of an extra divisor $F$ for which the basic algorithm
  corrects up to half the designed distance is discussed. It is
  in particular proved that such an $F$ exists and can be chosen
  in a set of $O(n)$ divisors as soon as
  $q \geq 37$.
\item Finally, an iterative approach to find an $F$ achieving
  half the designed distance is proposed by Ehrhard
  \cite{ehrhard1993it}.
\end{itemize}

All the above contributions are summarized with further details in
\cite[\S~4 to 7]{hoeholdt95it}.

\subsubsection{Decoding up to half the designed distance, Feng--Rao
  algorithm and error correcting arrays}\label{subsec:Feng_Rao}
The first success in getting a decoding algorithm correcting exactly
up to half the designed distance is due to Feng and Rao
\cite{feng1993it}. Their algorithm consists in using a filtration
of codes and to iterate {\em coset decoding}. Note that Feng--Rao's
original algorithm applied to codes $\COm{X}{\cP}{rP}$ for some
rational point $P$ of $X$. Later, their approach has been extended
to arbitrary divisors \cite{pellikaan1993eurocode, kirfel1995it}.

Similarly to the basic algorithm which lead to the abstract and
algebraic geometry--free formulation of error correcting pairs,
Feng--Rao's approach led to a purely coding theoretic formulation
called {\em error correcting arrays}. We present this approach in the
sequel using the notation of \cite{pellikaan1993eurocode} together
with refinements appearing in \cite{couvreur2017it}.

\begin{definition}[Array of codes]
  An {\em array}
  of codes for $\CC$ is a triple of sequences of codes
  ${(\AC_i)}_{0 \leq i \leq u}, {(\BC_j)}_{0 \leq j \leq v},
  {(\CC_r)}_{w \leq r \leq n}$ such that
  \begin{enumerate}[({A}1)]
  \item $\CC = \CC_w$
  \item $\forall i, j ,r,\ \dim \AC_i = i$, $\dim \BC_j = j$ and
    $\dim \CC_r = n-r$;
  \item The sequences ${(\AC_i)}_i$ and ${(\BC_j)}_j$ are increasing,
    while ${(\CC_r)}_r$ is decreasing;
  \item\label{eq} $\forall i, j$, we define $\hat r(i,j)$ to be the least
    integer $r$ such that $\AC_i \star \BC_j \subseteq \CC_r^\perp$.
    For any $i, j \geq 1$ if $\av \in \AC_i \setminus \AC_{i-1}$ and
    $\bv \in \BC_j \setminus \BC_{j-1}$ and $r=\hat r(i,j) > w$ then
    $\av \star \bv \in \CC_{r}^{\perp} \setminus \CC_{r-1}^{\perp}$.
  \end{enumerate}
\end{definition}

Note that the function $\hat r(i, j)$ is increasing in the variables $i$
and $j$ but not necessarily strictly increasing. This motivates the
following definition of {\em well--behaving pair}, a terminology borrowed from
that of {\em well--behaving sequences} \cite{geil2013ffa}:

\begin{definition}[Well--behaving pair]
  A pair $(i, j) \in \{1,\dots, u\} \times \{1, \dots, v\}$
  is said to be {\em well-behaving} (WB) if for any pair $(i',j')$
  such that $i' \leq i$, $j' \leq j$ and $(i',j') \neq (i,j)$
  we have $\hat r(i',j') < \hat r(i,j)$.
  Next, we define for any $r \in \{w, \dots, n-1\}$ the quantity:
  \[\hat n_r \eqdef \big|\left\{
        (i,j) \in \{1,\dots, u\}\times \{1, \dots, v\} ~|~
        (i,j)\ {\rm is\  WB\ and\ }
        \hat r(i,j) = r+1
      \right\}\big|.\]
\end{definition}

\begin{theorem}\label{thm:FR_bound}
  For any array of codes $({(\AC_i)}_i, {(\BC_j)}_j, {(\CC_r)}_r)$, we
  have
  \[
    \forall w \leq r \leq n-1,\quad d(\CC_r) \geq \min\{\hat n_{r'} ~|~ r \leq r' \leq n-1\}.
  \]
\end{theorem}

\begin{proof}
  See \cite[Th.~4.2]{pellikaan1993eurocode}.
\end{proof}

\begin{remark}
  Of course, Theorem~\ref{thm:FR_bound} is almost the same as
  Theorem~\ref{thm:order_bound}. The slight difference lies in the
  quantity $\hat n_r$, which is very close to the objects introduced
  at \eqref{eq:nr} in \S~\ref{subsec:order}. The difference relies on
  the fact that $n_r$ is defined on the level of function algebras
  while $\hat n_r$ is defined on the level of codes. On function
  algebras over curves, valuations assert a strict
  growth of the function
  $r(i,j)$
  that could be defined in this
  context. When dealing on the level of codes, we should restrict to a
  subset of pairs (the well--behaving ones) to keep this strict
  increasing property.

  This requirement of strict increasing is necessary to prove
  Theorem~\ref{thm:FR_bound}. This bound on the minimum distance is
  obtained by relating the Hamming weight of some codeword with the rank
  of a given matrix. The well--behaving pairs provide the position of
  some pivots in this matrix. Thus, their number give a lower bound
  for the rank. Of course two pivots cannot lie on the same row or
  column, hence the requirement of strict decreasing.  Note that in
  \cite[Th.~4.2]{pellikaan1993eurocode}, the strict increasing
  requirement is lacking, which makes the proof not completely correct. Replacing
  general pairs by well--behaving ones, as it is done in
  \cite{couvreur2017it}, fixes the proof.
\end{remark}

\begin{definition}[Error--correcting array]\label{def:ECA}
  Let $\CC \subseteq \Fq^n$ be a code. An array of codes
  $({(\AC_i)}_{1 \leq i \leq u}, {(\BC_j)}_{1 \leq j \leq v},
  {(\CC_r)}_{w \leq r \leq n})$ is said to be a {\em $t$--error
    correcting array for $\CC$} if
  \begin{enumerate}[(i)]
  \item $\CC_w = \CC$;
  \item $t \leq \frac {d(\CC_w)-1}{2}$.
  \end{enumerate}
\end{definition}

\begin{theorem}
  If a code $\CC$ has a $t$--error correcting array, then it has a decoding
  algorithm correcting up to $t$ errors in $O (n^{\omega})$ operations in $\Fq$.
\end{theorem}

\begin{proof}
  See \cite[Th.~4.4]{pellikaan1993eurocode}
\end{proof}

The spirit of the algorithm is to start from a received vector
$\yv = \cv + \ev$ with $\cv \in \CC = \CC_w$ and $w_H(\ev) \leq t$,
then use the array to ``guess an additional syndrome'' using a process
called {\em majority voting} and hence to transform this decoding
problem into
\[
  \yv_{w+1} = \cv_{w+1} + \ev
\]
where $\cv_{w+1} \in \CC_{w+1}$ and {\bf the error is
  unchanged}. Applying this process iteratively until
$\cv_{n} = 0$ yields the error.

\begin{remark}\label{rem:variant_of_FR}
  Actually, as noticed in \cite{pellikaan1993eurocode}, the process
  may be stopped before reaching $\cv_n$. Indeed, after $r \geq g$ iterations,
  one obtains a new decoding problem $\yv_r = \cv_r + \ev$ where
  $\CC_r$ benefits from a $t$--error correcting pair, one can at this
  step switch to the error correcting pair algorithm, which corrects up to
  $\frac{d(\CC_r)-1}{2} - \frac{g}{2} \geq \frac{d(\CC_w)-1}{2}$ errors
  and get $\ev$.
\end{remark}

The impact of
the previous constructions on algebraic geometry
codes is summarized in the following statement.

\begin{theorem}[{\cite[Th.~4.5]{pellikaan1993eurocode}}]
  Let $\CC = \COm{X}{\cP}{G}$ be a code on a curve $X$ of genus $g$
  such that $2g-2 < \deg G < n-g$. Then $\CC$ has a $t$--error correcting
  array with
  \[
    t = \frac{\dgop - 1}{2} \cdot
  \]  
\end{theorem}

\begin{remark}
  Of course, a similar result holds for evaluation codes by replacing
  $G$ by $K - G + D_{\cP}$.
\end{remark}

\begin{remark}
  A simple choice of error--correcting array for $\COm{X}{\cP}{G}$ can
  be obtained by choosing a rational point $P$ of $X$ and consider the
  following sequences to construct the $\AC_i$'s, the $\BC_j$'s and the
  $\CC_r$'s:
  \begin{align*}
    \cdots &\subseteq \CL{X}{\cP}{\mu_iP} \subseteq \CL{X}{\cP}{\mu_{i+1}P}
             \subseteq \cdots\\
    \cdots &\subseteq \CL{X}{\cP}{G+\nu_jP} \subseteq \CL{X}{\cP}{G+\nu_{j+1}P}
             \subseteq \cdots\\
    \cdots & \supseteq \COm{X}{\cP}{G+\nu_j P}\supseteq
             \COm{X}{\cP}{G+\nu_{j+1}P} \supseteq \cdots
  \end{align*}
  where ${(\mu_i)}_i$ and ${(\nu_j)}_j$ respectively denote the
  $0$-- and $G$--non gap sequences at $P$ (see Definition~\ref{def:non-gaps}).
\end{remark}

\subsection{List decoding, Guruswami Sudan algorithm}
The previous algorithms permit to correct up to half the designed
distance. Note that Shannon theory
essentially asserts that, for a random code, decoding up to the
minimum distance is almost always possible. A manner to fill the gap
between Hamming and Shannon's point of view is to use {\em list
  decoding}. That is to say, given a code $\CC \subseteq \Fq^n$, a
received word $\yv \in \Fq^n$ and a positive integer $t$, list
decoding consists in returning the whole list of codewords
$\cv \in \CC$ at distance less than or equal to $t$ from $\yv$.

Let $\CC = \CL{X}{\cP}{G}$, $\yv \in \Fq^n$ and $t$ be a positive
integer. Guruswami--Sudan algorithm may be regarded as a
generalisation of Algorithm~\ref{algo:basic} when reformulated as
follows.  In \S~\ref{subsec:basic}, we were looking for a function
$Q \in \Fq(X)[Y]$ of the form $Q = Q_0 + Q_1 Y$ with $Q_0 \in L(G+F)$
and $Q_1 \in L(F)$ such that the function $Q(f)$ is identically
$0$. In this reformulation, the function $Q_1$ is nothing but our
former error locating function $\lambda$.  Note that, regarded as a polynomial
in the variable $Y$, the polynomial $Q$ has degree $1$ and hence has a unique
root. Sudan's key idea is to consider a polynomial
of larger degree in $Y$ possibly providing a list of solutions instead
of a single one.

Fix a positive integer $\ell$ which will be our maximal list size, a
positive integer $s$ called the {\em multiplicity} and an extra
divisor $F$ whose support avoids the points of $\cP$ and whose
additional properties are to be decided later.  The algorithm is
divided in two main steps: {\em interpolation} and {\em root finding}
which are described in Algorithm~\ref{algo:GS}.

\begin{algorithm}
  \caption{Guruswami--Sudan algorithm for algebraic geometry codes}
  \label{algo:GS}
\begin{description}
\item[{\it Interpolation}] Compute a nonzero polynomial $Q \in \Fq(X)[Y]$
  of the form
  \[
    Q = Q_0 + Q_1 Y + \cdots + Q_\ell Y^\ell
  \]
  satisfying:
  \begin{enumerate}[(i)]
  \item For any $j \in \{0, \dots, \ell\}$, $Q_j \in L(F + (\ell - j)G)$;
  \item For any $i \in \{1, \dots, n\}$,
    the function $Q$ vanishes at $(P_i, y_i)$ with multiplicity at least $s$.
  \end{enumerate}
  \medskip
\item[{\it Root finding}] Compute the roots $f_1,\dots, f_m$
  ($m \leq \ell$) of $Q(Y)$ lying in $\Fq(X)$ and output the list of
  codewords of the form $(f_i(P_1), \dots, f_i(P_n))$ which are at
  distance at most $t$ from $\yv$.
\end{description}    
\end{algorithm}

\begin{remark}
  Geometrically speaking, the function $Q$ in Algorithm~\ref{algo:GS}
  can be interpreted as a rational function on the surface
  $X \times \P^1$.  In this context, saying that {\it $Q$ vanishes at
  $(P_i, y_i)$ with multiplicity at least $s$} means that
  $Q \in \mathfrak m_{(P_i, y_i)}^s$, where $\mathfrak m_{(P_i, y_i)}$
  denotes the maximal ideal of the local ring of the surface at the
  point $(P_i, y_i)$.  From a more computational point of view, given
  a local parameter $t \in \Fq(X)$ at $P_i$, the function $Q$ has a Taylor
  series expansion $\sum_{u,v} q_{uv}t^u (Y-y_i)^v$ and the vanishing
  requirement means that $q_{uv} = 0$ for any pair $(u,v)$ such that
  $u+v < s$.
\end{remark}

\begin{remark}
  In \cite[\S~6.3.4]{guruswami2005book}, a polynomial time algorithm
  to perform the root finding step is presented. This algorithm
  consists in a reduction of $Q$ modulo some place of $\Fq(X)$ of
  large enough degree, then a factorisation of the reduced polynomial
  using Berlekamp or Cantor--Zassenhaus algorithm is done, followed by
  a ``lifting'' of the roots in $\Fq(Y)$, under the assumption
  that they lie in a given Riemann--Roch space.
\end{remark}

\begin{theorem}\label{thm:GS_radius}
  Let
  \[
    t \leq n - \frac{n(s+1)}{2(\ell + 1)} - \frac{\ell \deg G}{2s} - \frac g s,
  \]
  then, for any extra divisor $F$ satisfying
  \begin{enumerate}[(a)]
  \item $\deg (F + \ell G) < s(n-t)$;
  \item
    $\deg (F + \ell G) > \frac{ns(s+1)}{2(\ell + 1)} + \frac{\ell \deg
      G}{2} + g - 1$,
  \end{enumerate}
  Guruswami--Sudan algorithm succeeds in returning in polynomial time
  the full list of codewords in $\CL{X}{\cP}{G}$ at distance less than
  or equal to $t$ from $\yv$.
\end{theorem}

\begin{proof}
  See for instance \cite[Lem.~2.3.2 \& 2.3.3]{beelen08}.
\end{proof}

\begin{remark}
  The case with no multiplicity, i.e. the case $s=1$ has been
  considered before Guruswami and Sudan by Sudan for Reed--Solomon
  codes \cite{sudan97jcomp} and Shokrollahi and Wasserman
  \cite{shokrollahi99it} for algebraic geometry codes.
\end{remark}

The following remark has been transmitted by Peter Beelen. We warmly
thank him for this help.

\begin{remark}
  When, $\ell = s = 1$, the algorithm should yield the basic
  algorithm.  However, Theorem~\ref{thm:GS_radius} yields a decoding
  radius of $\lfloor \frac{\dgop^* - 1}{2} - g \rfloor$: a defect $g$
  instead of $\frac g 2$. A further analysis shows that, whenever
  $g > 0$, the linear system whose solution is the polynomial $Q$ has
  never full--rank. On the other hand, the analysis of the decoding
  radius is done without taking this rank defect into account, leading
  to a slightly pessimistic estimate. A further analysis of the rank of the
  system in the general case might provide a slight improvement of
  Guruswami--Sudan decoding radius.
\end{remark}

 \section{Application to public-key cryptography: McEliece-type
  crypto\-sys\-tem}\label{sec:McEliece}
\subsection{History}
In \cite{berlekamp1978it} Berlekamp, McEliece
and van Tilborg proved that the following problem is NP-complete.

\bigbreak

\noindent \textbf{Problem.}
 {\em Let $\CC \subseteq \Fq^n$ be a code, $t \leq n$ be a positive integer,
  and $\yv \in \Fq^n$. Decide
  whether there exists $\cv \in \CC$ whose Hamming distance with
  $\yv$ is less than or equal to $t$.
 }
\bigbreak

Note that NP--completeness may only assert hardness in the worst case.
However, this problem is commonly believed by the community to be
``{difficult in average}''. By this, we mean that most of the instances
of the problem seem difficult.

This result motivated McEliece to propose a code
based encryption scheme whose security was relying on the hardness of
this problem \cite{mceliece1978dsn}.  Roughly speaking, McEliece
scheme can be described follows.
\begin{itemize}
\item The public key is a code $\CC$;
\item the secret key is a secret element related to $\CC$ and permitting
  to decode;
\item encryption consists in encoding the plain text and include errors
  in it;
\item decryption is decoding.
\end{itemize}

More formally, let $\mathcal{F}_{n,k}$ be a family of codes of fixed
length and dimension $[n, k]$.  Consider also a set $\mathcal S$ of
``secrets'' together with a surjective map
$\CC : \mathcal S \rightarrow \mathcal F_{n,k}$ such that for any
$s \in \mathcal S$, the code $\CC(s)$ benefits from a decoding
algorithm $\mathbf{Dec}(s)$ {\bf
  depending on $s$} and correcting up to $t$ errors.

\bigskip

\noindent {\bf McEliece encryption scheme:}

\smallskip
\begin{description}
\item[{\bf Key generation}] Draw a uniformly random element
  $s \in \mathcal S$:
  \begin{description}
\item[{\it Secret key}] the secret $s$;
\item[{\it Public key}] A pair $(\Gm, t)$, where $\Gm$ is a
  $k \times n$ generator matrix of the code $\CC(s) \in \mathcal F_{n,k}$
  and $t$, the number of errors that our decoder $\mathbf{Dec}(s)$ can
  correct for the code $\CC(s)$.
  \end{description}
\item[{\bf Encryption}] The plain text is a vector $\mv \in \Fq^k$.
  Then, pick $\ev$ a uniformly random element of the
  set of words of weight $t$ and define the cipher text as:
  \[
  \yv_{\rm cipher} \eqdef \mv \Gm + \ev.
  \]
\item[{\bf Decryption}] Apply $\mathbf{Dec}(s)(\yv_{\rm cipher})$ to
  recover the plain text $\mv$ from $\yv_{\rm cipher}$.
\end{description}

\bigskip

\begin{example}[Algebraic geometry codes]
  \label{ex:McEliece_AG}
  Let $X$ be a curve of genus $g$ and $\mathcal S$ be the set of pairs
  $(\cP, G)$ where $\cP$ is an ordered $n$--tuple of rational points
  of $X$ and $G$ a divisor of degree $k-1+g$ with $2g-2 < \deg G < n$.

  Then, the corresponding family $\mathcal F_{n,k}$ is nothing but the
  family of any $[n,k]$ algebraic geometry code $\CL{X}{\cP}{G}$. For
  such codes, according to the results of \S~\ref{subsec:Feng_Rao},
  using error--correcting arrays,
  these codes benefit from an efficient decoding algorithm correcting up
  to $t = \frac{\dgop^* - 1}{2}\cdot$
\end{example}

\begin{example}[Generalized Reed--Solomon codes]
  \label{ex:McEliece_GRS}
  A subcase of the previous one consists in considering generalized
  Reed--Solomon codes, i.e. algebraic geometry codes from $\P^1$.
  According to Definition~\ref{def:GRS}, the set $\mathcal S$
  can be constructed as the set of pairs $(\xv , \yv)$
  where $\xv$ is an ordered $n$--tuple of distinct elements
  of $\Fq$ and $\yv$ and ordered $n$--tuple of nonzero elements
  of $\Fq$. Then, the map $\CC$ is nothing but
  $\CC : (\xv, \yv) \mapsto \GRS{k}{\xv}{\yv}$.
\end{example}

\subsection{McEliece's original proposal using binary classical
  Goppa codes}
As explained further in \S~\ref{sec:AG_unsecure}, none of
Examples~\ref{ex:McEliece_AG} and~\ref{ex:McEliece_GRS} provide secure
instantiations for McEliece scheme.  An efficient and up to now secure
way to instantiate the scheme is to use subfield subcodes (see
Definition~\ref{def:subfield_subcode}) of GRS codes instead of genuine
GRS ones.  Such codes are usually called {\em alternant codes} (see
\cite[Chap.~12]{sloane1977book}). McEliece's historical proposal was
based on classical Goppa codes defined at
\eqref{eq:classical_Goppa_code} in \S~\ref{subsec:class_Goppa}.  That
is to say codes of the form
\[
\Gamma(\xv, f, \F_{q_0}) = \COm{\P^1}{\cP}{G} \cap \F_{q_0}^n
\]
where $\cP$ is the sequence of points $((x_1 : 1), \ldots, (x_n:1))$
and $G = (f)_0 + P_\infty$. Therefore, classical Goppa codes are
subfield subcodes of AG codes from $\P^1$.

\subsection{Advantages and drawbacks of McEliece scheme}
Since the works of Shor \cite{shor94focs}, it is known that if a
quantum computer exists in the future, then the currently used
cryptographic primitives based on number theoretic problems such as
integer factoring or the discrete logarithm problem would become
insecure. Code--based cryptography is currently of central interest
since it belongs to the few cryptographic paradigms that are believed
to resist to a quantum computer.

On the other hand, a major drawback of McEliece original proposal
\cite{mceliece1978dsn} is the size of the public key: about 32,7
kBytes to claim 65 bits of security\footnote{A cryptosystem is
  claimed to have {\em $x$ bits of security} if the best known attack
  would cost more than $2^x$ operations to the attacker.}. Note that
nowadays the limits in terms of computation lie around $2^{80}$
operations and the new standards require at least $128$ bits
security. To reach this latter level, the recent NIST submission {\em
  Classic McEliece} \cite{bernstein2019nist} suggests public keys of
at least 261 kBytes.  As a comparison, according to NIST
recommendations for key managements
\cite[\S~5.6.1.1,~Tab.~2]{barker2019nist}, to reach a similar security
level with RSA, a public key of 384 Bytes (3072 bits) would be
sufficient!

\subsection{Janwa and Moreno's proposals using AG codes}
Because of this drawback, many works subsequent to McEliece consisted
in proposing other code families in order to reduce the size of the
keys. In particular, Niederreiter suggests to use Reed--Solomon codes
in \cite{niederreiter1986pcit}, i.e. AG codes from $\P^1$.  Later,
the use of AG codes from curves of arbitrary genus has been suggested in
Janwa and Moreno's article \cite{janwa1996dcc}.  More precisely, Janwa
and Moreno's article contains three proposals, involving algebraic
geometry codes:
\begin{enumerate}[(JM1)]
  \item\label{item:concat} a proposal based on AG codes;
  \item\label{item:raw_AG} a proposal based on concatenated AG codes;
  \item\label{item:subfield_McE} a proposal based on subfield subcodes of AG codes.
  \end{enumerate}

  \subsection{Security}
  For any instantiation of McEliece scheme,
  one should distinguish two kind of attacks:
  \begin{enumerate}
  \item {\bf Message recovery attacks} consist in trying to recover
    the plain text from the data of the cipher text. Such an attack
    rests on generic decoding algorithms, such as Prange {\em
      Information Set Decoding} \cite{prange1962ireit} and its
    improvements \cite{stern1988cta, canteaut1998it, lee88eurocrypt,
      dumer89pit, may2011ac, becker2012ec, may2015ec}.
  \item {\bf Key recovery attacks} consist in recovering the secret
    key from the data of the public key and rest on {\em ad--hoc}
    methods depending on the family of codes composing the set of
    public keys.
  \end{enumerate}

  In the sequel, we discuss the security of Janwa and Moreno's three proposal
  with respect to key--recovery attacks. It is explained in particular that
  proposals (JM\ref{item:concat}) and (JM\ref{item:raw_AG}) are not
  secure.
  
  \subsubsection{Concatenated codes are not secure}
  In \cite{sendrier1994ec}, Sendrier showed that concatenated codes
  have an inherent weakness which makes them insecure for public key
  encryption. In particular, proposal (JM\ref{item:concat}) should not
  be used.

  \subsubsection{Algebraic geometry codes are not
    secure}\label{sec:AG_unsecure}
  The raw use of algebraic geometry codes
  (JM\ref{item:raw_AG}) has been subject to two kinds of
  attacks.
  \begin{itemize}
  \item The case of curves of genus $0$ was already proved to be
    insecure by Sidelnikov and Shestakov \cite{sidelnikov1992dma}.
    Note that actually, a procedure to recover the structure of a
    generalized Reed--Solomon code from the data of a generator matrix
    was already known by Roth and Seroussi \cite{roth85it}.
  \item An extension of Sidelnikov and Shestakov's attack due to
    Minder permitted to break AG codes from elliptic curves
    \cite{minder2007phd}. This attack has been extended to genus 2
    curves by Faure and Minder in \cite{faure2008iwcc}. Actually,
    Faure and Minder's attack can be extended to any AG code
    constructed from a hyperelliptic curve but its cost is
    exponential in the curve's genus.
  \item Finally, an attack due to Couvreur, M\'arquez--Corbella and
    Pellikaan~\cite{couvreur2017it} permits to recover an efficient
    decoding algorithm in polynomial time from the public key.  This
    attack can be extended to subcodes of small codimension.

    The attack
    is based on a distinguisher which rests on the $\star$--product
    operation (see Section~\ref{sec:notation}) and on the computation
    of a filtration of the public code composed of algebraic geometry
    subcodes.  This filtration leads to the construction of an error
    correcting array for the public code (see
    \S~\ref{subsec:Feng_Rao}). Similar approaches have been used to
    attack McEliece scheme based on variants of GRS codes
    \cite{couvreur2014dcc} and on classical Goppa codes when the
    subfield has index $2$ \cite{couvreur2017it-2}.
  \end{itemize}

  \medbreak
  
  \subsubsection{Conclusion: only subfield subcodes of AG codes resist}
  The raw use of algebraic geometry codes is not secure for public key
  encryption. On the other hand, subfield subcodes of AG codes
  (JM\ref{item:subfield_McE}) are still resistant to any known
  attack.  Some recent proposals make a step in this direction such as
  \cite[Chap.~5]{barelli2018phd}.  Recall that these codes can be
  regarded as an extension to arbitrary genus of classical Goppa
  codes which remain unbroken forty years after McEliece's original
  proposal.

 \section{Applications related to the $\star$--product:
  frameproof codes, mul\-tiplication algorithms,
  and secret sharing}\label{sec:MPC}
Recall that $\star$ denotes component wise multiplication in $\F_q^n$,
and that the $\star$-product of two linear codes
$\code,\code'\subseteq\F_q^n$ is the linear span of the pairwise
products of codewords from $\code$ and $\code'$:
\[\code\star\code'\eqdef \linspan_{\F_q}\{\word{c}\star\word{c'}:\,\word{c}\in\code,\,\word{c'}\in\code'\}.\]
Recall also that the \emph{square} of $\code$ is
$\code\deux=\code\star\code$, and likewise its higher powers are
defined by induction: $\code\deux[(t+1)]=\code\deux[t]\star\code$.

The link between AG codes and the theory of $\star$-product essentially comes from the obvious inclusion
\begin{equation}
\label{C(P,A)C(P,B)}
\CL{X}{\cP}{A}\star\CL{X}{\cP}{B}\subseteq\CL{X}{\cP}{A+B}.
\end{equation}
It was also observed that in many cases this inclusion is an equality.
For instance a sufficient condition for equality \cite[Cor.~9]{couvreur2017it}
(as a consequence of \cite[Th.~6]{mumford1970cime})
is that $A,B$ satisfy
$\deg(A)\geq 2g$ and $\deg(B)\geq 2g+1$.

\bigbreak

Somehow surprisingly, although $\star$-product is a very simple operation,
it turns out to have many interesting applications.
We already saw in Sections~\ref{sec:decoding} and~\ref{sec:McEliece} 
how it can be used for decoding 
and cryptanalysis.
In this section we will focus on a selection of further applications
in which AG codes play a prominent role:
frameproof codes, multiplication algorithms and arithmetic secret sharing.
A common feature of these constructions was pointed in~\cite{CCX2014}:
they involve codes $\CL{X}{\cP}{G}$ where $G$ is solution
to so-called \emph{Riemann-Roch equations},
arising from \eqref{C(P,A)C(P,B)}.

Eventually, the numerous applications of the $\star$-product of codes
made it desirable to have a better understanding of this operation for
its own sake, from a purely coding theoretical perspective.  Departing
from chronology, we will start with these aspects, and then come back
to applications later.

\subsection{The $\star$-product from the perspective of AG codes}
Works studying the $\star$-product of codes for itself include the research articles
\cite{HR-agis2013,HR-Singleton2013,CCMZ2015,HR-ISIT2015,MZ2015,Cascudo2019,CGR2020}
as well as the expository paper
\cite{HR-AGCT14} which organizes the theory in a more systematic way.

We will survey these works with emphasis on the results in which AG codes are involved.

\subsubsection{Basic properties}

\begin{proposition}
\begin{enumerate}[(i)]
\item \cite[Prop.~11]{HR-agis2013} For any $t\geq1$ we
  have \[\dim(\code\deux[(t+1)])\geq\dim(\code\deux[t])\]
  and \[d(\code\deux[(t+1)])\leq d(\code\deux[t]).\]
\item \cite[Cor.~2.33]{HR-AGCT14} If for some $r$ we have
  $\dim(\code\deux[(r+1)])=\dim(\code\deux[r])$, then also
  $\dim(\code\deux[(r+i)])=\dim(\code\deux[r])$ for all $i\geq0$.
\end{enumerate}
\end{proposition}
The smallest $r=r(\code)$ such that
$\dim(\code\deux[(r+1)])=\dim(\code\deux[r])$ is called the
\emph{regularity} of $\code$ \cite[Def.~1.5]{HR-AGCT14}.  Thus,
$\dim(\code\deux[t])$ strictly increases for $t<r$ and then it
stabilizes.

\bigbreak

If $\code$ is an $[n,k]$-code and $\code'$ is a $[n,k']$-code, we have
a short exact sequence \cite[\S~1.10]{HR-AGCT14}
\begin{equation}
\label{suite_exacte_produit}
0\longrightarrow I(\code,\code')\longrightarrow
\code\otimes\code'\overset{\pi_\Delta}{\longrightarrow}
\code\star\code'\longrightarrow 0
\end{equation}
where $\code\otimes\code'$ is the usual (tensor) product code, of
parameters $[n^2,kk']$, that identifies with the space of $n\times n$
matrices all of whose columns are in $\code$ and all of whose rows are
in $\code'$, and $I(\code,\code')$ is the subspace made of such
matrices that are zero on the diagonal and $\pi_\Delta$ is projection
onto the diagonal.

Equivalently, let $\word{p_1},\dots,\word{p_n}$
(resp. $\word{p'_1},\dots,\word{p'_n}$) be the columns of a full--rank
generator matrix of $\code$ (resp. of $\code'$), and consider the
evaluation map
\begin{equation*}
  \map{\operatorname{Bilin}(\F_q^k\times\F_q^{k'}\to\F_q)}{\F_q^n}
  {b}{(b(\word{p_1},\word{p'_1}),\dots,b(\word{p_n},\word{p'_n})),}
\end{equation*}
where $\operatorname{Bilin}(\F_q^k\times\F_q^{k'}\to\F_q)$ denotes the
space of $\Fq$--bilinear forms on $\Fq^k \times \Fq^{k'}$.  Then,
$\code\star\code'$ is the image of this evaluation map, and
$I(\code,\code')$ is its kernel.

\bigbreak

Similar to \eqref{suite_exacte_produit}, for any $t$ we have a natural
short exact sequence \cite[Sec.~3]{MMP2011}
\begin{equation}
\label{suite_exacte_puissance}
0\longrightarrow I_t(\code)\longrightarrow S^t\code\longrightarrow
\code\deux[t]\longrightarrow 0
\end{equation}
where $S^t\code$ is the $t$-th symmetric power of $\code$,
which actually defines $I_t(\code)$ as the kernel
of the natural map $S^t\code\to\code\deux[t]$.

Equivalently and more concretely,
let $\word{p_1},\dots,\word{p_n}$ be the columns of a
generator matrix of $\code$, denote by $\F_q[x_1,\dots,x_k]_t$ the
space of degree $t$ homogeneous polynomials in $k$ variables, and
consider the evaluation map
\begin{equation*}
\map{\F_q[x_1,\dots,x_k]_t}{\F_q^n}{f}{(f(\word{p_1}),\dots,f(\word{p_n})).}
\end{equation*}
Then $\code\deux[t]$ is the image of this evaluation map, and
$I_t(\code)$ is its kernel, the space of degree $t$
homogeneous forms vanishing at $\word{p_1},\dots,\word{p_n}$.

When $\CC$ is an AG code, M\'arquez-Corbella, Mart{\'\i}nez-Moro, and
Pellikaan showed how $I_2(\code)$ allows to retrieve the underlying
curve of an AG code:
\begin{proposition}[{\cite[Prop.~8 \& Th.~12]{MMP2011}}]
  Let $\code=\CL{X}{\cP}{G}$ be an evaluation code, and assume that
  $g=g(X)$, $n=\card{\cP}$, and $m=\deg(G)$ satisfy
  \[
    2g+2\leq m < \frac{n}{2}\cdot
  \]
  Then $I_2(\code)$ is a set of quadratic equations defining $X$
  embedded in $\P^{k-1}$.
\end{proposition}

\subsubsection{Dimension of $\star$-products}

Let $\code$ be a $[n,k]$-code and $\code'$ a $[n,k']$-code.
From~\eqref{suite_exacte_produit} it follows
\[\dim(\code\star\code')\leq\min(n,kk')\]
and one expects that if $\code,\code'$ are independent random codes,
then, with a high probability, this inequality becomes an equality:
$\dim(\code\star\code')=\min(n,kk')$.  We refer to
\cite[Th.~16--18]{HR-ISIT2015} for precise statements and proofs.  By
contrast, we observe that for AG codes $\code=\CL{X}{\cP}{G}$ and
$\code'=\CL{X}{\cP}{G'}$ with the same evaluation points sequence
$\cP$, we can get much stronger upper bounds.  For instance, if
$2g-2<\deg(G+G')<n$, then Theorem~\ref{thm:basics} together with
\eqref{C(P,A)C(P,B)} give $\dim(\code\star\code')\leq k+k'-1+g$.

Likewise from~\eqref{suite_exacte_puissance} it follows
\[\dim(\code\deux[t])\leq\min\left(n,\binom{k+t-1}{t}\right)\]
and, at least for $t=2$, one expects that if $\code$ is a random code, then with high probability
this inequality becomes an equality: $\dim(\code\deux)=\min\left(n,\frac{k(k+1)}{2}\right)$.
We refer to \cite[Th.~2.2--2.3]{CCMZ2015} for precise statements and proofs.

\begin{remark}[{\cite[Prop.~19]{HR-ISIT2015}}]
  When $t>q$, we always have strict inequality
  $\dim(\code\deux[t])<\binom{k+t-1}{t}$, because of the extra
  relations
  $\word{c}\deux[q]\star\word{c'}=\word{c}\star\word{c'}\deux[q]$.
\end{remark}

Concerning lower bounds, from now on we will make the simplifying
assumption that \emph{all our codes have full support}, i.e. any
generator matrix has no zero column.
 
\begin{proposition}[{\cite[\S~3.5]{HR-AGCT14}}]
\label{dim*=k+k'-1}
Assume $\code$ or $\code'$ is MDS. Then
\[\dim(\code\star\code')\geq\min(n,k+k'-1).\]
\end{proposition}

The \emph{stabilizing algebra} of a code is introduced in \cite[Def.~2.6]{HR-AGCT14} as
\[
  \stab(\code) \eqdef
  \{\word{x}\in\F_q^n:\,\word{x}\star\code\subseteq\code\}.
\]
It is a subalgebra of $\F_q^n$, and admits a basis made of the
characteristic vectors of the supports of the indecomposable
components of $\code$ (\cite[Th.~1.2]{knapp80jalg} and
\cite[Prop.~2.11]{HR-AGCT14}).  In particular, $\dim(\stab(\code))$ is
equal to the number of indecomposable components of $\code$.

Mirandola and Z\'emor \cite{MZ2015} then established a coding analogue
of Kneser's theorem from additive combinatorics
\cite[Th.~5.5]{tao06book}, in which $\stab(\code)$ plays the same
role as the stabilizer of a subgroup:
\begin{theorem}[{\cite[Th.~18]{MZ2015}}]
\label{Kneser}
We have
\[\dim(\code\star\code')\geq k+k'-\dim(\stab(\code\star\code')).\]
\end{theorem}

Pursuing the analogy with additive combinatorics, 
they also obtained the following characterisation of cases of equality in Proposition~\ref{dim*=k+k'-1}, 
which one might see as an analogue of Vosper theorem \cite[Th.~5.9]{tao06book}:
\begin{theorem}[{\cite[Th.~23]{MZ2015}}]
\label{Vosper}
Assume \emph{both} $\code$ and $\code'$ are MDS,
with $k,k'\geq2$, $k+k'\leq n-1$, and
\[\dim(\code\star\code')=k+k'-1.\]
Then $\code$ and $\code'$ are (generalized) Reed-Solomon codes with a common evaluation point sequence.
\end{theorem}

\begin{remark}
  The symmetric version ($\code = \code'$) of Theorem~\ref{Vosper} can
  actually be regarded as a direct consequence of Castelnuovo's Lemma
  \cite[\S~III.2,~p.~120]{arbarello1985book} asserting that for
  $n \geq 2k+1$, any $n$ points in general position in $\P^{k-1}$
  imposing $2r-1$ independent conditions on quadrics lie on a Veronese
  embedding of $\P^1$.
\end{remark}

\subsubsection{Joint bounds on dimension and distance}
A fundamental problem in coding theory is to find lower bounds
(existence results) or upper bounds (non-existence results) relating
the possible dimension and minimum distance of a code.  The analogue
for $\star$-products is to find similar bounds relating the dimensions
of a given number of codes and the minimum distance of their product
(or the dimension of a code and the minimum distance of some power).

Concerning lower bounds, a raw use of AG codes easily shows that for a
given $t$, if $q$ is large enough (depending on $t$), there are
asymptotically good codes over $\F_q$ whose $t$-th powers are also
asymptotically good.  For $t=2$ the following theorem shows that this
result actually holds for \emph{all} $q$.  The proof still uses AG
codes, but in combination with a specially devised concatenation
argument.
\begin{theorem}[{\cite{HR-agis2013}}]
  Over any finite field $\F_q$, there exist asymptotically good codes
  $\code$ whose squares $\code\deux$ are also asymptotically good.

  In particular for $q=2$ and for any $0<\delta<0.003536$ and
  $R\leq 0.001872-0.5294\delta$, there exist codes $\code$ of
  asymptotic rate at least $R$ whose squares $\code\deux$ have
  asymptotic relative minimum distance at least $\delta$.
\end{theorem} 

In finite length, other constructions have been studied by Cascudo and
co-authors that give codes such that both $\code$ and $\code\deux$
have good parameters.  This includes cyclic codes \cite{Cascudo2019}
and matrix-product codes \cite{CGR2020}.

\bigbreak

Concerning upper bounds, the following result is sometimes called the
\emph{product Singleton bound}.
\begin{theorem}[{\cite{HR-Singleton2013}}]
  Let $t\geq2$ be an integer and
  $\code_1,\dots,\code_t\subseteq\F_q^n$ linear codes with full
  support.  Then, there exist codewords
  $\word{c_1}\in\code_1,\dots,\word{c_t}\in\code_t$ whose product has
  Hamming weight
\[0<w_H(\word{c_1}\star\cdots\star\word{c_t})\leq\min(t-1,\,n-(k_1+\cdots+k_t)+t)\]
where $k_i=\dim(\code_i)$.
In particular we have
\[d(\code_1\star\cdots\star\code_t)\leq\min(t-1,\,n-(k_1+\cdots+k_t)+t).\]
\end{theorem}

Mirandola and Z\'emor describe cases of equality for $t=2$.
Essentially these are either pairs of (generalized) Reed-Solomon codes with a common evaluation point sequence,
or pairs made of a code and its dual (up to diagonal equivalence).
See \cite[Sec.~V]{MZ2015} for further details.

\subsubsection{Automorphisms}

Let $\code$ be a linear code of regularity $r(\code)$, and let $t\leq t'$ be two integers.
Assume one of the following two conditions holds:
\begin{itemize}
\item $t|t'$
\item $t'\geq r(\code)$.
\end{itemize}
Then, from \cite[\S~2.50-2.55]{HR-AGCT14}, it follows that
$\card{\Aut(\code\deux[t])}$ divides $\card{\Aut(\code\deux[t'])}$.

This makes one wonder whether one could compare $\Aut(\code\deux[t])$
and $\Aut(\code\deux[t'])$ for arbitrary $t\leq t'$.  For instance, do
we always have
$\card{\Aut(\code\deux)}\leq\card{\Aut(\code\deux[3])}$?

Motivated by Proposition~\ref{auto_courbe_code},
Couvreur and Ritzenthaler tested this question against AG codes
and eventually showed that the answer is negative.

\begin{example}[\cite{LacapelleBironDay2}]
  Let $E$ be the elliptic curve defined by the Weierstrass equation
  $y^2=x^3+1$ over $\F_7$, with point at infinity $O_E$.  Set
  $\cP=E(\F_7)$, $G=3O_E$, and consider the evaluation code
  $\code=\CL{E}{E(\F_7)}{3O_E}$, so $\code$ has generator matrix
\begin{equation*}
\left(\begin{array}{cccccccccccc}
0& 0& 1& 1& 2& 2& 3& 4& 4& 5& 6& 0\\
1& 6& 3& 4& 3& 4& 0& 3& 4& 0& 0& 1\\ 
1& 1& 1& 1& 1& 1& 1& 1& 1& 1& 1& 0\\
\end{array}\right).
\end{equation*}
Then, computer--aided calculations show that
$\code\deux=\CL{E}{E(\F_7)}{6O_E}$ has \[\card{\Aut(\code\deux)}=432\]
and $\code\deux[3]=\CL{E}{E(\F_7)}{9O_E}$
has \[\card{\Aut(\code\deux[3])}=108\] so in particular
\[\card{\Aut(\code\deux)}>\card{\Aut(\code\deux[3])}.\]
\end{example}

\subsection{Frameproof codes and separating systems}

Frameproof codes were introduced in the context of traitor tracing
schemes \cite{BCN1994,BS1998,SW1998}.  Slightly differing definitions
of this notion can be found.  Here, we work only with linear codes and
we say such a code $\code$ is {\em $t$-frameproof}, or {\em $t$-wise
intersecting}, if the supports of any $t$ nonzero codewords have a
nonempty common intersection.  In terms of $\star$-products, it means
for any $\word{c_1},\dots,\word{c_t}\in\code$,
\[\word{c_1},\dots,\word{c_t}\neq\word{0}\quad\Longrightarrow\quad\word{c_1}\star\cdots\star\word{c_t}\neq\word{0}.\]
Some elementary combinatorial constructions of frameproof codes can be
found in \cite{CE2000,Blackburn2003}.

In \cite{Xing2002} Xing considers asymptotic bounds, and in particular
constructions from AG codes.  The starting observation of his main
result is that a sufficient condition for an evaluation code
$\CL{X}{\cP}{G}$ to be $s$-frameproof is
\begin{equation}
\label{critereXing}
\ell(sG-D_\cP)=0.
\end{equation}
Condition \eqref{critereXing} is perhaps the simplest instance of what
was later called a Riemann-Roch equation~\cite{CCX2014}.

Xing uses a counting argument in the divisor class group of the curve
to prove the existence of solutions to~\eqref{critereXing}
with $\deg(sG-D_\cP)\approx(1-2\log_q s)g$.
This leads to the following result.
\begin{theorem}[\cite{Xing2002}]
For $2\leq s\leq A(q)$, there exist $s$-frameproof codes over $\F_q$
of length going to infinity and asymptotic rate at least
\[\frac{1}{s}-\frac{1}{A(q)}+\frac{1-2\log_q s}{sA(q)}\cdot\]
\end{theorem}
In this bound, the $2\log_q s$ term reflects the possible $s$-torsion in the class group that hinders the counting argument.

For $s=2$, an alternative method is proposed in \cite{HR-IJM2013}
that allows to construct solutions to~\eqref{critereXing}
up to $\deg(2G-D_\cP)\approx g-1$, which is best possible.
This gives:
\begin{theorem}[\cite{HR-IJM2013}]
\label{2fp}
If $A(q)\geq4$, one can construct $2$-frameproof codes over $\F_q$
of length going to infinity and asymptotic rate at least
\[\frac{1}{2}-\frac{1}{2A(q)}\cdot\]
\end{theorem}

\bigbreak

This result has a somehow unexpected application.  In~\cite{EF1983},
Erd\"os and F\"uredi studied a certain problem in combinatorial
geometry.  It led them to introduce certain configurations, which
admit the following equivalent descriptions:
\begin{itemize}
\item sets of $M$ vertices of the unit cube $\{0,1\}^n$ in the
  Euclidean space $\mathbb{R}^n$, any three
  $\word{x},\word{y},\word{z}$ of which form an acute angle:
  $\langle\word{x}-\word{z},\word{y}-\word{z}\rangle_{\textrm{Eucl}}>0$
\item sets of $M$ points in the binary Hamming space $\F_2^n$, any
  three $\word{x},\word{y},\word{z}$ of which satisfy the strict
  triangle inequality:
  $d_H(\word{x},\word{y})<d_H(\word{x},\word{z})+d_H(\word{y},\word{z})$
\item sets of $M$ binary vectors of length $n$, any three
  $\word{x},\word{y},\word{z}$ of which admit a position $i$ where
  $\word{x}_i=\word{y}_i\neq\word{z}_i$
\item sets of $M$ subsets in an $n$-set, no three $A,B,C$ of which
  satisfy $A\cap B\subseteq C\subseteq A\cup B$.
\end{itemize}
In the context of coding theory, such a configuration is also called a
{\em binary $(2,1)$-separating system}, of length~$n$ and size~$M$.  A
random coding argument shows that there exist such separating systems
of length $n\to\infty$ and size
\[M\approx(2/\sqrt{3})^n\approx 1.1547005^n.\] However, combining
Theorem~\ref{2fp} with a convenient concatenation argument provides a
dramatic(?) improvement.
\begin{theorem}[\cite{HR-IJM2013}]
  One can construct binary $(2,1)$-separating systems of length
  $n\to\infty$ and size
\[M\approx (11^{\frac{3}{50}})^n\approx 1.1547382^n.\]
\end{theorem}
This is an instance of a construction based on AG codes beating a
random coding argument, over the binary field!

\subsection{Multiplication algorithms}

The theory of bilinear complexity started with the celebrated algorithms
of Karatsuba~\cite{KO1963}, that allows to multiply two $2$-digit numbers with $3$ elementary
multiplications instead of $4$, and of Strassen~\cite{Strassen1969}, that allows to multiply two $2\times2$
matrices with $7$ field multiplications instead of $8$.
Used recursively, these algorithms then allow to multiply numbers with a large number of digits,
or matrices of large size,
with a dramatic improvement on complexity over the naive methods.

Authors subsequently studied the complexity of various multiplication maps,
such as the multiplication of two polynomials modulo a given polynomial \cite{FZ1977,Winograd1977}.
This includes in particular the multiplication map $m_{\F_{q^k}/\F_q}$ in an extension of finite fields $\F_{q^k}$ over $\F_q$.

\begin{definition}
A \emph{bilinear multiplication algorithm} of length $n$ for $\F_{q^k}$ over $\F_q$ is the data of
linear maps $\alpha,\beta:\F_{q^k}\to(\F_q)^n$ and $\omega:(\F_q)^n\to\F_{q^k}$ such that the following diagram commutes:
\begin{equation*}
\begin{CD}
\F_{q^k}\times\F_{q^k}  @>{m_{\F_{q^k}/\F_q}}>>\F_{q^k}\\
@V{\alpha\times\beta}VV @AA{\omega}A\\
(\F_q)^n\times(\F_q)^n @>{\star}>> (\F_q)^n.
\end{CD}
\end{equation*}

Equivalently, it is the data of linear forms $\alpha_1,\dots,\alpha_n,\beta_1,\dots,\beta_n:\F_{q^k}\to\F_q$
and elements $\omega_1,\dots,\omega_n\in\F_{q^k}$
such that for all $x,y\in\F_{q^k}$ we have
\begin{equation*}
xy=\sum_{i=1}^n\alpha_i(x)\beta_i(y)\omega_i.
\end{equation*}
\end{definition}

\begin{definition}
A multiplication algorithm as above is \emph{symmetric} if $\alpha=\beta$, or equivalently, if $\alpha_i=\beta_i$ for all $i$.
\end{definition}

\begin{definition}
The bilinear complexity $\mu_q(k)$ (resp. the symmetric bilinear complexity $\mu^{\textrm{sym}}_q(k)$)
of $\F_{q^k}$ over $\F_q$ is the smallest possible length of a bilinear multiplication algorithm
(resp. a symmetric bilinear multiplication algorithm) for $\F_{q^k}$ over $\F_q$.
\end{definition}

Obviously we always have $\mu_q(k)\leq\mu^{\textrm{sym}}_q(k)\leq k^2$.
For $k\leq\frac{q}{2}+1$ this is easily improved to $\mu_q(k)\leq\mu^{\textrm{sym}}_q(k)\leq 2k-1$,
using Fourier transform, or equivalently, Reed-Solomon codes.

In \cite{ChCh1988}, Chudnovsky and Chudnovsky highlighted further links
between multiplication algorithms and codes.
Then, using a construction similar to AG codes, they were able to prove
the first linear asymptotic upper bound on $\mu_q(k)$.
\begin{theorem}[{\cite[Th.~7.7]{ChCh1988}}]\label{thChCh}
If $q\geq 25$ is a square, then
\[\limsup_{k\to\infty}\frac{1}{k}\mu_q(k)\leq
  2\left(1+\frac{1}{\sqrt{q}-3}\right).\]
\end{theorem}
This result was originally stated for bilinear complexity, but the
proof also works for symmetric bilinear complexity, since it provides
symmetric algorithms.  So we actually get:
\[\limsup_{k\to\infty}\frac{1}{k}\mu^{\textrm{sym}}_q(k)\leq
  2\left(1+\frac{1}{\sqrt{q}-3}\right)\] for $q\geq 25$ a square.

Chudnovsky and Chudnovsky proceed by \emph{evaluation-interpolation}
on curves.  Suppose we are given a curve $X$ over $\F_q$, together
with a collection $\cP=\{P_1,\dots,P_n\}$ of $n$ distinct rational points, a
point $Q$ of degree $k$, and a suitably chosen auxiliary divisor $G$.
Also assume:
\begin{enumerate}[(i)]
\item\label{item:surj} the evaluation-at-$Q$ map
  $L(G)\longrightarrow\F_{q^k}$ is surjective;
\item\label{item:inj} the evaluation-at-$\cP$ map
  $L(2G)\longrightarrow(\F_q)^n$ is injective.
\end{enumerate}
Then, in order to multiply $x,y$ in $\F_{q^k}$
with only $n$ multiplications in $\F_q$, one can do as follows:
\begin{itemize}
\item Thanks to (\ref{item:surj}), lift $x,x'$ to functions
  $f_x,f_{x'}$ in $L(G)$ such that $f_x(Q)=x$, $f_{x'}(Q)=x'$, and
  then evaluate these functions at $\cP$ to get codewords
  $\word{c}_x=(f_x(P_1),\dots,f_x(P_n))$,
  $\word{c}_{x'}=(f_{x'}(P_1),\dots,f_{x'}(P_n))$ in $\CL{X}{\cP}{G}$.
\item Compute $\word{c}_x\star\word{c}_{x'}=(y_1,\dots,y_n)$ in $\CL{X}{\cP}{2G}$,
i.e. $y_i=f_x(P_i)f_{x'}(P_i)$ for $1\leq i\leq n$.
\item By ``Lagrange interpolation'', find a function $h$ in $L(2G)$
  that takes these values $h(P_1)=y_1, \dots ,h(P_n)=y_n$, and then
  evaluate $h$ at $Q$.
\end{itemize}
Then (\ref{item:inj}) ensures that we necessarily have $h=f_xf_{x'}$,
so this last evaluation step gives $h(Q)=xx'$ as wished.  

In \cite{STV1992}, Shparlinski, Tsfasman and \vladut{} propose several improvements.
First, they correct certain imprecise statements in \cite{ChCh1988},
or give additional details to the proofs. For instance, on the choice
of the curves to which the method is applied in order to get Theorem~\ref{thChCh},
they provide an explicit description of a family of Shimura
curves $X_i$ over $\F_q$, for $q$ a square, that satisfy:
\begin{itemize}
\item (optimality) $|X_i(\F_q)|/g(X_i)\to A(q)=\sqrt{q}-1$;
\item (density) $\quad g(X_{i+1})/g(X_i)\to 1$.
\end{itemize}

They further show how to deduce linearity of the complexity over an arbitrary
finite field.
\begin{lemma}[{\cite[Cor.~1.3]{STV1992}}]
Set $M_q\eqdef\limsup_{k\to\infty}\frac{1}{k}\mu_q(k)$.
Then for any prime power $q$ and integer $m$ we have
\[M_q\leq\mu_q(m)M_{q^m}.\]
\end{lemma}
\noindent From this and Theorem~\ref{thChCh} it readily follows that \[M_q<+\infty\]
for all $q$.
As above, this result was originally stated only for bilinear complexity,
but the proof also works for symmetric bilinear complexity, so we get likewise
\[M^{\textrm{sym}}_q\eqdef\limsup_{k\to\infty}\frac{1}{k}\mu^{\textrm{sym}}_q(k)<+\infty.\]

Last, they devise a possible improvement on Theorem~\ref{thChCh}.
For this they observe that for given $\cP$ and $Q$,
the choice of the auxiliary divisor $G$ satisfying (\ref{item:surj}) and
(\ref{item:inj}) above reduces 
to the solution of the following system of two Riemann-Roch equations:
\begin{equation}
\label{systSTV}
\left\{\begin{array}{c}
\ell(K_X-G+Q)=0\\
\ell(2G-D_\cP)=0.
\end{array}\right.
\end{equation}
In order to get the best possible parameters, one would like to
set $\cP=X(\F_q)$, and find
a solution with $\deg(K_X-G+Q)$ and $\deg(2G-D_\cP)$ close to $g-1$.
The authors propose a method to achieve this,
but it was later observed that their argument is incomplete.
A corrected method was then proposed in \cite{HR-JCompl2012},
relying on the tools introduced in the proof of Theorem~\ref{2fp}.
Combining all these ideas then gives:
\begin{theorem}[{\cite[Th.~6.4]{HR-JCompl2012}}]
If $q\geq 49$ is a square, then
\[M^{\textrm{sym}}_q\leq 2\left(1+\frac{1}{\sqrt{q}-2}\right).\]
\end{theorem}
In this same work, it is also observed that if one does not insist on
having symmetric algorithms, then \eqref{systSTV} can be replaced
with an asymmetric variant: for given $\cP$ and $Q$,
find two divisors $G$ and $G'$ satisfying
\begin{equation}
\label{asymRR}
\left\{\begin{array}{c}
\ell(K_X-G+Q)=0\\
\ell(K_X-G'+Q)=0\\
\ell(G+G'-D_\cP)=0.
\end{array}\right.
\end{equation}
This asymmetric system is easier to solve, leading to:
\begin{theorem}[{\cite[Th.~6.4]{HR-JCompl2012}}]
If $q\geq 9$ is a square, then
\[M_q\leq 2\left(1+\frac{1}{\sqrt{q}-2}\right).\]
\end{theorem}

The study of bilinear multiplication algorithms over finite fields
is a very active area of research and we covered only one specific aspect.
Other research directions include:
\begin{itemize}
\item give asymptotic bounds for non square $q$, or for very small
  $q=2,3,4,\dots$;
\item instead of asymptotic bounds valid for $k\to\infty$, give
  uniform bounds valid for all $k$;
\item study multiplication algorithms in more general finite
  dimensional algebras, not only extension fields;
\item give effective constructions of multiplication algorithms.
\end{itemize}
For a more exhaustive survey of recent results we refer to \cite{BCPRRR2019}.

\subsection{Arithmetic secret sharing}

Our starting point here will be Shamir's secret sharing scheme.
Suppose Alice has a secret $s\in\F_q$, and she wants to distribute it
among $n$ players.  We will assume $n<q$, and the $n$ players are
labelled by $n$ distinct nonzero elements $x_1,\dots,x_n$ in $\F_q$.
Given a certain threshold $t\leq n$, Alice picks $t-1$ random elements
$c_1,\dots,c_{t-1}$ in $\F_q$, and considers the polynomial
$P(X)=s+c_1X+c_2X^2+\cdots+c_{t-1}X^{t-1}$.  Then each player $x_i$ receives his share
$y_i=P(x_i)$.  The main property of this scheme is that:
\begin{enumerate}[(i)]
\item\label{item:shamir1} any coalition of at least $t$ players can use Lagrange
  interpolation to reconstruct the polynomial $P$, hence also
  the secret $s=P(0)$, from their shares;
\item\label{item:shamir2} any coalition of up to $t-1$ players has no
  information at all about the secret, i.e. all possible values for $s$
  appear equiprobable to them.
\end{enumerate}

Another property of Shamir's scheme is its linearity: suppose Alice
has two secrets $s$ and $\widetilde{s}$, and distributes them to the
same players.  Let $P$ be the polynomial used to distribute $s$, with
corresponding shares $y_1,\dots,y_n$, and let $\widetilde{P}$ be the
polynomial used to distribute $\widetilde{s}$, with corresponding
shares $\widetilde{y}_1,\dots,\widetilde{y}_n$.  Then
$y_1+\widetilde{y}_1,\dots,y_n+\widetilde{y}_n$ are shares
corresponding to the sum $s+\widetilde{s}$ of the two secrets: indeed,
these are the shares obtained when distributing $s+\widetilde{s}$ with
the polynomial $P+\widetilde{P}$.

Shamir's scheme also enjoys a \emph{multiplicative} property, but it
is more subtle: we have $(P\widetilde{P})(0)=s\widetilde{s}$ and
$(P\widetilde{P})(x_i)=y_i\widetilde{y_i}$, so in some sense,
$y_1\widetilde{y}_1,\dots,y_n\widetilde{y}_n$ are the shares obtained
when distributing $s\widetilde{s}$ with the polynomial
$P\widetilde{P}$.  A drawback is that $P\widetilde{P}$ may have degree
up to $2t-2$, instead of $t-1$ in the original scheme.  Still we can
say that any coalition of $2t-1$ players can reconstruct the
secret product $s\widetilde{s}$ from their product shares
$y_i\widetilde{y}_i$.  On the other hand, $P\widetilde{P}$ is not
uniformly distributed in the set of polynomials $R$ of degree up to
$2t-2$ satisfying $R(0)=s\widetilde{s}$, so it is unclear what
information a smaller coalition can get.

There are close links between linear secret sharing schemes and linear
codes, under which Shamir's scheme corresponds to Reed-Solomon codes.
Indeed, observe that share vectors $(y_1,\dots,y_n)$ in Shamir's
scheme are precisely the codewords of an RS code.  Properties
(\ref{item:shamir1}) and (\ref{item:shamir2}) then reflect the MDS
property of Reed-Solomon codes.  And a common limitation to Shamir's
scheme and to RS codes is that, for a given $q$, the number $n$ of
players, or the length of the code, remains bounded essentially by
$q$.

The importance of multiplicative linear secret sharing schemes perhaps
comes from a result in \cite{CDM2000}, that shows that these schemes
can serve as a basis for secure multiparty computation protocols.  In
\cite{IKOS2009} it is also shown that certain two-party protocols, for
instance a zero-knowledge proof, admit communication-efficient
implementations in which one player, for instance the verifier, has to
simulate ``in her head'' a multiparty computation with a large number
of players.  This last result makes it very desirable to overcome the
limitation on the number of players in Shamir's scheme.

In the same way that AG codes provide a generalisation of RS codes of
arbitrary length, one can construct linear secret sharing schemes with
an arbitrary number of players by using evaluation of functions on an
algebraic curve.  Moreover, under certain conditions, these schemes
also admit good multiplicative properties. This was first studied by
Chen and Cramer in \cite{CC2006}, and then refined and generalized in
several works such as \cite{CCCX2009} and \cite{CCX2011}.  We follow
the presentation of the latter.

Given a finite field $\F_q$ and integers $k,n\geq1$, we equip the
vector space $\F_q^k\times\F_q^n$ with the linear projection maps
\begin{equation*}
  \map[\pi_0:]{\F_q^k\times\F_q^n}{\F_q^k}{\word{v}=(s_1,\dots,s_k,c_1,\dots,c_n)}{\word{v}_0\eqdef(s_1,\dots,s_k)}
\end{equation*}
and, for any subset $B\subseteq\{1,\dots,n\}$,
\begin{equation*}
\map[\pi_B:]{\F_q^k\times\F_q^n}{\F_q^{|B|}}{\word{v}=(s_1,\dots,s_k,c_1,\dots,c_n)}{\word{v}_B\eqdef(c_i)_{i\in B}}.
\end{equation*}

\begin{definition}
A $(n,t,d,r)$-arithmetic secret sharing scheme for $\F_q^k$ over $\F_q$
is a linear subspace $\code\subseteq\F_q^k\times\F_q^n$ with the following properties:
\begin{enumerate}[(i)]
\item\label{item:disc} ($t$-disconnectedness) for any subset
  $B\subseteq\{1,\dots,n\}$ of cardinality $|B|=t$, the projection map
\[\map[\pi_{0,B}:]{\code}{\F_q^k\times\pi_B(\code)}{\word{v}}{(\word{v}_0,\word{v}_B)}\]
is surjective 
\item\label{item:power_rec} ($d$-th power $r$-reconstruction) for any
  subset $B\subseteq\{1,\dots,n\}$ of cardinality $|B|=r$ we have
\[(\ker\pi_B)\cap \code\deux[d]\subseteq(\ker\pi_0)\cap \code\deux[d].\]
\end{enumerate}
Moreover we say that $\code$ has {\em uniformity} if in
(\ref{item:disc}) we have $\pi_B(\code)=\F_q^{|B|}$ for all $B$ with
$|B|=t$.
\end{definition}
Such a scheme allows to distribute a secret
$\word{s}=(s_1,\dots,s_k)\in\F_q^k$ among $n$ players. To do this, one
chooses a random $\word{v}\in \code$ such that $\word{v}_0=\word{s}$,
and for each $i\in\{1,\dots,n\}$, the $i$-th player receives his share
$c_i=\word{v}_{\{i\}}$.  Now condition (\ref{item:disc}) means that
for each coalition $B$ of $t$ adversary players, the secret vector
$\word{v}_0$ is independently distributed from their share vector
$\word{v}_B$.  On the other hand, condition (\ref{item:power_rec})
means that any coalition of $r$ honest players can reconstruct the
$\star$-product of $d$ secret vectors from the $\star$-product of
their corresponding $d$ share vectors.

It turns out that these conditions (\ref{item:disc}) and
(\ref{item:power_rec}) can be captured by Riemann-Roch equations:
\begin{lemma}
  Let $X$ be an algebraic curve over $\F_q$, together with a
  collection of $k+n$ distinct rational points
  $\mathcal{S}=\{Q_1,\dots,Q_k,P_1,\dots,P_n\}$.  Define the divisor
  $Q=Q_1+\cdots+Q_k$, and for each subset $B\subseteq\{1,\dots,n\}$,
  $P_B=\sum_{i\in B}P_i$.  Let $G$ be a divisor on $X$ that satisfies
  the following system:
\begin{equation}
\label{systCCX}
\left\{\begin{array}{cl}
\ell(K_X-G+P_B+Q)=0 & \qquad\textrm{for all $B\subseteq\{1,\dots,n\}$ with $|B|=t$}\\
\ell(dG-P_B)=0 & \qquad\textrm{for all $B\subseteq\{1,\dots,n\}$ with $|B|=r$}.
\end{array}\right.
\end{equation}
Then $\code=\CL{X}{\mathcal{S}}{G}$ is a $(n,t,d,r)$-arithmetic secret sharing scheme for $\F_q^k$, with uniformity.
\end{lemma}
In \cite{CCX2011} a method is developed to solve \eqref{systCCX}
using some control on the $d$-torsion of the class group of the curve.
It is not known if this method is optimal, but it gives arithmetic secret
sharing schemes with the best parameters up to now\footnote{
Observe that the method of \cite{HR-IJM2013}, that gives
optimal solutions to \eqref{critereXing} and \eqref{systSTV},
does not operate with \eqref{systCCX}, even in the case $d=2$,
because the number of equations in the system is too high.
}.
It leads to:
\begin{theorem}[\cite{CCX2011}]
For all prime powers $q$ except perhaps for $q=2,3,4,5,7,11,13$,
there is an infinite family of $(n,t,2,n-t)$-arithmetic secret sharing
schemes for $\F_q^k$ over $\F_q$
with uniformity, where $n$ is unbounded, $k=\Omega(n)$ and $t=\Omega(n)$.
\end{theorem}
It should be noted that, if one is ready to drop the uniformity condition,
then the corresponding result holds for \emph{all} $q$.
This can be proved using a concatenation argument~\cite{CCCX2009}.

The literature on arithmetic secret sharing is rapidly evolving,
and we cannot cover all recent developments.
For further references on the topic, together with details on the
connection with multiparty computation, we recommend the book~\cite{CDN2015book}.

 \section{Application to distributed storage: locally recoverable
  codes}\label{sec:LRC}
\subsection{Motivation}
The impressive development of cloud computing and distributed storage
in the last decade motivated new paradigms and new questions in coding
theory yielding an impressive number of works studying the
construction, the features and the limitations of codes having ``good
local properties''.

While the literature provides many definitions, such as {\em locally
  correctable codes}, {\em locally decodable codes}, etc.
In this \ifx\localversion\undefined {chapter}
\else {article}
\fi
we only focus on so--called {\em locally recoverable codes} (LRC).
To define them, let us first define the notion of {\em restriction} of
a code.

\begin{definition}
  Let $\CC \subseteq \Fq^n$ be a code and
  $I \subseteq \{1, \dots, n\}$.  The {\em restriction of $\CC$ to
    $I$} denoted by $\rest{\CC}{I}$ is the image of $\CC$ under the
  projection
  \[\map{\Fq^n}{\Fq^{|I|}}{(c_1, \dots, c_n)}{(c_i)_{i\in I}.}\]
\end{definition}

\begin{remark}
  Classically in the literature, this operation is referred to as
  {\em puncturing $\CC$} at $\{1, \dots, n\} \setminus I$.
  In the sequel, the term {\em restriction} seems more relevant
  since, we will deal with evaluation codes and the restriction of
  the code will be obtained by evaluating restrictions of some given
  functions.
\end{remark}

\begin{definition}[Locally recoverable code]\label{def:LRC}
  A code $\CC \subseteq \Fq^n$ is {\em locally recoverable} with {\em
    locality} $\ell$ if for any $i \in \{1, \ldots, n\}$, there exists
  at least one subset $A(i) \subseteq \{1, \dots, n\}$ containing $i$
  such that $|A(i)| \leq \ell + 1$ and $\rest{\CC}{A(i)}$ has minimum
  distance $\geq 2$.
\end{definition}

\begin{definition}[Recovery set]
  In the context of Definition~\ref{def:LRC}, a subset $A(i)$ is
  called a {\em recovery set of $\CC$ for $i$}.
\end{definition}

\begin{remark}\label{rem:on_recovery_sets}
  Note that the sets $A(1), \dots, A(n)$ need not be distinct.  In
  addition, we emphasize that, for a given position
  $i \in \{1, \dots, n\}$, there might exist more than one recovery set
  for $i$; this is actually the point of the notion of {\em
    availability} discussed further in \S~\ref{subsec:availability}.
  
  On the other hand, in most of the examples presented in
  \S~\ref{subsec:TBcode}, \ref{subsec:BTVcodes}
  and~\ref{subsec:higher_local_distance}, we consider a partition
  $A_1 \sqcup \cdots \sqcup A_s$ of $\{1, \dots , n\}$ such that for
  any $i \in \{1, \dots, n\}$, the unique recovery set for $i$ is the
  unique subset $A_j$ containing $i$.
\end{remark}

\begin{remark}\label{rem:word_in_dual}
  One can prove that Definition~\ref{def:LRC} is equivalent to the
  following one.  For any $i \in \{1, \dots, n\}$, there exists at
  least one codeword $\cv^{(i)} \in \CC^\perp$ of weight less than or
  equal to $\ell + 1$ and whose $i$--th entry is nonzero.
\end{remark}

Let us give some practical motivation for this definition.  Suppose we
distribute data on $n$ servers. Our file (or a part of it) is an
element $\mv \in \Fq^k$ that has been encoded as a codeword
$\cv \in \CC$, where $\CC$ is a code with locality $\ell$. For any
$i \in \{1, \ldots, n\}$ the element $c_i \in \Fq$ is stored in the
$i$--th server.  In distributed storage systems, data should be
recoverable at any moment, even if failures or maintenance operations
are performed. For this sake, when a machine fails, the data it
contains should be recovered and saved on another machine. To perform
such operation efficiently, we wish to limit the number of
machines from which data is downloaded. Here comes the interest of
codes having a small locality!  Suppose the $i$--th server
failed. Then, we need to reconstruct $c_i$ from the knowledge of the
$c_j$'s for $j \neq i$. The objective is to recover $c_i$ from the
knowledge of an as small as possible number of $c_j$'s.  From
Remark~\ref{rem:word_in_dual}, there exists $\dv \in \CC^\perp$ of
weight less than or equal to $\ell + 1$ with $\dv_i \neq 0$. Then, the
support of $\dv$ is $\{i, i_1, \ldots, i_s\}$ with $s \leq \ell $ and
\[
c_i = -\frac{1}{d_i} \sum_{j=1}^s c_{i_j}d_{i_j}.
\]
Consequently, $c_i$ can be recovered after downloading data from
at most $\ell$ distinct servers. The smaller the $\ell$, the more efficient
the recovery process.

\begin{remark}
  Note that the literature on distributed storage actually involves
  two distinct families of codes:
  \begin{itemize}
  \item {\em locally recoverable codes} which are the purpose of the present section;
  \item {\em regenerating codes} which are not discussed in the present
\ifx\localversion\undefined {chapter}
\else {article}
\fi
.
  \end{itemize}
  We refer the interested reader to \cite{dimakis2011procieee} for an
  introduction to regenerating codes and to \cite{ng2019dcc} for a
  more geometric presentation of them.
\end{remark}

\subsection{A bound on the parameters involving the locality}\label{subsec:Gopalan_bound}

The most classical bound, which can be regarded as a Singleton bound
for locally recoverable codes is due to Gopalan, Haung and Simitci.
\begin{theorem}[\cite{gopalan2012it,papailiopoulos2014it}]
  Let $\CC \subseteq \Fq^n$ be a locally recoverable code of dimension $k$, minimum distance $d$
  and locality $\ell$. Then,
  \begin{equation}\label{eq:Singleton_local}
    d \leq n - k - \left\lceil \frac{k}{\ell} \right\rceil + 2.
  \end{equation}
\end{theorem}

\subsection{Tamo--Barg codes}\label{subsec:TBcode}

The original proposal of {\em optimal locally recoverable codes}, i.e. codes reaching
bound~\eqref{eq:Singleton_local} is due to Barg and Tamo
\cite{tamo2014it} and are codes derived from Reed--Solomon codes. The
construction is as follows.
\begin{itemize}
\item Consider a subset
$A\subseteq \Fq$ of cardinality $n$ such that
$(\ell + 1)$ divides $n$ and a partition of $A$:
\[
  A = A_1 \sqcup \cdots \sqcup A_{\frac{n}{\ell +1 }}
\]
into disjoint subsets of size $\ell + 1$.  Denote by $x_1, \dots, x_n$
the elements of $A$.
\item Consider a polynomial $g$ of degree
$\ell + 1$ which is constant on any element of the partition, i.e.
\[
  \forall i \in \left\{1, \dots, \frac{n}{\ell + 1}\right\},\ \forall x, y \in A_i,
  g(x) = g(y).
\]
\end{itemize}
Then, for $k$ divisible by $\ell$, one defines the $[n,k]$ Tamo--Barg code of
locality $\ell$ as the code
\begin{equation}\label{eq:TBcode}
\CC \eqdef  \left\{ (f(x_1), \dots, f(x_n)) ~\bigg|~ f(X) = \sum_{i=0}^{\ell-1}\sum_{j =
      0}^{\frac{k}{\ell} -1}  a_{ij} X^i g(X)^j\right\}.
\end{equation}
This code has length $n$ and dimension $k$.
In addition, the polynomials that are evaluated to generate codewords
have degree at most
\[
  \deg f \leq \ell - 1 + (\ell + 1)\left(\frac{k}{\ell}-1 \right)
  = k + \frac{k}{\ell} -2.
\]
Next, to get a lower bound for its minimum distance, it suffices to embed this
code into a larger code whose minimum distance is known, namely a Reed--Solomon
code:
\[
  \CC \subseteq \left\{(f(x_1), \dots, f(x_n)) ~|~ \deg f \leq
    k+\frac{k}{\ell}-2\right\} = \CL{\P^1}{A}{\left(k+ \frac{k}{\ell}
      - 2\right)P_{\infty}},
\]
whose minimum distance equals $n- k - \frac{k}{\ell} +2$.
Finally, it remains to be shown that the code has locality $\ell$
which will explain the rationale behind the construction.
Suppose we wish to recover a symbol $c_r$ from a given
codeword $\cv \in \CC$. The index $r \in A_s$ for some integer $s$ and
since $g$ is constant on $A_s$ and identically equal to some constant
$\gamma_s$, the restriction  to $A_s$ of any polynomial $f$
as in (\ref{eq:TBcode}) coincides with the
polynomial $\sum_{i,j} a_{ij} \gamma_s^jx^i$ which is a polynomial of
degree $< \ell$.  By Lagrange interpolation, this polynomial
is entirely determined by its evaluations at the $\ell$ elements of
$A_s \setminus \{r\}$ and hence, its evaluation at $r$ can be deduced
from these evaluations.

In summary,
this code has parameters $\left[n, k, n- k - \frac{k}{\ell} +2\right]$
and locality $\ell$, hence it is optimal with respect to
Bound~(\ref{eq:Singleton_local}).

\begin{example}\label{ex:locality_with_group_action}
  An explicit example of a polynomial $g$ which is constant on each
  element of a given partition can be obtained from a polynomial which
  is invariant under some group action on the affine line and take the
  partition given by the cosets with respect to this action. For
  instance, suppose that $(\ell + 1)\mid (q-1)$, then $\Fq$ contains a
  primitive $(\ell + 1)$--th root of unity $\zeta$. The cyclic
  subgroup of $\Fq^\times $ generated by $\zeta$ acts multiplicatively
  on the affine line via the map $z \mapsto \zeta z$ which splits
  $\Fq^\times$ into $\frac{q-1}{\ell + 1}$ cosets.  Next, the
  polynomial $g(X) = X^{\ell + 1}$ is obviously constant on these
  cosets. This provides optimal $[n, k]$ codes of locality $\ell$ for
  any $n \leq q-1$ and any $k < n$ such that $(\ell + 1)|n$
  and $\ell | k$.
\end{example}

\subsection{Locally recoverable codes from coverings of algebraic
  curves: Barg--Tamo--\vladut{} codes}\label{subsec:BTVcodes}
The discussion to follow requires the introduction of the following
definition.

\begin{definition}[Galois cover] A morphism $\phi : Y \rightarrow X$ is
  a {\em Galois cover} if the induced field extension
  $\phi^* : \Fq(X) \hookrightarrow \Fq(Y)$ (see
  \S~\ref{subsec:morphisms}) is Galois. The {\em Galois group} of the
  cover is nothing but the Galois group of the extension.
\end{definition}

\begin{remark}
It is well known that, given a Galois cover $\phi : Y \rightarrow X$,
the Galois group $\Gamma$ acts on $Y$ and the orbits are the pre-images
of points of $X$.
\end{remark}

Tamo and Barg's construction can be generalized in terms of curve
morphisms. Indeed, the situation of
Example~\ref{ex:locality_with_group_action} can be interpreted as
follows. The elements $x_1, \ldots, x_n$ at which polynomials are
evaluated are now regarded as a sequence of rational points
$P_1, \dots, P_n$ of $\P^1$ that are from a disjoint union of orbits
under the action of the automorphism
$\sigma : (x:y) \longmapsto (\zeta x : y)$.  Next the polynomial $g$
induces a cyclic cover
$\phi : \P^1 \stackrel{g}{\longrightarrow} \P^1$ with Galois group
spanned by $\sigma$. The orbits can be regarded as fibres of
rational points of $\P^1$ that split totally.

Similarly to Reed--Solomon codes, for a given ground field, Tamo--Barg
approach permits to generate optimal codes with respect to
(\ref{eq:Singleton_local}) but the code length will be bounded from
above by the number of rational points of $\P^1$, i.e. by $q+1$. If
one wishes to create longer good locally recoverable codes, the construction can be
generalized as proposed in \cite{barg2017it}. 
Consider
\begin{itemize}
\item two curves $X, Y$ and a Galois cover
$\phi : Y \rightarrow X$ of degree $\ell + 1$ with Galois group
$\Gamma$;
\item rational points $Q_1, \dots, Q_{\frac{n}{(\ell + 1)}}$ of
$X$ which split completely in the cover $\phi$;
\item their pre-images by $\phi$, the points
  $P_{1,1}, \dots, P_{1,\ell+1}, \dots, P_{\frac{n}{\ell+1}, 1}, \dots,
  P_{\frac{n}{\ell +1}, \ell + 1} \in Y(\Fq)$ which are grouped by
  orbits of size $\ell + 1$ under the action of $\Gamma$.
  These orbits are the recovery sets.
\end{itemize}
Equivalently, the recovery sets are {\em fibres} of $\phi$, that is to
say, such a set is the pre-image set of a given totally split rational
point of $X$.  Let $x \in \Fq(Y)$ be a primitive element of the
extension $\Fq(Y)/\Fq(X)$ whose pole locus avoids\footnote{Here again,
  this avoiding condition can be relaxed thanks to
  Remark~\ref{UseWeakApprox}.} the points $P_{i,j}$ and $G$ be a
divisor on $X$.  Then, one can construct the code
\begin{equation}\label{eq:TBZcode}
  \CC \eqdef \left\{\left(f\left(P_{1,1}\right), \dots,
      f(P_{\frac{n}{\ell + 1}, \ell + 1})\right) ~\bigg|~
  f = \sum_{i=0}^{\ell - 1} (\phi^*f_i)\cdot x^i,\ \ f_i \in L (G)\right\},
\end{equation}
where $L(G) \subseteq \Fq(X)$.

Note that $\phi$ is constant on any recovery set and hence so are the
functions $\phi^* f_i$. Therefore, we have the following statement.

\begin{lemma}
The restriction of
$f = \sum (\phi^*f_i) x^i$ to a recovery set coincides with a
polynomial in $x$ with constant coefficients. Consequently, the
restriction of $\CC$ to a recovery set is a Reed--Solomon code of
length $\ell+1$ and dimension $\ell$.  
\end{lemma}

\begin{remark}
  Since $x$ is a primitive element of the extension $\Fq(Y)/\Fq(X)$,
  then, its restriction to a fibre is injective. Indeed, suppose that
  for two distinct points $R_1, R_2$ in a given fibre we have
  $x(R_1) = x(R_2)$, then since for any $f \in \Fq(X)$,
  $\phi^*f (R_1) = \phi^* f (R_2)$ and since $\Fq(Y)$ is generated as
  an algebra by $\phi^* \Fq(X)$ and $x$, this would entail that no function
  in $\Fq(Y)$ takes distinct values at $R_1$ and $R_2$, a contradiction.
\end{remark}

Similarly to the Reed--Solomon--like construction, these codes have
locality $\ell$. To estimate the other parameters, denote by
$\deg \phi $ the degree of the morphism, which is nothing but
the extension degree $[\Fq(Y) : \Fq(X)]$.  This degree is
also the degree of the pole locus divisor $(x)_{\infty}$ of the function $x$
(\cite[Lem.~2.2]{moreno1990book}).  

\begin{theorem}[{\cite[Th.~3.1]{barg2017it}}]
  The Barg--Tamo--\vladut{} code defined in \eqref{eq:TBZcode}
  has locality $\ell$ and parameters $[n, k, d]$ with
  \begin{align*}
    k & \geq \ell (\deg G + 1 - g);\\
    d & \geq n - (\ell - 1) \deg \phi - (\ell + 1) \deg G.
  \end{align*}
\end{theorem}

\begin{proof}
  The dimension is a consequence of the definition and of
  Riemann--Roch theorem. The proof of the locality is
  the very same as that of Tamo--Barg codes given in
  \S~\ref{subsec:TBcode}.
  For the minimum distance, observe that the
  code $\CC$ is contained in the code
  $\CL{Y}{\cP}{(\ell - 1)(x)_\infty + \phi^* G}$. Therefore, it
  suffices to bound from below the minimum distance of this code to
  get a lower bound for the minimum distance of $\CC$.  From
  (\ref{eq:pullback_div}), $\deg \phi^* G = (\ell + 1) \deg G$ and,
  from Theorem~\ref{thm:basics}, the code
  $\CL{Y}{\cP}{(\ell - 1)(x)_\infty + \phi^* G}$ has minimum distance
  at least $n - (\ell - 1)\deg (x)_{\infty} - (\ell + 1)\deg G$.
  Finally, from \cite[Lem.~2.2]{moreno1990book}, we get
  $\deg (x)_\infty = \deg \phi$, which concludes the proof.
\end{proof}

\begin{remark}
  The above proof is more or less that of \cite[Th.~3.1]{barg2017it}
  we chose to reproduce it here, in order describe the general
  strategy of evaluation of the parameters of such locally recoverable
  codes constructed from curves. Namely:
  \begin{itemize}
  \item the dimension is obtained by applying Riemann--Roch Theorem
    on $X$ together with an elementary count of monomials;
  \item the minimum distance is obtained by observing that the
    constructed LRC is contained in an actual algebraic geometry code
    to which Goppa bound (Theorem~\ref{thm:basics}) can be applied.
  \end{itemize}
\end{remark}

\begin{example}
  See \cite[\S~IV.A]{barg2017it} for examples of LRC from
  the Hermitian curve.
\end{example}

\subsection{Improvement: locally recoverable codes with higher local
  distance}\label{subsec:higher_local_distance}
Up to now, we introduced codes whose restriction to any recovery set
is nothing but a parity code, that is to say a code of minimum
distance $2$ which permits only to recover one symbol from the other
ones. One can expect more, such as being able to correct errors for
such a {\em local code}. Thus, one could look for codes whose
restrictions to recovery sets have a minimum distance larger than $2$.

\begin{definition}
  The {\em local distance} of a locally recoverable code $\CC$ 
  is defined as
  \[
    \min_{i} \min_{A(i)} \left\{\dmin
        (\rest{\CC}{A(i)})\right\},
  \]
  where $i \in \{1, \dots, n\}$ and $A(i)$ ranges over all the recovery
  sets of $\CC$ for $i$ (which may be non unique according to
  Remark~\ref{rem:on_recovery_sets}).
\end{definition}

The codes introduced in~\eqref{eq:TBZcode} have local distance $2$.
Actually, improving the local distance permits to reduce the amount of
requested symbols for the recovery of a given symbol as suggested
by the following statement.

\begin{lemma}
  Let $\CC \subseteq \Fq^n$ be a locally recoverable code with local
  distance $\rho$ and recovery sets of cardinality $\ell + 1$. For a
  codeword $\cv \in \CC$, any symbol $c_i$ of $\cv$ can be recovered
  from any $(\ell - \rho + 2)$--tuple of other symbols in the same
  recovery set.
\end{lemma}

\begin{proof}
  Let $A(i) \subseteq \{1, \dots, n\}$ be a recovery set for the
  position $i$. The restricted code $\rest{\CC}{A(i)}$ has length
  $\ell + 1$ and minimum distance $\rho$. Therefore, by definition of
  the minimum distance for any $I \subseteq A_i \setminus \{i\}$ with
  $|I| = (\ell + 1) - (\rho - 1) = \ell - \rho + 2$, the restriction map
  $\rest{\CC}{A(i)} \longrightarrow \rest{\CC}{I}$
  is injective.
\end{proof}

This can be done by reducing the degree in $x$ of the evaluated functions.
That is to say, considering a code
\begin{equation}\label{eq:TBZcode2}
  \CC' = \left\{\left(f(P_{1,1}, \dots, f(P_{\frac{n}{\ell +1}, \ell +
        1})\right) ~\bigg|~ f = \sum_{i=0}^{s-1} (\phi^*f_i) x^i,\quad f_i \in
     (L(G))\right\}
\end{equation}
for some integer $0 \leq s \leq \ell$.  Here again, the code
restricted to a recovery set is nothing but an
$[\ell + 1, s, \ell - s +2]$ Reed--Solomon code.
In such a code, a codeword is entirely determined by any $s$--tuple
of its entries.

\begin{theorem}
  The Tamo--Barg--\vladut{} code $\CC'$ defined in \eqref{eq:TBZcode2}
  has length $n$, locality $\ell$, local distance $\rho$
  and dimension $k$ and minimum distance $d$ satisfying
  \begin{align*}
    \rho & = \ell - s + 2\\
    k & \geq  s (\deg G + 1 - g) \\
    d & \geq n - (s-1) \deg (x) - (\ell + 1)\deg G.
  \end{align*}
\end{theorem}

\subsection{Fibre products of curves and the availability problem}
\label{subsec:availability}
Still motivated by distributed storage applications another parameter
called {\em availability} is of interest.

\begin{definition}
  The {\em availability} of a a locally recoverable code is the minimum over all position
  $i \in \{1, \dots, n\}$ of the number of recovery sets for $i$.
\end{definition}

Practically, a distributed storage system with a large availability
benefits from more flexibility in choosing the servers contacted for
recovering the contents of a given one.

The availability of an LRC is a positive integer and all the
previous constructions had availability $1$.  Constructing LRC with
multiple recovery sets is a natural problem.

\begin{itemize}
\item For Reed--Solomon--like LRC, this problem is discussed in
  \cite[\S~IV]{tamo2014it} by considering two distinct
  group actions on $\Fq$.  The cosets with respect of these group
  actions provide two partitions of the support yielding LRC
  with availability $2$.
\item In the curve case, constructions of LRC with availability
  $2$ from a curve $X$ together with two distinct morphisms from this
  curve to other curves $Y^{(1)}, Y^{(2)}$ is considered in
  \cite[\S~V]{barg2017it}, in \cite[\S~6]{barg2017agctc} and the case
  of availability $t \geq 2$ is treated in \cite{haymaker2018amc}.
\end{itemize}

The construction can be realized from the bottom using the notion of
{\em fibre product}.  Given three curves
$Y^{(1)}, Y^{(2)}$ and $X$ with morphisms
\[
  Y^{(1)} \stackrel{\phi_2}{\longrightarrow} X \qquad {\rm and} \qquad
    Y^{(2)} \stackrel{\phi_2}{\longrightarrow} X,
  \]
  then the fibre product $Y^{(1)} \times_X Y^{(2)}$ is defined as
\[
  (Y^{(1)} \times_X Y^{(2)})(\Fqbar) \eqdef \left\{(P_1,P_2) \in
    (Y^{(1)} \times Y^{(2)})(\Fqbar) ~|~ \phi_1(P_1) =
    \phi_2(P_2)\right\}.
\]
It comes with two canonical projections $\psi_1$ and $\psi_2$ onto $Y^{(1)}$
and $Y^{(2)}$ respectively:
\[
  \xymatrix @!0 @C=3.5pc @R=3.5pc {\relax
    & Y^{(1)} \times_X Y^{(2)} \ar[ld]_{\psi_1} \ar[rd]^{\psi_2} \ar[dd]^{\Phi}& \\
    \relax Y^{(1)} \ar[rd]_{\phi_1} & & Y^{(2)} \ar[ld]^{\phi_2}\\
     & X & 
  }
\]
Finally, we denote by $\Phi$ the morphism $\Phi = \phi_1 \circ \psi_1 = \phi_2 \circ \psi_2 : Y^{(1)} \times_X Y^{(2)} \longrightarrow X$.

The construction of LRC with availability $2$ can be done
as follows. Let $\ell_1, \ell_2$ be integers such that
$\deg \phi_1 = \ell_1 + 1$, $\deg \phi_2 = \ell_2 + 1$ (hence
$\deg \psi_1 = \ell_2 + 1$ and $\deg \psi_2 = \ell_1 + 1$) and
\begin{itemize}
\item $x_1, x_2$ be respective primitive elements of the extensions
  $\Fq(Y^{(1)})/\Fq(X)$ and $\Fq(Y^{(2)})/\Fq(X)$;
\item $G$ be a divisor on $X$;
\item $Q_1, \dots, Q_{s}$ be rational points of $X$ that are totally
  split in $\Phi$ and denote by $P_1,\dots , P_n$ their
  pre-images. 
\end{itemize}
Similarly to the previous cases, either we suppose that the supports
of $G$ and the points $Q_1, \dots, Q_s$ avoid the image under $\phi_1$ of
the pole locus of $x_1$ and the image under $\phi_2$ of the pole locus of
$x_2$, or we use Remark~\ref{UseWeakApprox}.

\begin{definition}
  With the above data, we define a locally recoverable code
  with availability $2$ as follows.
  \[
    \CC \eqdef \left\{ (f(P_1), \dots, f(P_n)) ~\bigg|~
      f =
      \sum_{i_1=0}^{\ell_1 - 1} \sum_{i_2=0}^{\ell_2 - 1}
      (\Phi^* h_{i_1 i_2})\psi_1^*(x_1)^{i_1}\psi_2^*(x_2)^{i_2},\
    h_{i_1 i_2} \in L(G)\right\}\! \cdot
  \]
\end{definition}
Let $i\in \{1, \dots, n\}$, the point $P_i$ is associated to two
recovery sets. Namely $\psi_1^{-1}(\{\psi_1(P_1)\})$ and
$\psi_2^{-1}(\{\psi_2(P_2)\})$ which have respective cardinalities
$\ell_2+1$ and $\ell_1 + 1$.

\begin{theorem}
  The code $\CC$ has availability $2$, with localities $(\ell_1, \ell_2)$.
  Its parameters $[n,k,d]$ satisfy
  \begin{align*}
    n &= s(\ell_1 + 1)(\ell_2 + 1)\\
    k &\geq (deg G + 1 - g_X)\ell_1 \ell_2 \\
    d &\geq n - \deg(G)(\ell_1 + 1)(\ell_2 + 1) \\
       &  \qquad-(\ell_1 - 1)(\ell_2 +1)\deg1(x_2)_{\infty}) -
        (\ell_2 - 1)(\ell_1 + 1)\deg((x_2)_{\infty}),
  \end{align*}
  where $g_X$ denotes the genus of $X$.
\end{theorem}

\begin{proof}
  The minimum distance comes from the fact that the code is a subcode
  of
  $\CL{Y^{(1)} \times_X Y^{(2)}}{\cP}{\Phi^*(G) + (\ell_1 -
    1)(\psi_1^*({(x_1)}_\infty)) + (\ell_2 - 1)(\psi_2^*{(x_2)}_{\infty})}$.
  The dimension is a direct consequence of the definition of the code.
  For further details, see for instance \cite[Th.~3.1]{haymaker2018amc}.
\end{proof}

\begin{example}
  See \cite[\S~5,6,7]{haymaker2018amc} for examples of LRC
  with availability $\geq 2$ from Giulietti--Korchmaros curves, Suzuki curves
  and Artin--Schreier curves.
\end{example}

\bibliographystyle{abbrv}

\begin{thebibliography}{100}

\bibitem{arbarello1985book}
E.~Arbarello, M.~Cornalba, P.~Griffiths, and J.~D. Harris.
\newblock {\em Geometry of algebraic curves I}, volume 267.
\newblock Springer-Verlag, {F}irst edition, 1985.

\bibitem{BCPRRR2019}
S.~Ballet, J.~Chaumine, J.~Pieltant, M.~Rambaud, H.~Randriambololona, and
  R.~Rolland.
\newblock On the tensor rank of multiplication in finite extensions of finite
  fields and related issues in algebraic geometry.
\newblock {\em Mosc. Math. J.}, to appear.

\bibitem{barelli2018phd}
{\'E}.~Barelli.
\newblock {\em {{\'E}tude de la s{\'e}curit{\'e} de certaines cl{\'e}s
  compactes pour le sch{\'e}ma de McEliece utilisant des codes
  g{\'e}om{\'e}triques}}.
\newblock PhD thesis, {Universit{\'e} Paris-Saclay}, Dec. 2018.

\bibitem{barg2017agctc}
A.~Barg, K.~Haymaker, E.~W. Howe, G.~L. Matthews, and A.~V{\'a}rilly-Alvarado.
\newblock Locally recoverable codes from algebraic curves and surfaces.
\newblock In E.~W. Howe, K.~E. Lauter, and J.~L. Walker, editors, {\em
  Algebraic Geometry for Coding Theory and Cryptography}, pages 95--127, Cham,
  2017. Springer International Publishing.

\bibitem{barg2017it}
A.~{Barg}, I.~{Tamo}, and S.~G. \vladut{}.
\newblock Locally recoverable codes on algebraic curves.
\newblock {\em IEEE Trans. Inform. Theory}, 63(8):4928--4939, 2017.

\bibitem{barker2019nist}
E.~Barker.
\newblock Recommendation for key management, 2019.
\newblock Draft NIST special publication 800--57 Part 1. Available online on
  \url{https://doi.org/10.6028/NIST.SP.800-57pt1r5-draft}.

\bibitem{BBGS2015}
A.~Bassa, P.~Beelen, A.~Garcia, and H.~Stichtenoth.
\newblock Towers of function fields over non-prime finite fields.
\newblock {\em Moscow Mathematical Journal}, 15(1):1--29, 2015.

\bibitem{becker2012ec}
A.~Becker, A.~Joux, A.~May, and A.~Meurer.
\newblock Decoding random binary linear codes in {$2^{n/20}$}: How {$1+1=0$}
  improves information set decoding.
\newblock In {\em Advances in Cryptology - EUROCRYPT~2012}, Lecture Notes in
  Comput. Sci. Springer, 2012.

\bibitem{beelen2007ffa}
P.~Beelen.
\newblock The order bound for general algebraic geometric codes.
\newblock {\em Finite Fields Appl.}, 13(3):665--680, 2007.

\bibitem{beelen08}
P.~Beelen and T.~H{\o}holdt.
\newblock The decoding of algebraic geometry codes.
\newblock In {\em Advances in algebraic geometry codes}, volume~5 of {\em Ser.
  Coding Theory Cryptol.}, pages 49--98. World Sci. Publ., Hackensack, NJ,
  2008.

\bibitem{berlekamp1973it}
E.~{Berlekamp}.
\newblock Goppa codes.
\newblock {\em IEEE Trans. Inform. Theory}, 19(5):590--592, 1973.

\bibitem{berlekamp1978it}
E.~Berlekamp, R.~McEliece, and H.~van Tilborg.
\newblock On the inherent intractability of certain coding problems.
\newblock {\em IEEE Trans. Inform. Theory}, 24(3):384--386, May 1978.

\bibitem{bernstein2019nist}
D.~J. Bernstein, T.~Chou, T.~Lange, I.~von Maurich, R.~Mizoczki,
  R.~Niederhagen, E.~Persichetti, C.~Peters, P.~Schwabe, N.~Sendrier,
  J.~Szefer, and W.~Wen.
\newblock Classic {M}c{E}liece: conservative code-based cryptography.
\newblock \url{https://classic.mceliece.org}, Mar. 2019.
\newblock Second round submission to the NIST post-quantum cryptography call.

\bibitem{Blackburn2003}
S.~R. Blackburn.
\newblock Frameproof codes.
\newblock {\em SIAM J. Discrete Math.}, 16(3):499--510, 2003.

\bibitem{BS1998}
D.~Boneh and J.~Shaw.
\newblock Collusion-secure fingerprinting for digital data.
\newblock {\em IEEE Trans. Inform. Theory}, 44(5):1897--1905, 1998.

\bibitem{bosma1997jsc}
W.~Bosma, J.~Cannon, and C.~Playoust.
\newblock The {M}agma algebra system. {I}. {T}he user language.
\newblock {\em J. Symbolic Comput.}, 24(3-4):235--265, 1997.
\newblock Computational algebra and number theory (London, 1993).

\bibitem{canteaut1998it}
A.~Canteaut and F.~Chabaud.
\newblock A new algorithm for finding minimum-weight words in a linear code:
  Application to {M}c{E}liece's cryptosystem and to narrow-sense {BCH} codes of
  length 511.
\newblock {\em IEEE Trans. Inform. Theory}, 44(1):367--378, 1998.

\bibitem{Cascudo2019}
I.~Cascudo.
\newblock On squares of cyclic codes.
\newblock {\em IEEE Trans. Inform. Theory}, 65(2):1034--1047, 2019.

\bibitem{CCCX2009}
I.~Cascudo, H.~Chen, R.~Cramer, and C.~Xing.
\newblock Asymptotically good linear secret sharing with strong multiplication
  over any finite field.
\newblock In S.~Halevi, editor, {\em Advances in Cryptology --- CRYPTO 2009},
  volume 5677 of {\em Lecture Notes in Computer Science}, pages 466--486.
  Springer-Verlag Berlin Heidelberg, 2009.

\bibitem{CCMZ2015}
I.~Cascudo, R.~Cramer, D.~Mirandola, and G.~Z{\'e}mor.
\newblock Squares of random linear codes.
\newblock {\em IEEE Trans. Inform. Theory}, 61(3):1159--1173, 2015.

\bibitem{CCX2011}
I.~Cascudo, R.~Cramer, and C.~Xing.
\newblock The torsion-limit for algebraic function fields and its application
  to arithmetic secret sharing.
\newblock In P.~Rogaway, editor, {\em Advances in Cryptology --- CRYPTO 2011},
  volume 6841 of {\em Lecture Notes in Computer Science}, pages 685--705.
  Springer-Verlag Berlin Heidelberg, 2011.

\bibitem{CCX2014}
I.~Cascudo, R.~Cramer, and C.~Xing.
\newblock Torsion limits and {R}iemann--{R}och systems for function fields and
  applications.
\newblock {\em IEEE Trans. Inform. Theory}, 60(7):3871--3888, 2014.

\bibitem{CGR2020}
I.~Cascudo, J.~S. Gundersen, and D.~Ruano.
\newblock Squares of matrix-product codes.
\newblock {\em Finite Fields Appl.}, 62, 2020.

\bibitem{CC2006}
H.~Chen and R.~Cramer.
\newblock Algebraic geometric secret sharing schemes and secure multi-party
  computations over small fields.
\newblock In C.~Dwork, editor, {\em Advances in Cryptology --- CRYPTO 2006},
  volume 4117 of {\em Lecture Notes in Computer Science}, pages 521--536.
  Springer-Verlag Berlin Heidelberg, 2006.

\bibitem{BCN1994}
B.~Chor, A.~Fiat, and M.~Naor.
\newblock Tracing traitors.
\newblock In Y.~G. Desmedt, editor, {\em Advances in Cryptology --- CRYPTO
  '94}, volume 839 of {\em Lecture Notes in Computer Science}, pages 257--270.
  Springer-Verlag Berlin Heidelberg, 1994.

\bibitem{ChCh1988}
D.~Chudnovsky and G.~Chudnovsky.
\newblock Algebraic complexities and algebraic curves over finite fields.
\newblock {\em Journal of Complexity}, 4:285--316, 1988.

\bibitem{CE2000}
G.~Cohen and S.~Encheva.
\newblock Efficient constructions of frameproof codes.
\newblock {\em Electronics Letters}, 36(22):1840--1842, 2000.

\bibitem{couvreur2012jalg}
A.~Couvreur.
\newblock The dual minimum distance of arbitrary-dimensional
  algebraic–geometric codes.
\newblock {\em J. Algebra}, 350(1):84--107, 2012.

\bibitem{couvreur2014ams}
A.~Couvreur.
\newblock {Codes and the Cartier Operator}.
\newblock {\em Proc. Amer. Math. Soc.}, 142:1983--1996, Mar 2014.

\bibitem{couvreur2014dcc}
A.~Couvreur, P.~Gaborit, V.~Gauthier{-}Uma{\~{n}}a, A.~Otmani, and J.-P.
  Tillich.
\newblock Distinguisher-based attacks on public-key cryptosystems using
  {Reed-Solomon} codes.
\newblock {\em Des. Codes Cryptogr.}, 73(2):641--666, 2014.

\bibitem{couvreur2020cm}
A.~Couvreur, P.~Lebacque, and M.~Perret.
\newblock Toward good families of codes from towers of surfaces, 2020.
\newblock To appear in AMS Contemp. Math. Available online : ArXiv:2002.02220.

\bibitem{couvreur2017it}
A.~Couvreur, I.~M{\'a}rquez-Corbella, and R.~Pellikaan.
\newblock Cryptanalysis of {M}c{E}liece cryptosystem based on algebraic
  geometry codes and their subcodes.
\newblock {\em IEEE Trans. Inform. Theory}, 63(8):5404--5418, Aug 2017.

\bibitem{couvreur2017it-2}
A.~Couvreur, A.~Otmani, and J.-P. Tillich.
\newblock Polynomial time attack on wild {M}c{E}liece over quadratic
  extensions.
\newblock {\em IEEE Trans. Inform. Theory}, 63(1):404--427, Jan 2017.

\bibitem{LacapelleBironDay2}
A.~Couvreur and C.~Ritzenthaler.
\newblock Oral presentation, {D}ay~2 of {A}{N}{R} {M}anta {F}irst {R}etreat
  ({L}acapelle-{B}iron, {F}rance), {A}pril 12, 2016.

\bibitem{CDM2000}
R.~Cramer, I.~Damg\r{a}rd, and U.~Maurer.
\newblock General secure multi-party computation from any linear secret-sharing
  scheme.
\newblock In B.~Preneel, editor, {\em Advances in Cryptology --- EUROCRYPT
  2000}, volume 1807 of {\em Lecture Notes in Computer Science}, pages
  316--334. Springer-Verlag Berlin Heidelberg, 2000.

\bibitem{CDN2015book}
R.~Cramer, I.~Damg\r{a}rd, and J.~Nielsen.
\newblock {\em Secure multiparty computation and secret sharing}.
\newblock Cambridge University Press, Cambridge, 2015.

\bibitem{dimakis2011procieee}
A.~G. Dimakis, K.~Ramchandran, Y.~Wu, and C.~Suh.
\newblock A survey on network codes for distributed storage.
\newblock {\em Proceedings of the IEEE}, 99(3):476--489, 2011.

\bibitem{dumer89pit}
I.~Dumer.
\newblock Two decoding algorithms for linear codes.
\newblock {\em Probl. Inf. Transm.}, 25(1):17--23, 1989.

\bibitem{duursma2011jpaa}
I.~Duursma, R.~Kirov, and S.~Park.
\newblock Distance bounds for algebraic geometric codes.
\newblock {\em J. Pure Appl. Algebra}, 215(8):1863--1878, 2011.

\bibitem{duursma2010ffa}
I.~Duursma and S.~Park.
\newblock Coset bounds for algebraic geometric codes.
\newblock {\em Finite Fields Appl.}, 16:36--55, 01 2010.

\bibitem{duursma1993it}
I.~M. {Duursma}.
\newblock Algebraic decoding using special divisors.
\newblock {\em IEEE Trans. Inform. Theory}, 39(2):694--698, 1993.

\bibitem{ehrhard1993it}
D.~{Ehrhard}.
\newblock Achieving the designed error capacity in decoding
  algebraic--geometric codes.
\newblock {\em IEEE Trans. Inform. Theory}, 39(3):743--751, 1993.

\bibitem{elkies1997allerton}
N.~D. Elkies.
\newblock Explicit modular towers.
\newblock In {\em Proceedings of the Thirty-Fifth Annual Allerton Conference on
  Communication, Control and Computing}, 1997.

\bibitem{elkies2001ecm}
N.~D. Elkies.
\newblock Explicit towers of {D}rinfeld modular curves.
\newblock In C.~Casacuberta, R.~M. Mir{\'o}-Roig, J.~Verdera, and
  S.~Xamb{\'o}-Descamps, editors, {\em European Congress of Mathematics}, pages
  189--198, Basel, 2001. Birkh{\"a}user Basel.

\bibitem{Elkies2003}
N.~D. Elkies.
\newblock Still better nonlinear codes from modular curves.
\newblock {\em Preprint}, 2003.

\bibitem{EF1983}
P.~Erd\"os and Z.~F\"uredi.
\newblock The greatest angle among $n$ points in the $d$-dimensional
  {E}uclidean space.
\newblock In C.~Berge, D.~Bresson, P.~Camion, J.-F. Maurras, and F.~Sterboul,
  editors, {\em Combinatorial mathematics}, volume~75 of {\em North-Holland
  Mathematics Studies}, pages 275--283. North-Holland, 1983.

\bibitem{faure2008iwcc}
C.~Faure and L.~Minder.
\newblock Cryptanalysis of the {McEliece} cryptosystem over hyperelliptic
  curves.
\newblock In {\em Proceedings of the eleventh International Workshop on
  Algebraic and Combinatorial Coding Theory}, pages 99--107, Pamporovo,
  Bulgaria, June 2008.

\bibitem{feng1993it}
G.-L. {Feng} and T.~R.~N. {Rao}.
\newblock Decoding algebraic-geometric codes up to the designed minimum
  distance.
\newblock {\em IEEE Trans. Inform. Theory}, 39(1):37--45, Jan 1993.

\bibitem{FZ1977}
C.~M. Fiduccia and Y.~Zalcstein.
\newblock Algebras having linear multiplicative complexities.
\newblock {\em Journal of the Association for Computing Machinery},
  24(2):311--331, 1977.

\bibitem{fulton1989book}
W.~Fulton.
\newblock {\em Algebraic curves}.
\newblock Advanced Book Classics. Addison-Wesley Publishing Company Advanced
  Book Program, Redwood City, CA, 1989.
\newblock An introduction to algebraic geometry, Notes written with the
  collaboration of Richard Weiss, Reprint of 1969 original.

\bibitem{GS1995}
A.~Garcia and H.~Stichtenoth.
\newblock A tower of {A}rtin-{S}chreier extensions of function fields attaining
  the {D}rinfeld-{V}ladut bound.
\newblock {\em Inventiones Mathematicae}, 121:211--222, 1995.

\bibitem{geil2013ffa}
O.~Geil, R.~Matsumoto, and D.~Ruano.
\newblock {F}eng-{R}ao decoding of primary codes.
\newblock {\em Finite Fields Appl.}, 23:35--52, 2013.

\bibitem{gopalan2012it}
P.~Gopalan, C.~Huang, S.~Huseyin, and S.~Yekhanin.
\newblock On the locality of codeword symbols.
\newblock {\em IEEE Trans. Inform. Theory}, 58(11):6925--6934, Nov 2012.

\bibitem{goppa1970ppi}
V.~D. Goppa.
\newblock A new class of linear error-correcting codes.
\newblock {\em Probl. Peredach. Inform.}, 6(3):24--30, 1970.
\newblock In Russian.

\bibitem{goppa1971ppi}
V.~D. Goppa.
\newblock Rational representation of codes and $({L}, g)$ codes.
\newblock {\em Probl. Peredach. Inform.}, 7(3):41--49, 1971.
\newblock In Russian.

\bibitem{goppa1977ppi}
V.~D. Goppa.
\newblock {Codes associated with divisors.}
\newblock {\em Probl. Peredach. Inform.}, 13(1):33--39, 1977.
\newblock In Russian.

\bibitem{goppa1981dansssr}
V.~D. Goppa.
\newblock Codes on algebraic curves.
\newblock {\em Dokl. Akad. Nauk SSSR}, 259(6):1289--1290, 1981.
\newblock In Russian.

\bibitem{goppa1982iansssr}
V.~D. Goppa.
\newblock Algeraico--geometric codes.
\newblock {\em Izv. Akad. Nauk SSSR Ser. Mat.}, 46(4):762--781, 1982.
\newblock In Russian.

\bibitem{grassl:codetables}
M.~Grassl.
\newblock Bounds on the minimum distance of linear codes and quantum codes.
\newblock Online available at \url{http://www.codetables.de}, 2007.
\newblock Accessed on 2020-07-05.

\bibitem{guneri2009jpaa}
C.~G{\"u}neri, H.~Stichtenoth, and I.~Ta\c{s}kin.
\newblock Further improvements on the designed minimum distance of algebraic
  geometry codes.
\newblock {\em J. Pure Appl. Algebra}, 213:87--97, 01 2009.

\bibitem{guruswami2005book}
V.~Guruswami.
\newblock {\em List Decoding of Error-Correcting Codes: Winning Thesis of the
  2002 ACM Doctoral Dissertation Competition (Lecture Notes in Computer
  Science)}.
\newblock Springer-Verlag, Berlin, Heidelberg, 2005.

\bibitem{guruswami99it}
V.~Guruswami and M.~Sudan.
\newblock Improved decoding of {R}eed--{S}olomon and {A}lgebraic--{G}eometry
  codes.
\newblock {\em IEEE Trans. Inform. Theory}, 45(6):1757--1767, 1999.

\bibitem{hartshorne1977book}
R.~Hartshorne.
\newblock {\em Algebraic geometry}, volume~52 of {\em Graduate Texts in
  Mathematics}.
\newblock Springer-Verlag, New York, 1977.

\bibitem{haymaker2018amc}
K.~Haymaker, B.~Malmskog, and G.~L. Matthews.
\newblock Locally recoverable codes with availability $t \geq 2$ from fiber
  products of curves.
\newblock {\em Adv. Math. Commun.}, 12(2):317--336, 2018.

\bibitem{hess2002jsc}
F.~Hess.
\newblock Computing {R}iemann--{R}och spaces in algebraic function fields and
  related topics.
\newblock {\em J. Symbolic Comput.}, 33(4):425--445, 2002.

\bibitem{hoeholdt95it}
T.~{H{\o}holdt} and R.~{Pellikaan}.
\newblock On the decoding of algebraic--geometric codes.
\newblock {\em IEEE Trans. Inform. Theory}, 41(6):1589--1614, Nov 1995.

\bibitem{huffman03book}
W.~C. Huffman and V.~Pless.
\newblock {\em Fundamentals of error-correcting codes}.
\newblock Cambridge University Press, Cambridge, 2003.

\bibitem{ihara1981jfsuTokyo}
Y.~Ihara.
\newblock Some remarks on the number of rational points of algebraic curves
  over finite fields.
\newblock {\em J. Fac. Sci. Univ. Tokyo Sect. IA Math.}, 28:721--724, 1981.

\bibitem{IKOS2009}
Y.~Ishai, E.~Kushilevitz, R.~Ostrovsky, and A.~Sahai.
\newblock Zero-knowledge proofs from secure multiparty computation.
\newblock {\em SIAM J. Comput.}, 39(3):1121--1152, 2009.

\bibitem{janwa1996dcc}
H.~Janwa and O.~Moreno.
\newblock {McEliece} public key cryptosystems using algebraic-geometric codes.
\newblock {\em Des. Codes Cryptogr.}, 8(3):293--307, 1996.

\bibitem{Justesen1972}
J.~Justesen.
\newblock A class of constructive asymptotically good algebraic codes.
\newblock {\em IEEE Trans. Inform. Theory}, 18(5):652--656, 1972.

\bibitem{justesen89it}
J.~{Justesen}, K.~J. {Larsen}, H.~E. {Jensen}, A.~{Havemose}, and
  T.~{H{\o}holdt}.
\newblock Construction and decoding of a class of algebraic geometry codes.
\newblock {\em IEEE Trans. Inform. Theory}, 35(4):811--821, July 1989.

\bibitem{KO1963}
A.~Karatsuba and Y.~Ofman.
\newblock Multiplication of multi-digit numbers on automata.
\newblock {\em Soviet Physics Doklady}, 7:595--596, 1963.

\bibitem{katsman1989cm}
G.~L. Katsman and M.~A. Tsfasman.
\newblock A remark on algebraic geometric codes.
\newblock In {\em Representation theory, group rings, and coding theory},
  volume~93 of {\em Contemp. Math.}, pages 197--199. Amer. Math. Soc.,
  Providence, RI, 1989.

\bibitem{kirfel1995it}
C.~{Kirfel} and R.~{Pellikaan}.
\newblock The minimum distance of codes in an array coming from telescopic
  semigroups.
\newblock {\em IEEE Trans. Inform. Theory}, 41(6):1720--1732, Nov 1995.

\bibitem{knapp80jalg}
W.~Knapp and P.~Schmid.
\newblock Codes with prescribed permutation group.
\newblock {\em J. Algebra}, 67(2):415--435, 1980.

\bibitem{koetter92acct}
R.~K{\"o}tter.
\newblock A unified description of an error locating procedure for linear
  codes.
\newblock In {\em Proceedings Algebraic and Combinatorial Coding Theory III},
  pages 113--117. Hermes, 1992.

\bibitem{lebrigand1988smf}
D.~Le~Brigand and J.-J. Riesler.
\newblock Algorithme de {B}rill--{N}oether et codes de {G}oppa.
\newblock {\em Bull. Soc. Math. France}, 116:231--253, 1988.

\bibitem{lee88eurocrypt}
P.~J. Lee and E.~F. Brickell.
\newblock An observation on the security of {McEliece}'s public-key
  cryptosystem.
\newblock In {\em Advances in Cryptology - EUROCRYPT'88}, volume 330 of {\em
  Lecture Notes in Comput. Sci.}, pages 275--280. Springer, 1988.

\bibitem{little2008chapter}
J.~B. Little.
\newblock Algebraic geometry codes from higher dimensional varieties.
\newblock In {\em Advances in algebraic geometry codes}, volume~5 of {\em Ser.
  Coding Theory Cryptol.}, pages 257--293. World Sci. Publ., Hackensack, NJ,
  2008.

\bibitem{lundell2006jpaa}
B.~Lundell and J.~McCullough.
\newblock A generalized floor bound for the minimum distance of geometric
  {G}oppa codes.
\newblock {\em J. Pure Appl. Algebra}, 207:155--164, 09 2006.

\bibitem{sloane1977book}
F.~J. MacWilliams and N.~J.~A. Sloane.
\newblock {\em The theory of error-correcting codes. {I}}.
\newblock North-Holland Publishing Co., Amsterdam, 1977.
\newblock North-Holland Mathematical Library, Vol. 16.

\bibitem{maharaj2005jpaa}
H.~Maharaj, G.~L.~Matthews, and I.~Pirsic.
\newblock {R}iemann--{R}och spaces of the {H}ermitian function field with
  applications to algebraic geometry codes and low-discrepancy sequences.
\newblock {\em J. Pure Appl. Algebra}, 195:261--280, 02 2005.

\bibitem{MMP2011}
I.~M\'arquez-Corbella, E.~Mart{\'\i}nez-Moro, and R.~Pellikaan.
\newblock Evaluation of public-key cryptosystems based on algebraic geometry
  codes.
\newblock In J.~Borges and M.~Villanueva, editors, {\em Proceedings of the
  Third International Castle Meeting on Coding Theory and Applications}, pages
  199--204, 2011.

\bibitem{may2011ac}
A.~May, A.~Meurer, and E.~Thomae.
\newblock Decoding random linear codes in {$O(2^{0.054n})$}.
\newblock In D.~H. Lee and X.~Wang, editors, {\em Advances in Cryptology -
  ASIACRYPT~2011}, volume 7073 of {\em Lecture Notes in Comput. Sci.}, pages
  107--124. Springer, 2011.

\bibitem{may2015ec}
A.~May and I.~Ozerov.
\newblock On computing nearest neighbors with applications to decoding of
  binary linear codes.
\newblock In E.~Oswald and M.~Fischlin, editors, {\em Advances in Cryptology -
  EUROCRYPT~2015}, volume 9056 of {\em Lecture Notes in Comput. Sci.}, pages
  203--228. Springer, 2015.

\bibitem{mceliece1978dsn}
R.~J. McEliece.
\newblock {\em A Public-Key System Based on Algebraic Coding Theory}, pages
  114--116.
\newblock Jet Propulsion Lab, 1978.
\newblock DSN Progress Report 44.

\bibitem{minder2007phd}
L.~Minder.
\newblock {\em Cryptography based on error correcting codes}.
\newblock PhD thesis, {\'E}cole Polytechnique F{\'e}d{\'e}rale de Lausanne,
  2007.

\bibitem{MZ2015}
D.~Mirandola and G.~Z{\'e}mor.
\newblock Critical pairs for the product {S}ingleton bound.
\newblock {\em IEEE Trans. Inform. Theory}, 61(9):4928--4937, 2015.

\bibitem{moreno1990book}
C.~J. Moreno.
\newblock {\em Algebraic curves over finite fields}.
\newblock Cambridge tracts in mathematics. Cambridge University Press,
  Cambridge, 1990.

\bibitem{mumford1970cime}
D.~Mumford.
\newblock Varieties defined by quadratic equations.
\newblock In {\em Questions on algebraic varieties, C.I.M.E., III Ciclo,
  Varenna, 1969}, pages 29--100. Edizioni Cremonese, Rome, 1970.

\bibitem{ng2019dcc}
S.-L. Ng and M.~Paterson.
\newblock Functional repair codes: a view from projective geometry.
\newblock {\em Des. Codes Cryptogr.}, 87:2701--2722, 2019.

\bibitem{niederreiter1986pcit}
H.~Niederreiter.
\newblock Knapsack-type cryptosystems and algebraic coding theory.
\newblock {\em Problems of Control and Information Theory}, 15(2):159--166,
  1986.

\bibitem{NO2004}
H.~Niederreiter and F.~\"Ozbudak.
\newblock Constructive asymptotic codes with an improvement on the
  tsfasman-\vladut-zink and xing bounds.
\newblock In K.~Feng, H.~Niederreiter, and C.~Xing, editors, {\em Coding,
  Cryptography and Combinatorics}, volume~23 of {\em Progress in Computer
  Science and Applied Logic}, pages 259--275. Birkh\"auser Verlag Basel, 2004.

\bibitem{niederreiter2001book}
H.~Niederreiter and C.~Xing.
\newblock {\em Rational points on curves over finite fields: theory and
  applications}, volume 285 of {\em London Mathematical Society Lecture Note
  Series}.
\newblock Cambridge University Press, Cambridge, 2001.

\bibitem{sullivan1998decoding}
M.~E. {O'Sullivan}.
\newblock Decoding of codes on surfaces.
\newblock In {\em 1998 Information Theory Workshop (Cat. No.98EX131)}, pages
  33--34, 1998.

\bibitem{papailiopoulos2014it}
D.~{Papailiopoulos} and A.~{Dimakis}.
\newblock Locally repairable codes.
\newblock {\em IEEE Trans. Inform. Theory}, 60(10):5843--5855, 2014.

\bibitem{patterson1975it}
N.~{Patterson}.
\newblock The algebraic decoding of {G}oppa codes.
\newblock {\em IEEE Trans. Inform. Theory}, 21(2):203--207, 1975.

\bibitem{pellikaan88}
R.~Pellikaan.
\newblock On decoding linear codes by {E}rror {C}orrecting {P}airs.
\newblock Preprint Technical University Eindhoven, 1988.

\bibitem{pellikaan1989it}
R.~{Pellikaan}.
\newblock On a decoding algorithm for codes on maximal curves.
\newblock {\em IEEE Trans. Inform. Theory}, 35(6):1228--1232, 1989.

\bibitem{pellikaan92dm}
R.~Pellikaan.
\newblock On decoding by error location and dependent sets of error positions.
\newblock {\em Discrete Math.}, 106--107:369--381, 1992.

\bibitem{pellikaan1993eurocode}
R.~Pellikaan.
\newblock On the efficient decoding of algebraic-geometric codes.
\newblock In {\em Eurocode'92 ({U}dine, 1992)}, volume 339 of {\em CISM Courses
  and Lectures}, pages 231--253. Springer, Vienna, 1993.

\bibitem{prange1962ireit}
E.~Prange.
\newblock The use of information sets in decoding cyclic codes.
\newblock {\em {IRE} Transactions on Information Theory}, 8(5):5--9, 1962.

\bibitem{HR-JCompl2012}
H.~Randriambololona.
\newblock Bilinear complexity of algebras and the {C}hudnovsky-{C}hudnovsky
  interpolation method.
\newblock {\em Journal of Complexity}, 28(4):489--517, 2012.

\bibitem{HR-IJM2013}
H.~Randriambololona.
\newblock $(2,1)$-separating systems beyond the probabilistic bound.
\newblock {\em Israel J. Math.}, 195(1):171--186, 2013.

\bibitem{HR-agis2013}
H.~Randriambololona.
\newblock Asymptotically good binary linear codes with asymptotically good
  self-intersection spans.
\newblock {\em IEEE Trans. Inform. Theory}, 59(5):3038--3045, 2013.

\bibitem{HR-Singleton2013}
H.~Randriambololona.
\newblock An upper bound of {S}ingleton type for componentwise products of
  linear codes.
\newblock {\em IEEE Trans. Inform. Theory}, 59(12):7936--7939, 2013.

\bibitem{HR-AGCT14}
H.~Randriambololona.
\newblock On products and powers of linear codes under componentwise
  multiplication.
\newblock In S.~Ballet, M.~Perret, and A.~Zaytsev, editors, {\em Algorithmic
  Arithmetic, Geometry, and Coding Theory}, volume 637 of {\em Contemporary
  Mathematics}, pages 3--77. American Mathematical Society, 2014.

\bibitem{HR-ISIT2015}
H.~Randriambololona.
\newblock Linear independence of rank $1$ matrices and the dimension of
  $*$-products of codes.
\newblock In {\em Proceedings of 2015 IEEE International Symposium on
  Information Theory}. IEEE Information Theory Society, 2015.

\bibitem{reed60siam}
I.~S. Reed and G.~Solomon.
\newblock Polynomial codes over certain finite fields.
\newblock {\em Journal of the society for industrial and applied mathematics},
  8(2):300--304, 1960.

\bibitem{roth06book}
R.~M. Roth.
\newblock {\em Introduction to Coding Theory}.
\newblock Cambridge University Press, New York, NY, USA, 2006.

\bibitem{roth85it}
R.~M. Roth and G.~Seroussi.
\newblock On generator matrices of {MDS} codes.
\newblock {\em IEEE Trans. Inform. Theory}, 31(6):826--830, 1985.

\bibitem{sendrier1994ec}
N.~Sendrier.
\newblock On the structure of a randomly permuted concatenated code.
\newblock In {\em EUROCODE'94}, pages 169--173, 1994.

\bibitem{Serre1983}
J.-P. Serre.
\newblock Nombre de points des courbes alg\'ebriques sur $\mathbf{F}_q$.
\newblock {\em S\'eminaire de th\'eorie des nombres de {B}ordeaux},
  12:22--01--22--08, 1982-1983.

\bibitem{shafarevich1994book}
I.~R. Shafarevich.
\newblock {\em Basic algebraic geometry. 1}.
\newblock Springer-Verlag, Berlin, second edition, 1994.

\bibitem{shokrollahi99it}
M.~A. {Shokrollahi} and H.~{Wasserman}.
\newblock List decoding of algebraic-geometric codes.
\newblock {\em IEEE Trans. Inform. Theory}, 45(2):432--437, March 1999.

\bibitem{shor94focs}
P.~W. Shor.
\newblock Algorithms for quantum computation: Discrete logarithms and
  factoring.
\newblock In S.~Goldwasser, editor, {\em FOCS}, pages 124--134, 1994.

\bibitem{STV1992}
I.~E. Shparlinski, M.~A. Tsfasman, and S.~G. \vladut{}.
\newblock Curves with many points and multiplication in finite fields.
\newblock In H.~Stichtenoth and M.~A. Tsfasman, editors, {\em Coding theory and
  algebraic geometry}, volume 1518 of {\em Lecture Notes in Mathematics}, pages
  145--169. Springer-Verlag, 1992.

\bibitem{sidelnikov1992dma}
V.~M. Sidelnikov and S.~Shestakov.
\newblock On the insecurity of cryptosystems based on generalized
  {Reed-Solomon} codes.
\newblock {\em Discrete Math. Appl.}, 1(4):439--444, 1992.

\bibitem{silverman2009book}
J.~Silverman.
\newblock {\em The arithmetic of elliptic curves}, volume 106.
\newblock Springer-Verlag New York, second edition, 2009.

\bibitem{skorobogatov90it}
A.~N. {Skorobogatov} and S.~G. {\vladut{}}.
\newblock On the decoding of algebraic-geometric codes.
\newblock {\em IEEE Trans. Inform. Theory}, 36(5):1051--1060, Sep. 1990.

\bibitem{stern1988cta}
J.~Stern.
\newblock A method for finding codewords of small weight.
\newblock In G.~D. Cohen and J.~Wolfmann, editors, {\em Coding Theory and
  Applications}, volume 388 of {\em Lecture Notes in Comput. Sci.}, pages
  106--113. Springer, 1988.

\bibitem{stichtenoth2009book}
H.~Stichtenoth.
\newblock {\em Algebraic function fields and codes}, volume 254 of {\em
  Graduate Texts in Mathematics}.
\newblock Springer-Verlag, Berlin, second edition, 2009.

\bibitem{SX2005}
H.~Stichtenoth and C.~Xing.
\newblock Excellent nonlinear codes from algebraic function fields.
\newblock {\em IEEE Trans. Inform. Theory}, 51(11):4044--4046, 2005.

\bibitem{SW1998}
D.~R. Stinson and R.~Wei.
\newblock Combinatorial properties and constructions of traceability schemes
  and frameproof codes.
\newblock {\em SIAM J. Discrete Math.}, 11(1):41--53, 1998.

\bibitem{Strassen1969}
V.~Strassen.
\newblock Gaussian elimination is not optimal.
\newblock {\em Numerische Mathematik}, 13:354--356, 1969.

\bibitem{sudan97jcomp}
M.~Sudan.
\newblock Decoding of {Reed--Solomon} {C}odes beyond the {E}rror-{C}orrection
  {B}ound.
\newblock {\em J. Complexity}, 13(1):180--193, 1997.

\bibitem{sugiyama1976it}
Y.~Sugiyama, M.~Kasahara, S.~Hirasawa, and T.~Namekawa.
\newblock Further results on {Goppa} codes and their applications to
  constructing efficient binary codes.
\newblock {\em IEEE Trans. Inform. Theory}, 22:518--526, 1976.

\bibitem{tamo2014it}
I.~{Tamo} and A.~{Barg}.
\newblock A family of optimal locally recoverable codes.
\newblock {\em IEEE Trans. Inform. Theory}, 60(8):4661--4676, 2014.

\bibitem{tao06book}
T.~Tao and V.~H. Vu.
\newblock {\em Additive combinatorics}, volume 105 of {\em Cambridge studies in
  advanced mathematics}.
\newblock Cambridge University Press, 2006.

\bibitem{tsfasman2007book}
M.~A. Tsfasman, S.~Vl{\u{a}}du{\c{t}}, and D.~Nogin.
\newblock {\em Algebraic geometric codes: basic notions}, volume 139 of {\em
  Mathematical Surveys and Monographs}.
\newblock American Mathematical Society, Providence, RI, 2007.

\bibitem{tsfasman1982mn}
M.~A. Tsfasman, S.~G. Vl{\u{a}}du{\c{t}}, and T.~Zink.
\newblock Modular curves, {S}himura curves, and {G}oppa codes, better than
  {V}arshamov-{G}ilbert bound.
\newblock {\em Math. Nachr.}, 109:21--28, 1982.

\bibitem{vanLint99book}
J.~H. van Lint.
\newblock {\em Introduction to coding theory}.
\newblock Graduate texts in mathematics. Springer, 3rd edition, 1999.

\bibitem{vanLint1986it}
J.~H. {van Lint} and R.~M. {Wilson}.
\newblock On the minimum distance of cyclic codes.
\newblock {\em IEEE Trans. Inform. Theory}, 32(1):23--40, January 1986.

\bibitem{vladut1990it}
S.~G. \vladut.
\newblock On the decoding of algebraic--geometric codes over $\mathbb{F}_q$ for
  $q \geq 16$.
\newblock {\em IEEE Trans. Inform. Theory}, 36(6):1461--1463, 1990.

\bibitem{vladut1983faa}
S.~G. \vladut{} and V.~G. Drinfeld.
\newblock Number of points of an algebraic curve.
\newblock {\em Funct. Anal. Appl.}, 17:53--54, 1983.

\bibitem{Winograd1977}
S.~Winograd.
\newblock Some bilinear forms whose multiplicative complexity depends on the
  field of constants.
\newblock {\em Mathematical Systems Theory}, 10:169--180, 1977.

\bibitem{wirtz1988it}
M.~Wirtz.
\newblock On the parameters of {G}oppa codes.
\newblock {\em IEEE Trans. Inform. Theory}, 34(5, part 2):1341--1343, 1988.
\newblock Coding techniques and coding theory.

\bibitem{Xing2002}
C.~Xing.
\newblock Asymptotic bounds on frameproof codes.
\newblock {\em IEEE Trans. Inform. Theory}, 48(11):2991--2995, 2002.

\bibitem{Xing2003}
C.~Xing.
\newblock Nonlinear codes from algebraic curves improving the
  {T}sfasman-\vladut-{Z}ink bound.
\newblock {\em IEEE Trans. Inform. Theory}, 49(7):1653--1657, 2003.

\end{thebibliography}

\end{document}